\documentclass[oneside, 11pt, reqno]{amsart}
\usepackage[left=1in, right=1in, top=1in, bottom=1in]{geometry}

\usepackage[usenames,dvipsnames]{xcolor}
\usepackage{amsmath}
\DeclareMathOperator*{\argmax}{argmax}
\usepackage{graphicx,enumitem}
\usepackage{mathtools}
\usepackage{extarrows} 
\usepackage{verbatim}

\definecolor{myblue}{RGB}{94, 129, 181}
\definecolor{myorange}{RGB}{225, 156, 36}
\definecolor{mygreen}{RGB}{143, 176, 50}
\definecolor{myred}{RGB}{235, 98, 53}

\usepackage{pdflscape}

\usepackage{adjustbox}

\usepackage{colortbl} 

\usepackage{amssymb} 
\usepackage{soul}
\usepackage{cancel}
\allowdisplaybreaks	

\usepackage[margin=0.7cm]{caption}
\usepackage{multirow}

\DeclareMathOperator*\li{\underline{lim}}
\DeclareMathOperator*\ls{\overline{lim}}
\DeclareMathOperator*{\esssup}{ess\,sup}

\providecommand{\R}{} \renewcommand{\R}{{\mathbb R}}
 
\newcommand{\N}{{\mathbb N}}
\newcommand{\PP}{{\mathbb P}}
\newcommand{\EE}{{\mathbb E}}
\newcommand{\FF}{{\mathcal F}}
\newcommand{\Chi}{{\mathcal X}}

\newcommand{\Zitkovic}[1]{{\v Z}itkovi\'c}
\newcommand{\Sirbu}[1]{S\^\i rbu}

\newcommand\tV{{\tilde{V}}}
\newcommand\ue{\underline{\epsilon}}
\newcommand\be{\overline{\epsilon}}
\newcommand\uy{\underline{y}}
\newcommand\by{\overline{y}}
\newcommand\uz{\underline{z}}
\newcommand\bz{\overline{z}}

\numberwithin{equation}{section}
\theoremstyle{plain}                
\newtheorem{theorem}{Theorem}[section]
\newtheorem{lemma}[theorem]{Lemma}
\newtheorem{proposition}[theorem]{Proposition}

\theoremstyle{definition}           

\newtheorem{assumption}[theorem]{Assumption}
\newtheorem{notation}[theorem]{Notation}
\theoremstyle{remark}
\newtheorem{remark}[theorem]{Remark}

\newcommand{\thmref}[1]{Theorem~\ref{#1}}
\newcommand{\proref}[1]{Proposition~\ref{#1}}

\newcommand{\lemref}[1]{Lemma~\ref{#1}}

\newcommand{\assref}[1]{Assumption~\ref{#1}}
\usepackage{graphicx} 

\usepackage[toc,title,page]{appendix}

\usepackage{fancyhdr}	
\begin{document}

\title{Unified Asymptotics for Investment under Illiquidity: Transaction Costs and Search Frictions}

\author{Tae Ung Gang$^{\ast}$}
\thanks{$^{\ast}$Stochastic Analysis and Application Research Center, Korea Advanced Institute of Science and Technology (gangtaeung@kaist.ac.kr).}
\author{Jin Hyuk Choi$^{\ast\ast}$}\thanks{$^{\ast\ast}$Department of Mathematical Sciences, Ulsan National Institute of Science and Technology (jchoi@unist.ac.kr).\\
{\bf Conflict of interest statement}: We declare that we have no known competing financial interests or personal relationships that could have appeared to influence the work reported in this paper.\\
{\bf Data availability statement}: Data sharing not applicable to this article as no datasets were generated or analysed during the current study.\\
{\bf Funding}: This work was supported by the National Research Foundation of Korea (NRF) grant funded by the Korea government (MSIT) (No. 002086221G0001278, No. 2019R1A5A1028324 and No. RS-2023-00237770).}

\maketitle

\begin{abstract}

This paper investigates the optimal investment problem in a market with two types of illiquidity: transaction costs and search frictions. Extending the framework established by \cite{GC23}, we analyze a power-utility maximization problem where an investor encounters proportional transaction costs and trades only when a Poisson process triggers trading opportunities. We show that the optimal trading strategy is described by a no-trade region. We introduce a novel asymptotic framework applicable when both transaction costs and search frictions are small. Using this framework, we derive explicit asymptotics for the no-trade region and the value function along a specific parametric curve. This approach unifies existing asymptotic results for models dealing exclusively with either transaction costs or search frictions. 
\end{abstract}

\bigskip


\bigskip

\section{Introduction}

Understanding the impact of illiquidity on optimal investment is one of the key topics in mathematical finance. Illiquidity arises from various factors, such as exogenous transaction costs, search frictions (difficulty in finding a trading counterparty), and price impacts.\footnote{\cite{Kyl85, bac92, asym, CHOI201922} have explored how asymmetric information affects price impact and optimal trading strategy in equilibrium, while \cite{almgren2001optimal, gatheral2011optimal, predoiu2011optimal, robert2012measuring} have examined optimal order execution problems with given price impacts.}
Building on the idea of \cite{GC23}, this paper investigates an optimal investment problem in a market with two types of illiquidity:  transaction costs and search frictions.

Assuming perfect liquidity, where assets can be traded at any time without transaction costs, Merton's seminal works \cite{Mer69, Mer71} formulate the optimal investment problem using geometric Brownian motion for a risky asset price and a CRRA (constant relative risk aversion) investor, showing that the optimal strategy is to maintain a constant fraction of wealth in the risky asset. Subsequent research has extended this framework to more general stock price processes and utility functions, deriving broader optimal investment strategies.

The perfect liquidity assumption can be relaxed by incorporating transaction costs, such as order processing fees or transaction taxes, which contribute to market illiquidity and have been extensively studied. \cite{MagCon76, DavNor90, ShrSon94} examine the Merton model with proportional transaction costs, demonstrating that the optimal strategy is to keep the investment within a ``no-trade region." The boundaries of this region are determined by the free-boundaries of the HJB (Hamilton-Jacobi-Bellman) equation. Models with transaction costs and multiple risky assets have been investigated (e.g., \cite{Aki96, Liu04, Mut06, CheDai13} for costs on all assets and \cite{Dai11, Guasoni, Hobson16, Choi2020} for costs on only one asset). More general stock price processes have been considered within the framework of optimal investment with transaction costs (e.g., \cite{czichowsky2014transaction, czichowsky2016duality, bayraktar2020extended}). Additionally, \cite{GarPed2016, MorMuhSon2017, GuaWeb2017, EkrMuh2019} investigate models with quadratic transaction costs. 

Search frictions, or difficulty in finding a trading counterparty, are another source of market illiquidity. \cite{Andrew14} presents frequency of trading in various markets, showing that many asset classes are illiquid, with their total sizes rivaling that of the public equity market. An intuitive way to model search frictions is by restricting trade times. For example, \cite{Rogers:2001aa} considers an investor who can change portfolios only at fixed intervals, while \cite{Rogers2002, matsumoto2006, Andrew14} assume that an illiquid asset can only be traded when randomly occurring opportunities arise, modeled by a Poisson process. \cite{2008Pham, 2011Pham} add the assumption that the asset price is observed only at these trade times. \cite{2014Pham} further complicates the model by incorporating random intensity of trade times, regime-switching, and liquidity shocks. \cite{Dai2016} considers trading at deterministic intervals with proportional transaction costs. \cite{hugonnier2015capital} develops a model for investment, financing, and cash management decisions, incorporating financing opportunities modeled as a Poisson process.

Due to the lack of explicit solutions to the HJB equations, asymptotic analysis has been employed for small transaction costs or small search frictions. For a small transaction cost parameter $\epsilon\ll 1$, in various models with proportional transaction costs only (e.g., \cite{JanShr04, GerMuhSch11, Cho11, BicShr11, Pos15, Guasoni, Choi2020}), the first correction terms of the no-trade boundaries are of the order of $\epsilon^{\frac{1}{3}}$, and the first correction term of the value function is of the order of $\epsilon^{\frac{2}{3}}$. In \cite{matsumoto2006,matsumotopower}, the parameter $\lambda$ is the intensity of the Poisson process, and search frictions can be represented by $\tfrac{1}{\lambda}$. For small search frictions $\tfrac{1}{\lambda}\ll1$, the first correction terms of the optimal trading strategy and the value function are of the order of $\tfrac{1}{\lambda}$.

Merging the aforementioned frameworks, \cite{GC23} studies log-utility maximization of the terminal wealth in a model with both transaction costs and search frictions. As in the models with transaction costs only, the optimal trading strategy in \cite{GC23} is characterized by a no-trade region. In \cite{GC23}, for a small transaction cost parameter $\epsilon$ (with fixed $\lambda$), the first correction terms of the no-trade boundaries and the value function are of the order of $\epsilon$, instead of $\epsilon^{\frac{1}{3}}$ or $\epsilon^{\frac{2}{3}}$ in the models with transaction costs only. The asymptotic results imply that the effects of the transaction costs are more pronounced in the market with fewer search frictions.

In this paper, we analyze the power-utility maximization problem with both transaction costs and search frictions. The model setup is the same as that of \cite{GC23} except that we consider power-utility instead of log-utility. In our model, proportional transaction costs (with parameter $\epsilon$) are imposed on an investor, and the investor's trading opportunities arise only when a Poisson process (with intensity $\lambda$) jumps. The investor's objective is to maximize the expected utility of wealth at the terminal time $T>0$.
As in other models with proportional transaction costs, the optimal trading strategy in our model is characterized by a no-trade region: there are functions $\uy,\by:[0,T)\to [0,1]$ such that the investor tries to keep the fraction of wealth invested in the risky asset within the interval $[\uy(t),\by(t)]$ whenever trading opportunities arise.

The main contribution of this paper is the establishment of a novel framework for asymptotics applicable in scenarios where both transaction costs and search frictions are small, i.e., $\epsilon\ll 1$ and $\tfrac{1}{\lambda}\ll 1$. We focus on the asymptotics of the no-trade region and the value function. The results in \cite{GC23} imply that the multivariable limits as $(\epsilon,\lambda)\to (0,\infty)$ do not exist; the resulting values depend on the order of taking these limits (see discussion around \eqref{heuristic1} in Section 5). To address this issue, we compare the asymptotics in \cite{GC23} with those in the benchmark cases of transaction costs only ($\lambda=\infty$) and search frictions only ($\epsilon=0$), leading us to conjecture that a specific scaling relation $\lambda=c \, \epsilon^{-\frac{2}{3}}$ for $c>0$ is relevant to consider (for details, see discussion for \eqref{curve_motive} in Section 5). Our findings confirm that along the parametric curve $\lambda=c \, \epsilon^{-\frac{2}{3}}$, the first correction terms of the no-trade boundaries and the value function are of the order of $\epsilon^{\frac{1}{3}}$ and $\epsilon^{\frac{2}{3}}$, respectively. 

Our framework for finding asymptotics along the parametric curve $\lambda=c \, \epsilon^{-\frac{2}{3}}$ offers two notable benefits when dealing with small $\epsilon$ and large $\lambda$. First, the coefficients of the correction terms in our asymptotics are explicit in terms of the model parameters. In contrast, the coefficients in the asymptotics in \cite{GC23} are expressed in terms of solutions to partial differential equations, making them not explicit.\footnote{This lack of explicit expression is the main weakness of the asymptotic results in \cite{GC23}. Additionally, the asymptotics in \cite{GC23} are only for small $\epsilon$, with fixed $\lambda$.} Therefore, given model parameters, including $\epsilon$ and $\lambda$, one can compute the auxiliary parameter $c=\lambda \epsilon^{\frac{2}{3}}$ and use the explicit expressions in \thmref{Joint_limit_of_Wnt_and_Vd} to estimate the optimal trading strategy and value. 

Second, our framework using $\lambda=c \, \epsilon^{-\frac{2}{3}}$ unifies the existing asymptotic results for seemingly different benchmark models with only transaction costs and only search frictions. Indeed, \thmref{Joint_limit_of_Wnt_and_Vd} bridges the benchmark asymptotics, where the case $c\to\infty$ corresponds to the asymptotics with only transaction costs and the case $c\to 0$ corresponds to the asymptotics with only search frictions (see discussion around \eqref{asymptotic_connection} in Section 5).

Our proof of the asymptotic analysis involves various estimations. One of the main difficulties in the analysis is the rigorous treatment of subtle limiting behaviors that do not appear in the benchmark models with only transaction costs or only search frictions (see discussion after \thmref{Joint_limit_of_Wnt_and_Vd}).

The remainder of the paper is organized as follows. Section 2 describes the model. In Section 3, we provide the verification argument and some properties of the value function. Section 4 characterizes the optimal trading strategy in terms of the no-trade region and presents properties of its boundaries. In Section 5, we motivate the relation $\lambda=c \, \epsilon^{-\frac{2}{3}}$ and provide asymptotic results. Section 6 numerically examines the impact of search frictions on trading.
Section 7 is devoted to the proof of the results in Section 5. Section 8 summarizes the paper. Proofs of  technical lemmas can be found in Appendix.

\section{The Model}

The model setup is identical to that described in \cite{GC23}, except for the utility function. 
Consider a filtered probability space $(\Omega, \FF, (\FF_t)_{t \geq 0}, \PP)$ satisfying the usual conditions. Under the filtration, let $(B_{t})_{t \geq 0}$ be a standard Brownian motion and $(P_{t})_{t \geq 0}$ be a Poisson process with constant intensity $\lambda > 0$. Then $(B_{t})_{t \geq 0}$ and $(P_{t})_{t \geq 0}$ are independent as the quadratic covariation of the two Levy processes is zero. 

We consider a financial market consisting of a constant saving account (zero interest rate) and a stock with its price process $( S_{t} )_{t \geq 0}$ defined by the following stochastic differential equation (SDE):
\begin{align*}
  d S_{t} & = S_{t} ( \mu d t + \sigma d B_{t} ), 
\end{align*}
where $\mu$, $\sigma$ and $S_{0}$ are constants and $\sigma$ and $S_{0}$ are strictly positive.

We assume that the market has two types of illiquidity. 
\begin{itemize}
\item {\bf Proportional Transaction Costs}: These costs are imposed on an investor when purchasing and selling stocks. 
There are two constants $\be \in (0, \infty)$ and $\ue \in (0, 1)$ such that the investor purchases one share of stock at the price of $(1 + \be) S_{t}$ and sells one share at the price of $(1 - \ue) S_{t}$ at time $t$, respectively.
\item {\bf Limited Trading Opportunities}: An investor's trading opportunity is available only when the Poisson process $(P_{t})_{t \geq 0}$ jumps. Hence, a larger $\lambda$ implies more frequent trading opportunities on average, resulting in fewer search frictions.
\end{itemize}

Let $W_{t}^{(0)}$ and $W_{t}^{(1)}$ be the amount of wealth in the saving account and stock at time $t \geq 0$, respectively. 
 Let $M_t$ represent the nominal change in portfolio position at time $t\geq 0$, which is the dollar amount the investor intends to transfer to the stock market.   
Then, $W_{t}^{(0)}$ and $W_{t}^{(1)}$ satisfy
\begin{equation}
\begin{split}\label{SDE_W}
  & W_{t}^{(0)} = w_{0}^{(0)} + \int_{0}^{t} \left( (1 - \ue)  M_{s}^{-} - (1 + \be) M_{s}^{+} \right) d P_{s}, \\
  & W_{t}^{(1)} = w_{0}^{(1)} + \int_{0}^{t} W_{s-}^{(1)}( \mu ds + \sigma dB_{s} ) + \int_{0}^{t} M_{s} d P_{s}, 
\end{split}  
\end{equation}
where the pair of nonnegative constants $( w_{0}^{(0)}, w_{0}^{(1)} )$ represents the initial position of the investor and we use notation $x^{\pm} = \max \{ \pm x, 0 \}$ for $x\in \R$. We assume that the initial total wealth is strictly positive, $w_0:=w_{0}^{(0)}+ w_{0}^{(1)}>0$.

The trading strategy $(M_{t})_{t \geq 0}$ is called \textit{admissible} if it is a predictable process and the corresponding total wealth process $W:=W^{(0)}+W^{(1)}$ remains nonnegative all the times. Since the rebalancing times are discrete, the condition $W_t\geq 0$ for all $t\geq 0$ is equivalent to $W_t^{(0)}\geq 0$ and $W_t^{(1)}\geq 0$ for all $t\geq 0$, imposing a constraint against short positions. Consequently, an admissible strategy $M$ satisfies
\begin{align} \label{Admissible_class}
  & - W_{t-}^{(1)} \leq M_{t} \leq \tfrac{W_{t-}^{(0)}}{1 + \be}, \quad t\geq 0.
\end{align}
The above inequalities and $w_0>0$ ensure that the corresponding total wealth process $W$ is strictly positive all the time.\footnote{
The strict positivity of $W_t$ is discussed in \cite{GC23}, footnote 7. For completeness, we include the explanation here.  
Let $\tau_n:=\inf\{ t\geq 0: P_t=n\}$ and $\tau_0=0$. 
For $\tau_n\leq t<\tau_{n+1}$, the dynamics \eqref{SDE_W} produce $W_t=W_{\tau_n}^{(0)} e^{r(t-\tau_n)} +W_{\tau_n}^{(1)} e^{(\mu-\frac{\sigma^2}{2})(t-\tau_n)+\sigma(B_t-B_{\tau_n})}$. If $W_{\tau_n-}>0, W_{\tau_n-}^{(0)}\geq 0, W_{\tau_n-}^{(1)}\geq 0$ and \eqref{Admissible_class} are satisfied, then $W_{\tau_n}>0, W_{\tau_n}^{(0)}\geq 0$ and $W_{\tau_n}^{(1)}\geq 0$ hold. By this way, we can inductively show that $W_{\tau_n}>0, W_{\tau_n}^{(0)}\geq 0$ and $W_{\tau_n}^{(1)}\geq 0$ for all $n$, almost surely. Now the expression of $W_t$ above implies that $W_t>0$ for all $t\geq 0$, almost surely.}

For an admissible strategy $M$ and the corresponding solutions $W^{(0)}$ and $W^{(1)}$ of the SDEs in \eqref{SDE_W}, let $X_{t} := W_{t}^{(1)}/W_{t}$ be the fraction of the total wealth invested in the stock market at time $t$. Then, the inequalities in \eqref{Admissible_class} imply $0\leq X_t \leq 1$. The SDEs for $W$ and $X$ are
\begin{equation}
\begin{split}\label{SDE_X}
 d W_{t} & =  \mu  X_{t-} W_{t-} d t +  \sigma X_{t-} W_{t-} d B_{t} - ( \be  M_{t}^{+} + \ue M_{t}^{-} ) d P_{t},  \\
  d X_{t} & = X_{t-} (1 - X_{t-}) \left( \mu - \sigma^{2} X_{t-} \right) d t + \sigma X_{t-} ( 1 - X_{t-}) d B_{t}   +  \tfrac{M_{t} + (\be M_{t}^{+} + \ue M_{t}^{-}) X_{t-}}{W_{t-} - \be M_{t}^{+} - \ue M_{t}^{-}} d P_{t},
\end{split}  
\end{equation}
where the initial conditions are $W_0=w_0$ and $X_0=x_0:=w^{(1)}_0/w_0$.

Let $T>0$ be a constant representing the terminal time. The investor's utility maximization problem is defined as follows: for a given $\gamma \in (0, \infty) \setminus \{ 1 \}$, 
\begin{align} \label{Goal}
  \sup_{(M_{t})_{t \in [0,T]}} \; \EE \Big[ \tfrac{W_{T}^{1 - \gamma}}{1 - \gamma} \Big],
\end{align}
where the supremum is taken over all admissible trading strategies.
\bigskip

\section{The value function}

Let $V$ be the value function of the utility maximization problem \eqref{Goal}: 
\begin{align}
  V(t, x, w) & = \sup_{(M_{s})_{s \in [t,T]}} \EE \Big[ \tfrac{W_{T}^{1 - \gamma}}{1 - \gamma} \, \Big | \, \mathcal{F}_{t} \Big] \bigg|_{(X_{t}, W_{t}) = (x, w)}. \label{Definition_of_value_function}
\end{align}
The scaling property of the wealth process and the property of the power function enable us to conjecture the form of the value function as
\begin{align*}
  V(t, x, w) = \tfrac{w^{1 - \gamma}}{1 - \gamma} \cdot v(t, x)
\end{align*}
for a function $v:[0,T]\times [0,1]\to \R$. The Hamilton-Jacobi-Bellman (HJB) equation for \eqref{Definition_of_value_function} produces the following partial differential equation (PDE) for $v$:
\begin{align} \label{Main_PDE}
  \begin{cases}
    1 = v (T, x),\\
  0 = v_{t}(t, x) + x (1 - x) ( \mu - \gamma \sigma^{2} x ) v_{x}(t, x) + \frac{\sigma^{2}x^{2} (1 - x)^{2} }{2}v_{x x}(t, x) + \left( Q(x) - \lambda \right) v(t,x)  \\
  \qquad + \lambda (1 - \gamma) \cdot \sup_{y \in [0, 1]} \left( \tfrac{v(t, y)}{1 - \gamma} \left( \big( \tfrac{1 + \be x}{1 + \be y} \big)^{1 - \gamma} 1_{\left\{ x\leq y  \right\}} + \big( \tfrac{1 - \ue x}{1 - \ue y} \big)^{1 - \gamma} 1_{\left\{ x>y \right\}} \right) \right), 
  \end{cases}
\end{align}
where $v_{t}$, $v_{x}$, $v_{x x}$ are partial derivatives and 
\begin{align}
Q(x) := - \tfrac{\gamma(1-\gamma)\sigma^2}{2}(x-y_M)^2 +  \tfrac{\gamma(1-\gamma)\sigma^2 y_M^2}{2}
 \quad \textrm{with}\quad y_M := \tfrac{\mu}{\gamma \sigma^{2}}.  \label{Q_def}
\end{align}
Note that $y_M$ is the Merton fraction, the optimal fraction in the frictionless market.


\begin{lemma} \label{v_classical_solution_and_properties}
There exists a unique $v \in C^{1, 2} ( [0, T] \times (0, 1) ) \cap C( [0, T] \times [0, 1] )$ that satisfies the following conditions:\\ 
(i) $v$ satisfies \eqref{Main_PDE} for $(t, x) \in [0, T] \times (0, 1)$. \\
(ii) For $x\in \{0,1\}$, the map $t\mapsto v(t,x)$ is continuously differentiable on $[0,T]$ and satisfies
\begin{align} \label{PDE_at_boundary}
  \begin{cases}
    1 = v (T, x),\\
  0 = v_{t}(t, x) +  \left( Q(x) - \lambda \right) v(t,x)  \\
  \qquad + \lambda (1 - \gamma) \cdot \sup_{y \in [0, 1]} \left( \tfrac{v(t, y)}{1 - \gamma} \left( \big( \tfrac{1 + \be x}{1 + \be y} \big)^{1 - \gamma} 1_{\left\{ x\leq y  \right\}} + \big( \tfrac{1 - \ue x}{1 - \ue y} \big)^{1 - \gamma} 1_{\left\{ x>y \right\}} \right) \right).
  \end{cases}
\end{align}
(iii) $v_{t}(t, x)$, $x (1 - x) v_{x}(t, x)$, $x^{2} (1 - x)^{2} v_{x x}(t, x)$ are uniformly bounded on $(t, x) \in [0, T] \times (0, 1)$.
\end{lemma}
\begin{proof}
See Appendix~\ref{appendix_A}. \bigskip
\end{proof}

\noindent The next theorem provides the verification. Its proof is similar to the proof of Theorem 3.5 in \cite{GC23}.

\begin{theorem} \label{Verification_v} 
Let $V$ be as in \eqref{Definition_of_value_function}, and $v$ be as in \lemref{v_classical_solution_and_properties}. Then, for $(t,x)\in [0,T]\times [0,1]$,
\begin{align}
 V(t, x, w) = \tfrac{w^{1 - \gamma}}{1 - \gamma} \cdot v(t, x). \label{value_check}
\end{align}
\end{theorem}

\begin{proof}
Without loss of generality, we prove $V(0, x_{0}, w_{0}) = \frac{w_{0}^{1 - \gamma}}{1 - \gamma} \cdot v(0, x_{0})$. Let $M$ be an admissible trading strategy and $(W,X)$ be the corresponding solution of \eqref{SDE_X}. 
Then \eqref{Admissible_class} and \eqref{SDE_X} imply  
\begin{align*}
0 \geq W_t - W_{t-}= - ( \be  M_{t}^{+} + \ue M_{t}^{-} ) d P_{t}  \geq - \tfrac{\be}{1+\be} W_{t-}^{(0)} - \ue W_{t-}^{(1)}.
\end{align*}
The above inequalities imply that there is a constant $c_0\in (0,1)$ such that
\begin{align}
c_0 W_{t-} \leq W_t \leq W_{t-} \quad \textrm{for}\quad t\in [0,T]. \label{W_bound1}
\end{align}
Then, \eqref{W_bound1} and \eqref{SDE_X} produce
\begin{align}
 w_0 c_0^{P_t} e^{ \int_0^t \mu X_s d s } \mathcal{E}(\sigma X \cdot B)_t \leq W_t \leq w_0 e^{ \int_0^t \mu X_s d s } \mathcal{E}(\sigma X \cdot B)_t, \label{W_bound2}
\end{align}
where $\mathcal{E}(\sigma X \cdot B)$ is the Dol{\'e}ans-Dade exponential of the process $\big(\int_0^t \sigma X_s dB_s \big)_{t\geq 0}$. Since $0\leq X \leq 1$, Novikov's condition implies that $\mathcal{E}(4(1-\gamma)\sigma X \cdot B)$ is a martingale. For $\tfrac{d\tilde{\PP}}{d\PP}\big|_{\mathcal{F}_t}=\mathcal{E}(4(1-\gamma)\sigma X \cdot B)_t$, $t\in [0,T]$ and a constant $b>0$,
\begin{align}
\EE\Big[ b^{P_t} \mathcal{E}(\sigma X \cdot B)_t^{2(1-\gamma)}\Big] 
&\leq \sqrt{\EE\big[ b^{2P_t} \big] \cdot  \EE\big[  \mathcal{E}(\sigma X \cdot B)_t^{4(1-\gamma)}\big]  } =e^{\frac{b^2-1}{2} \lambda t} \sqrt{   \EE^{\tilde{\PP}}\big[e^{2(3-4\gamma)(1-\gamma)\sigma^2 \int_0^t X_s^2 ds} \big]}   \nonumber \\
&\leq e^{|\frac{b^2-1}{2}| \lambda T + |(3-4\gamma)(1-\gamma)| \sigma^2 T}.\label{W_bound3}
\end{align}
We combine \eqref{W_bound2} and \eqref{W_bound3} to conclude that
\begin{align}
\sup_{t\in [0,T]} \EE \big[ W_{t}^{2 (1 - \gamma)} \big]<\infty. \label{W_bound4}
\end{align}

 Let $\tau_{n} :=T \wedge \inf \{ t \geq 0 : P_t = n \}$ for $n\in \N$ and $\tau_0:=0$. We observe that 
\begin{align}
\textrm{for} \quad t\in [\tau_n,\tau_{n+1}), \quad 
\begin{cases}
 \textrm{if  }X_{\tau_n}=0, \textrm{  then  } X_t=0. \\
\textrm{if  } X_{\tau_n}=1, \textrm{  then  }X_t=1.\\ 
\textrm{if  } X_{\tau_n}\in (0,1), \textrm{  then  } X_t\in (0,1).
\end{cases}\label{X_01}
\end{align}
We apply Ito's formula to $\frac{W_t^{1 - \gamma}}{1 - \gamma} \cdot v(t, X_t)$ with \eqref{SDE_X} and \eqref{X_01}, and use the fact that $X_t$ and $W_t$ can only jump at $t=\tau_n$ for $n\in \N$ to obtain
\begin{align} \label{Proof_of_verification}
  & \tfrac{W_{\tau_{n+1}}^{1 - \gamma}}{1 - \gamma} \cdot v(\tau_{n + 1}, X_{\tau_{n + 1}}) - \tfrac{W_{\tau_{n}}^{1 - \gamma}}{1 - \gamma} \cdot v(\tau_{n}, X_{\tau_{n}}) \nonumber \\
  &=\begin{dcases}
      \int_{\tau_n}^{\tau_{n+1}} \tfrac{W_{s-}^{1 - \gamma}}{1 - \gamma} \bigg(\Big( v_t(s,x)+ x (1 - x) \left( \mu - \gamma \sigma^{2} x \right) v_{x}(s, x) + \tfrac{\sigma^{2}x^{2} (1 -x)^{2}}{2} v_{x x}(s, x)\\
   \,\,   + Q(x) v(s,x) \Big)\Big|_{x=X_{s-}}ds   + \sigma \Big( (1-\gamma)x v(s,x)+x(1-x)v_x(s,x)\Big)\Big|_{x=X_{s-}} dB_s \bigg) \quad \textrm{if  }  X_{\tau_n}\in (0,1),\\
      \int_{\tau_n}^{\tau_{n+1}} \tfrac{W_{s-}^{1 - \gamma}}{1 - \gamma} \, \big( v_t(s,0) + Q(0)v(s,0)\big) ds    \quad \textrm{if} \,\, X_{\tau_n}=0,\\
       \int_{\tau_n}^{\tau_{n+1}} \tfrac{W_{s-}^{1 - \gamma}}{1 - \gamma}\Big( \big( v_t(s,1) + Q(1)v(s,1)\big) ds + (1-\gamma)\sigma v(s,1) dB_s \Big)   \quad \textrm{if} \,\, X_{\tau_n}=1,\\
\end{dcases} \nonumber\\
&\qquad + \tfrac{W_{\tau_{n+1}}^{1 - \gamma}}{1 - \gamma} \cdot v(\tau_{n + 1}, X_{\tau_{n + 1}})-\tfrac{W_{\tau_{n+1}-}^{1 - \gamma}}{1 - \gamma} \cdot v(\tau_{n + 1}, X_{\tau_{n + 1}-}).
\end{align}
Since $\lim_{n\to \infty}\tau_n = T$ almost surely, the above expression produces
\begin{align}
& \tfrac{W_T^{1 - \gamma}}{1 - \gamma} \cdot v(T, X_T)-  \tfrac{w_0^{1 - \gamma}}{1 - \gamma} \cdot v(0, x_0) \nonumber\\
&=\sum_{n=0}^\infty \Big(\tfrac{W_{\tau_{n+1}}^{1 - \gamma}}{1 - \gamma} \cdot v(\tau_{n + 1}, X_{\tau_{n + 1}}) - \tfrac{W_{\tau_{n}}^{1 - \gamma}}{1 - \gamma} \cdot v(\tau_{n}, X_{\tau_{n}}) \Big) \nonumber\\
&=\int_0^T  \tfrac{W_{s-}^{1 - \gamma}}{1 - \gamma} \bigg(\Big( v_t(s,x)+ x (1 - x) \left( \mu - \gamma \sigma^{2} x \right) v_{x}(s, x) + \tfrac{\sigma^{2}x^{2} (1 -x)^{2}}{2} v_{x x}(s, x) \nonumber\\
&\qquad \qquad\qquad+ Q(x) v(s,x) \Big) 1_{\{ x\in(0,1)\}} + \big(v_t(s,x)+Q(x)v(s,x)\big)   1_{\{ x\in\{0,1\}\}} \bigg)\bigg|_{x=X_{s-}} ds \nonumber\\
&\quad +\int_0^T \tfrac{W_{s-}^{1 - \gamma}}{1 - \gamma} \Big( \sigma \big( (1-\gamma)xv(s,x)+x(1-x)v_x(s,x)\big) 1_{\{ x\in(0,1)\}} + (1-\gamma)\sigma v(s,1)  1_{\{x=1\}}\Big) \Big|_{x=X_{s-}} dB_s \nonumber\\
&\quad + \sum_{0<s\leq T} \Big(\tfrac{W_s^{1 - \gamma}}{1 - \gamma} \cdot v(s, X_s) - \tfrac{W_{s-}^{1 - \gamma}}{1 - \gamma} \cdot v(s, X_{s-}) \Big). \nonumber
\end{align}
The stochastic integral term above is a martingale due to \lemref{v_classical_solution_and_properties} (iii) and \eqref{W_bound4}. The sum of jumps term above can be written as
\begin{align*}
\int_0^T \tfrac{W_{s-}^{1 - \gamma}}{1 - \gamma} \left( v(s, y)\left( \big( \tfrac{1 + \be x}{1 + \be y} \big)^{1 - \gamma} 1_{\left\{ x\leq y  \right\}} + \big( \tfrac{1 - \ue x}{1 - \ue y} \big)^{1 - \gamma} 1_{\left\{ x>y \right\}} \right) - v(s,x) \right)\Big|_{(x,y)=\left(X_{s-}, Y_{s}\right)} dP_s,
\end{align*}
where $Y_{s}:=\frac{X_{s-}W_{s-}+M_s}{W_{s-}-\be M_s^+ - \ue M_s^-}$.
Since $(P_t-\lambda t)_{t\in [0,T]}$ is a martingale, \lemref{v_classical_solution_and_properties} (iii) and \eqref{W_bound4} imply that the expected value of the above expression is
\begin{align*}
\EE \Big[  \int_0^T \tfrac{W_{s-}^{1 - \gamma}}{1 - \gamma} \lambda \Big( v(s, y) \left( \big( \tfrac{1 + \be x}{1 + \be y} \big)^{1 - \gamma} 1_{\left\{ x\leq y  \right\}} + \big( \tfrac{1 - \ue x}{1 - \ue y} \big)^{1 - \gamma} 1_{\left\{ x>y \right\}} \right) - v(s,x) \Big)\Big|_{(x,y)=\left(X_{s-}, Y_s\right)} ds \Big].
\end{align*}
We combine these observations to obtain
\begin{align}
& \EE \Big[ \tfrac{W_T^{1 - \gamma}}{1 - \gamma} \cdot v(T, X_T) \Big] -  \tfrac{w_0^{1 - \gamma}}{1 - \gamma} \cdot v(0, x_0) \nonumber\\
&=\int_0^T  \tfrac{W_{s-}^{1 - \gamma}}{1 - \gamma} \bigg(\Big( v_t(s,x)+ x (1 - x) \left( \mu - \gamma \sigma^{2} x \right) v_{x}(s, x) + \tfrac{\sigma^{2}x^{2} (1 -x)^{2}}{2} v_{x x}(s, x)+ (Q(x)-\lambda) v(s,x)  \nonumber\\
& \qquad + \lambda v(s, y) \Big( \big( \tfrac{1 + \be x}{1 + \be y} \big)^{1 - \gamma} 1_{\left\{ x\leq y  \right\}} + \big( \tfrac{1 - \ue x}{1 - \ue y} \big)^{1 - \gamma} 1_{\left\{ x>y \right\}} \Big) \Big) 1_{\{ x\in(0,1)\}} + \Big(v_{t}(s, x) +  \left( Q(x) - \lambda \right) v(s,x)  \nonumber \\
& \qquad\qquad+ \lambda  v(s, y) \Big( \big( \tfrac{1 + \be x}{1 + \be y} \big)^{1 - \gamma} 1_{\left\{ x\leq y  \right\}} + \big( \tfrac{1 - \ue x}{1 - \ue y} \big)^{1 - \gamma} 1_{\left\{ x>y \right\}} \Big) \Big)   1_{\{ x\in\{0,1\}\}} \bigg) \bigg|_{(x,y)=\left(X_{s-}, Y_s\right)}ds. \label{verification0}
\end{align}
The above equality and \lemref{v_classical_solution_and_properties} imply that for any admissible trading strategy $M$,
\begin{align}
 \EE \Big[ \tfrac{W_T^{1 - \gamma}}{1 - \gamma} \Big] \leq  \tfrac{w_0^{1 - \gamma}}{1 - \gamma} \cdot v(0, x_0). \label{verification1}
\end{align}

To complete the proof, we construct an optimal strategy $\hat M$ that satisfies the equality in \eqref{verification1}. We observe that the following map is continuous on $(t, x, y) \in [0, T] \times [0, 1]^{2}$:
\begin{align} \label{Supremum_part_map}
  (t, x, y) \mapsto \; \tfrac{v(t, y)}{1 - \gamma} \left( \big( \tfrac{1 + \be x}{1 + \be y} \big)^{1 - \gamma}1_{ \left\{ x\leq y \right\} } + \big( \tfrac{1 - \ue x}{1 - \ue y} \big)^{1 - \gamma}1_{ \left\{ x>y \right\} } \right).
\end{align}
Then, due to \lemref{meas_lem}, there exists a measurable function $\hat{y}: [0,T] \times [0,1] \to [0, 1]$ such that 
\begin{align} \label{Argmax_map}
  & \hat{y}(t, x) \in \argmax_{y \in [0, 1]} \left( \tfrac{v(t, y)}{1 - \gamma} \left( \big( \tfrac{1 + \be x}{1 + \be y} \big)^{1 - \gamma}1_{ \left\{ x\leq y \right\} } + \big( \tfrac{1 - \ue x}{1 - \ue y} \big)^{1 - \gamma} 1_{ \left\{ x>y \right\} } \right) \right).
\end{align}
We define a measurable function $m:[0,T]\times [0,\infty)\times [0,1] \to \R$ as
\begin{align} \label{Define_m}
  m(t, w,x) & := \tfrac{w (\hat{y}(t, x) - x)}{1 + \be \hat{y}(t, x)} \cdot 1_{ \left\{ x\leq \hat{y}(t, x) \right\} } + \tfrac{w (\hat{y}(t, x) - x)}{1 - \ue \hat{y}(t, x)} \cdot 1_{ \left\{ x>\hat{y}(t, x)  \right\} }.
\end{align}
Let $(\hat W, \hat X)$ be the solution of the SDEs in \eqref{SDE_X} with $M_t = m(t,W_{t-},X_{t-})$, and $\hat M_t :=  m(t,\hat W_{t-},\hat X_{t-})$. By construction, we have $\hat y(t, \hat X_{t-})= \tfrac{\hat X_{t-}\hat W_{t-} + \hat M_t}{\hat W_{t-} - \be \hat M_t^+ - \ue \hat M_t^-}$. Then \eqref{verification0}, \eqref{Argmax_map} and \lemref{v_classical_solution_and_properties} produce \eqref{verification1} with the equality. Therefore, we conclude \eqref{value_check} and the optimality of $\hat M$.
\end{proof}

The next lemma shows that $x\mapsto \tfrac{v(t,x)}{1-\gamma}$ is strictly concave and $v$ has uniform bounds independent of $\be, \ue$ and $\lambda$. To treat the concavity part, we define $\tV: [0,T]\times ([0,\infty)^2\setminus \{(0,0)\}) \to \R$ as
\begin{align}
\tilde V(t,a,b):=V\big(t, \tfrac{b}{a+b},a+b\big). \label{tV_def}
\end{align}
We notice that $\tilde V(t,a,b)$ is the value function of our control problem with $W_t^{(0)}=a$ and $W_t^{(1)}=b$.

\begin{lemma} \label{v_concave}
(i) For $t\in [0,T)$, the maps $(a,b)\mapsto \tilde V(t,a,b)$ and $x\mapsto \tfrac{v(t,x)}{1-\gamma}$ are strictly concave.

\noindent (ii) There are constants $\overline{v}\geq \underline{v}>0$ independent of $\be, \ue$ and $\lambda$ such that
\begin{align}
  \underline{v} \leq v(t, x) \leq \overline{v} \quad \textrm{for}\quad (t, x) \in [0, T] \times [0, 1]. \label{Range_of_v}
\end{align}
\end{lemma}
\begin{proof}
(i) This part of the proof is essentially the same as the proof of Proposition 3.6 in \cite{GC23}.

(ii) Let $M$ be an admissible trading strategy and $(W,X)$ be the corresponding solution of \eqref{SDE_X}. Then the right-hand side inequality in \eqref{W_bound2} implies
\begin{align}
\tfrac{W_T^{1-\gamma}}{1-\gamma}  \leq \tfrac{w_0^{1-\gamma}}{1-\gamma} \, e^{\int_0^T Q(X_s)ds} \mathcal{E}((1-\gamma)\sigma X \cdot B)_T, \label{W_bound11}
\end{align}
where $Q$ is defined in \eqref{Q_def}. Since $0\leq X \leq 1$, Novikov's condition implies that $\mathcal{E}((1-\gamma)\sigma X \cdot B)$ is a martingale. Then \eqref{W_bound11} implies that for $\tfrac{d\tilde{\PP}}{d\PP}\big|_{\mathcal{F}_T}=\mathcal{E}((1-\gamma)\sigma X \cdot B)_T$, 
\begin{align}
\EE\Big[\tfrac{W_T^{1-\gamma}}{1-\gamma} \Big] \leq \tfrac{w_0^{1-\gamma}}{1-\gamma}\cdot  \EE^{\tilde \PP} \Big[e^{\int_0^T Q(X_s)ds} \Big]. \label{W_bound12}
\end{align}
The definition of $V$ in \eqref{Definition_of_value_function}, \thmref{Verification_v} and \eqref{W_bound12} produce the following inequalities:
\begin{align}
v(0,x_0)\leq e^{ \| Q \|_\infty T} \quad \textrm{for}\quad 0<\gamma<1, \quad v(0,x_0)\geq e^{ - \| Q \|_\infty T} \quad \textrm{for}\quad \gamma>1.  \label{v_bound1}
\end{align}

Since $M_{s} = 0$ for all $s \in [0, T]$ is an admissible strategy, we have
\begin{align}
&\tfrac{w_0^{1-\gamma}}{1-\gamma} \cdot v(0,x_0)=V(0,x_0,w_0) \geq \EE\Big[ \tfrac{1}{1-\gamma} \Big( (1-x_0)w_0+x_0 w_0 e^{(\mu-\frac{\sigma^2}{2})T + \sigma B_T}\Big)^{1-\gamma} \Big] \nonumber\\
&\Longrightarrow \quad \tfrac{v(0,x_0)}{1-\gamma}\geq  \tfrac{1}{1-\gamma} \EE\Big[ \Big( 1-x_0+x_0  e^{(\mu-\frac{\sigma^2}{2})T + \sigma B_T}\Big)^{1-\gamma} \Big].\label{v_bound2}
\end{align}
The following inequalities can be checked easily:
\begin{align}
\begin{dcases}
\textrm{If  $0<\gamma<1$, then $( 1-x+x a)^{1-\gamma} \geq 1-x + x a^{1-\gamma}$ for $x\in [0,1]$ and $a>0$}.\\
\textrm{If  $\gamma>1$, then $( 1-x+x a)^{1-\gamma} \leq 1+ a^{1-\gamma}$ for $x\in [0,1]$ and $a>0$}.
\end{dcases} \label{v_bound3}
\end{align}
We combine \eqref{v_bound2}, \eqref{v_bound3} and $\EE\big[ e^{(1-\gamma)(\mu-\frac{\sigma^2}{2})T + (1-\gamma)\sigma B_T} \big]= e^{(1-\gamma)(\mu-\frac{\gamma\sigma^2}{2})T}$ to obtain
\begin{equation}
\begin{split}\label{v_bound4}
&v(0,x_0)\geq 1-x_0 + x_0 e^{(1-\gamma)(\mu-\frac{\gamma\sigma^2}{2})T} \geq e^{-|(1-\gamma)(\mu-\frac{\gamma\sigma^2}{2})T|}  \quad \textrm{for}\quad 0<\gamma<1,\\
&v(0,x_0)\leq 1+ e^{(1-\gamma)(\mu-\frac{\gamma\sigma^2}{2})T} \leq 1+ e^{|(1-\gamma)(\mu-\frac{\gamma\sigma^2}{2})T|}  \quad \textrm{for}\quad \gamma>1.\\
\end{split}
\end{equation}
We check that the inequalities in \eqref{v_bound1} and \eqref{v_bound4} still hold after replacing $v(0,x_0)$ by $v(t,x)$. 
\end{proof}

\section{The Optimal trading strategy}

In this section, we characterize the optimal trading strategy in terms of the {\it no-trade region}. We start with the construction of the candidate boundary points $\uy$ and $\by$ of the no-trade region.

\begin{lemma}\label{Boundaries_of_NT_region}
For each $t \in [0, T)$, there exist $0 \leq \uy(t) \leq \by(t) \leq 1$ such that 
\begin{align} \label{uy_oy_argmax}
  \{ \uy(t) \} = \argmax_{y \in [0, 1]} \left( \tfrac{v(t, y)}{(1 - \gamma) ( 1 + \be y )^{1 - \gamma}} \right), \quad \{ \by(t) \} = \argmax_{y \in [0, 1]} \left( \tfrac{v(t, y)}{(1 - \gamma) ( 1 - \ue y )^{1 - \gamma}} \right).
\end{align}
To be more specific, the following statements hold:\\
(i) The map $y \mapsto \frac{v(t, y)}{(1 - \gamma) ( 1 + \be y )^{1 - \gamma}}$ strictly increases on $y\in [0,\uy(t)]$ and decreases on $y\in [\uy(t),1]$.\\
 If $0<\uy(t)<1$, then $\uy(t)$ satisfies $\tfrac{v_{x}(t, \uy(t))}{1 - \gamma} = \tfrac{\be v(t, \uy(t))}{1 + \be \uy(t)}$.\\
(ii) The map $y \mapsto \frac{v(t, y)}{(1 - \gamma) ( 1 - \ue y )^{1 - \gamma}}$ strictly increases on $y\in [0,\by(t)]$ and decreases on $y\in [\by(t),1]$.\\
If $0<\by(t)<1$, then $\by(t)$ satisfies $\tfrac{v_{x}(t, \by(t))}{1 - \gamma} = - \tfrac{\ue v(t, \by(t))}{1 - \ue \by(t)}$.\\
(iii) For $(t, x) \in [0, T) \times [\uy(t), \by(t)]$, $- \tfrac{\ue}{1 - \ue}\, \overline{v} \leq \tfrac{v_{x}(t, x)}{1 - \gamma} \leq \be \, \overline{v}$ with $\overline v$ appears in \lemref{v_concave}.
\end{lemma}

\begin{proof}
Recall $\tilde V$ in \eqref{tV_def}. Due to \thmref{Verification_v}, the following equation holds:
\begin{align}
\tilde V\big(t,1-(1+\be)\eta,\eta\big) = \tfrac{v(t,y)}{(1-\gamma)(1+\be y)^{1-\gamma}} \Big|_{y=\frac{\eta}{1-\be \eta}} \quad \textrm{for}\quad \eta\in [0,\tfrac{1}{1+\be}]. \label{ylem1}
\end{align} 
\lemref{v_concave} (i) implies that the map $\eta\mapsto \tilde V(t,1-(1+\be)\eta,\eta)$ is strictly concave on $\eta\in [0,\tfrac{1}{1+\be}]$. Let $\underline \eta(t):=\argmax_{0\leq \eta\leq \frac{1}{1+\be}} \tilde V(t,1-(1+\be)\eta,\eta)$ be the unique maximizer. Since the map $\eta\mapsto \tfrac{\eta}{1-\be \eta}$ strictly increases on $\eta\in [0,\tfrac{1}{1+\be}]$, the definition of $\underline \eta(t)$ and \eqref{ylem1} imply that the left-hand side equation of \eqref{uy_oy_argmax} holds with $\uy(t)=\tfrac{\underline \eta(t)}{1-\be \underline \eta(t)}$ and the statements in (i) hold. 

Similarly, the following equation holds:
\begin{align}
\tilde V\big(t,1-(1-\ue)\eta,\eta\big) = \tfrac{v(t,y)}{(1-\gamma)(1-\ue y)^{1-\gamma}} \Big|_{y=\frac{\eta}{1+\ue \eta}} \quad \textrm{for}\quad \eta\in [0,\tfrac{1}{1-\ue}]. \label{ylem2}
\end{align} 
The strict concavity of $\eta\mapsto \tilde V(t,1-(1-\ue)\eta,\eta)$ ensures the existence of the unique maximizer $\overline \eta(t):=\argmax_{0 \leq \eta \leq \frac{1}{1-\ue}} \tilde V\big(t,1-(1-\ue)\eta,\eta\big)$.
Then, the right-hand side equation of \eqref{uy_oy_argmax} holds with $\by(t)=\tfrac{\overline \eta(t)}{1+\ue \overline \eta(t)}$ and the statements in (ii) hold. 

We observe that (i) and (ii) imply $ - \tfrac{\ue v(t, x)}{1 - \ue x}\leq \tfrac{v_{x}(t, x)}{1 - \gamma} \leq  \tfrac{\be v(t, x)}{1 + \be x}$ for $x\in [\uy(t),\by(t)]$. Then, we conclude (iii) by this observation and \lemref{v_concave} (ii). 

It only remains to check $\uy(t)\leq \by(t)$. The inequality holds when $\by(t)=1$. Suppose that $\by(t)<1$. Then, (ii) implies     
$\tfrac{v_{x}(t, \by(t))}{1 - \gamma} \leq  - \tfrac{\ue v(t, \by(t))}{1 - \ue \by(t)} \leq \tfrac{\be v(t, \by(t))}{1 + \be \by(t)}$, where the second inequality is due to the positivity of $v$. This observation and (i) produce $\uy(t)\leq \by(t)$. 
\end{proof}

In the proof of \thmref{Verification_v}, we construct the optimal trading strategy via $\hat y$ in \eqref{Argmax_map}. The next theorem explicitly characterizes $\hat y(t,x)$ in terms of $\uy(t)$ and $\by(t)$ as defined in \lemref{Boundaries_of_NT_region}.

\begin{theorem} For fixed $t \in [0, T)$, the argmax in \eqref{Argmax_map} is a singleton and $\hat y$ is
\begin{align} \label{Optimal_strategy}
  \hat{y}(t, x) =
  \begin{cases}
  \uy(t) &\textrm{if  } x\in [0,\uy(t))\\
  x  &\textrm{if  } x\in [\uy(t),\by(t)]\\
  \by(t)  &\textrm{if  } x\in (\by(t),1]\\
  \end{cases}
\end{align}
where $\uy(t)$ and $\by(t)$ are determined in \lemref{Boundaries_of_NT_region}.
\end{theorem}

\begin{proof}
We rephrase the maximization in \eqref{Argmax_map} as
\begin{equation}
\begin{split} \label{max_form1}
  & \max_{y \in [0, 1]}  \left( \tfrac{v(t, y)}{1 - \gamma} \left( \big( \tfrac{1 + \be x}{1 + \be y} \big)^{1 - \gamma}1_{ \left\{ x\leq y \right\} } + \big( \tfrac{1 - \ue x}{1 - \ue y} \big)^{1 - \gamma} 1_{ \left\{ x>y \right\} } \right) \right)  \\
  & = \max \bigg\{ \max_{y \in [x, 1]} \left( \tfrac{v(t, y)}{1 - \gamma} \big( \tfrac{1 + \be x}{1 + \be y} \big)^{1 - \gamma} \right), \max_{y \in [0, x]} \left( \tfrac{v(t, y)}{1 - \gamma} \big( \tfrac{1 - \ue x}{1 - \ue y} \big)^{1 - \gamma} \right) \bigg\}. 
\end{split}
\end{equation}
Using \lemref{Boundaries_of_NT_region}, we evaluate 
\begin{equation}
\begin{split} \label{max_form2}
  & \max_{y \in [x, 1]} \left( \tfrac{v(t, y)}{1 - \gamma} \big( \tfrac{1 + \be x}{1 + \be y} \big)^{1 - \gamma} \right) = \begin{cases} \tfrac{v(t, x)}{1 - \gamma} &\text{if  }  x \geq \uy(t) \\ 
  \tfrac{v(t, \uy(t))}{1 - \gamma} \big( \tfrac{1 + \be x}{1 + \be \uy(t)} \big)^{1 - \gamma} & \text{if  } x < \uy(t) 
  \end{cases},   \\  
  & \max_{y \in [0, x]} \left( \tfrac{v(t, y)}{1 - \gamma} \big( \tfrac{1 - \ue x}{1 - \ue y} \big)^{1 - \gamma} \right) = \begin{cases} \tfrac{v(t, x)}{1 - \gamma} & \text{if  } x \leq \by(t) \\ \tfrac{v(t, \by(t))}{1 - \gamma} \big( \frac{1 - \ue x}{1 - \ue \by(t)} \big)^{1 - \gamma} & \text{if  } x >\by(t)
   \end{cases}.
\end{split}
\end{equation}
Since $\uy(t) \leq \by(t)$, \eqref{max_form1} and \eqref{max_form2} imply
\begin{align}
  &\max_{y \in [0, 1]}  \Big( \tfrac{v(t, y)}{1 - \gamma} \Big( \big( \tfrac{1 + \be x}{1 + \be y} \big)^{1 - \gamma}1_{ \left\{ x\leq y \right\} } + \big( \tfrac{1 - \ue x}{1 - \ue y} \big)^{1 - \gamma} 1_{ \left\{ x>y \right\} } \Big) \Big) \nonumber\\
&  = \begin{cases} \tfrac{v(t, \uy(t))}{1 - \gamma} \big( \tfrac{1 + \be x}{1 + \be \uy(t)} \big)^{1 - \gamma} & \text{if  } x \in [0, \uy(t)) \\ \tfrac{v(t, x)}{1 - \gamma} & \text{if  }  x \in [\uy(t), \by(t)] \\ \tfrac{v(t, \by(t))}{1 - \gamma} \big( \tfrac{1 - \ue x}{1 - \ue \by(t)} \big)^{1 - \gamma} & \text{if  } x \in (\by(t), 1] \end{cases}  \label{L_origin}
\end{align}
and we conclude that the corresponding unique maximizer is as in \eqref{Optimal_strategy}.
\end{proof}

Our next task is to provide recursive stochastic representations for $v$ and its partial derivatives ($v_x, v_{xx}, v_t$, etc.) to facilitate further analysis. By applying the Feynman-Kac formula to the PDE \eqref{Main_PDE}, we derive the following recursive stochastic representation:\footnote{The representation in \eqref{implicit_expression} is recursive because it involves $v$ itself.}
\begin{align}
v(t,x)
&= \EE \Big[e^{ \int_t^T \left( Q(\Chi_s^{(t,x)})-\lambda \right) du} +  \int_{t}^{T} e^{ \int_t^s \left( Q(\Chi_s^{(t,x)})-\lambda  \right) du} \cdot \lambda (1 - \gamma) \nonumber\\
&\qquad  \cdot \sup_{y \in [0, 1]} \Big( \tfrac{v(s, y)}{1 - \gamma} \Big( \big( \tfrac{1 + \be \Chi_s^{(t,x)}}{1 + \be y} \big)^{1 - \gamma} 1_{\{ \Chi_s^{(t,x)}\leq y  \}} + \big( \tfrac{1 - \ue \Chi_s^{(t,x)}}{1 - \ue y} \big)^{1 - \gamma} 1_{\{ \Chi_s^{(t,x)}>y \}} \Big) \Big) ds \Big], \label{implicit_expression}
\end{align}
where $\Chi_s^{(t,x)}$ satisfies the following SDE:
\begin{align}
d\Chi_s^{(t,x)} &= \Chi_s^{(t,x)}(1-\Chi_s^{(t,x)})(\mu- \gamma \sigma^2 \Chi_s^{(t,x)})ds + \sigma \Chi_s^{(t,x)}(1-\Chi_s^{(t,x)}) dB_s, \quad \Chi_t^{(t,x)}=x. \label{Chi_SDE}
\end{align}
One may also derive the expressions for the partial derivatives of $v$ with respect to $x$ by taking the derivatives inside the expectation in \eqref{implicit_expression}. These expressions are in terms of $\Chi_s^{(t,x)}$ and its partial derivatives with respect to $x$, which are not {\it explicit}: the SDE in \eqref{Chi_SDE} does not have an explicit solution.\footnote{For example, $v_{xx}$ is expressed in terms of $\big(\Chi_s^{(t,x)}, \frac{\partial}{\partial x} \Chi_s^{(t,x)},\frac{\partial^2}{\partial x^2} \Chi_s^{(t,x)}\big)$. These processes satisfy a 3-dimensional SDE with Markovian structure, but the solutions are not explicit. In contrast, the processes $(Y_s^{(t,x)}, Z_s^{(t,x)})$ in \eqref{AYZ} and their partial derivatives with respect to $x$ are explicit, as shown in \eqref{YZ_derivative}.

 For the detailed analysis in later sections, it is more convenient to express $v$ and its derivatives in terms of {\it explicit stochastic processes} rather than of using the expression \eqref{implicit_expression} involving $\Chi_s^{(t,x)}$. 
 To this end, we define $ A_{s, t}$, $Y_{s}^{(t, x)}$ and $Z_{s}^{(t, x)}$ for $x\in [0,1]$ and $0\leq t \leq s \leq T$ as follows:
\begin{equation}
\begin{split}\label{AYZ}
  A_{s, t} & := e^{(\mu-\frac{\sigma^2}{2})(s - t) + \sigma (B_{s} - B_{t})} > 0, \\
  Y_{s}^{(t, x)} & := \tfrac{x A_{s, t}}{x A_{s, t} + 1 - x} \in [0, 1], \\
  Z_{s}^{(t, x)} & := \left( x A_{s, t} + 1 - x \right)^{1 - \gamma} > 0.
\end{split}
\end{equation}
These processes are {\it explicit} because their dependence on $x$ is straightforward, allowing for easy computation of their derivatives with respect to $x$. Subsequently, recursive expressions for $v$ and its partial derivatives are derived in terms of these explicit processes.
}
  
Observe that
\begin{align}
Z_{s}^{(t, x)} & = \left( \tfrac{1 - x}{1 - Y_{s}^{(t, x)}} \right)^{1 - \gamma} \quad \textrm{for} \quad x\in [0,1), \label{YZ}
\end{align}
and $A_{s, t}$ and $Y_{s}^{(t, x)}$ solve the following SDEs:
\begin{align}
dA_{s,t}& = \mu A_{s,t} ds + \sigma A_{s,t} dB_s, \quad A_{t,t}=1, \label{A_SDE}\\
dY_s^{(t,x)} &= Y_s^{(t,x)}(1-Y_s^{(t,x)})(\mu-\sigma^2 Y_s^{(t,x)})ds + \sigma Y_s^{(t,x)}(1-Y_s^{(t,x)}) dB_s, \quad Y_t^{(t,x)}=x. \label{Y_SDE}
\end{align} 
The following lemma is used to justify our later applications of the Leibniz integral rule.

\begin{lemma}\label{YZ_bound1}
For nonnegative integers $n$ and $m$ and nonnegative constants $k$ and $l$, 
\begin{align}
&\max_{0\leq t \leq s \leq T} \mathbb{E} \left[ \max_{0\leq x\leq 1} \Big( \left| \tfrac{\partial^{m} Z_{s}^{(t, x)}}{\partial x^{m}} \right |^{k} \cdot \left| \tfrac{\partial^{n} Y_{s}^{(t, x)}}{\partial x^{n}} \right |^{l} \Big) \right]  < \infty.\label{Uniform_boundedness_of_all_derivatives_exponent}
\end{align}
\end{lemma}
\begin{proof}
For $n\in \N$, direct computations produce 
\begin{equation}
\begin{split}\label{YZ_derivative}
\tfrac{\partial^{n} Y_{s}^{(t, x)}}{\partial x^{n}} &= \tfrac{n! A_{s,t}(1-A_{s,t})^{n-1}}{(x A_{s,t}+1-x)^{n+1}},\\\tfrac{\partial^{n} Z_{s}^{(t, x)}}{\partial x^{n}} &=(1-\gamma)(-\gamma)\cdots (2-\gamma-n)(A_{s,t}-1)^n (x A_{s,t}+1-x)^{1-\gamma-n}.
\end{split}
\end{equation}
Observe that
\begin{equation}
\begin{split}\label{A_ineq}
(x A+1-x)^c &\leq 1+ A^c \quad \textrm{for}\quad c\in \R, \,\,x\in [0,1],\,\,  A>0,\\
|A-1|^c &\leq 1+ A^c \quad \textrm{for}\quad c\geq 0, \,\, A>0,\\
\mathbb{E} [A_{s, t}^{c}] &\leq \exp \left( \left( \left|c \big(\mu-\tfrac{\sigma^2}{2}\big)\right| + \tfrac{c^{2} \sigma^{2}}{2} \right) T \right) \quad \textrm{for}\quad c\in \R, \,\, 0\leq t\leq s\leq T.
\end{split}
\end{equation}
The expression in \eqref{YZ_derivative} and the inequalities in \eqref{A_ineq} produce \eqref{Uniform_boundedness_of_all_derivatives_exponent}. 
\end{proof}

\begin{proposition}\label{vvx_representation}
Let $A_{s, t}$, $Y_{s}^{(t, x)}$ and $Z_{s}^{(t, x)}$ be defined as in \eqref{AYZ}, and let $L(t,x)$ be defined by 
\begin{align}
   L(t, x)  := \begin{cases} v(t, \uy(t)) \left( \tfrac{1 + \be x}{1 + \be \uy(t)} \right)^{1 - \gamma} & \text{if} \quad x \in [0, \uy(t)], \\ v(t, x) & \text{if} \quad x \in (\uy(t), \by(t)), \\ v(t, \by(t)) \left( \tfrac{1 - \ue x}{1 - \ue \by(t)} \right)^{1 - \gamma} & \text{if} \quad x \in [\by(t), 1]. \end{cases} \label{Supremum_term_L}
\end{align}
\noindent (i) For $(t,x)\in [0,T]\times [0,1]$, $v$ has the following representation:
\begin{align}
  v(t, x) & = e^{- \lambda (T - t)} \EE \left[ Z_{T}^{(t, x)} \right] + \lambda \int_{t}^{T} e^{- \lambda (s - t)} \EE \left[ Z_{s}^{(t, x)} L( s, Y_{s}^{(t, x)}) \right] d s. \label{v_with_optimizer} 
\end{align}

\noindent (ii) For $(t,x)\in [0,T) \times (0,1)$, the function $L$ is continuously differentiable with respect to $x$ and
\begin{align}
L_x(t,x)&=
\begin{cases}
\tfrac{\be (1 - \gamma) v(t, \uy(t))}{1 + \be x} \left( \tfrac{1 + \be x}{1 + \be \uy(t)} \right)^{1 - \gamma} &\textrm{if} \quad x\in (0,\uy(t)], \\
 v_{x} \left( t, x \right)&\textrm{if} \quad   x \in ( \uy(t), \by(t) ),   \\
  - \tfrac{\ue (1 - \gamma) v(t, \by(t))}{1 - \ue x} \left( \tfrac{1 - \ue x}{1 - \ue \by(t)} \right)^{1 - \gamma}&\textrm{if} \quad   x \in [\by(t), 1).
\end{cases}  \label{Lx_expression}
\end{align}
For $(t,x)\in [0,T)\times (0,1)$, $v_x(t,x)$ has the following representation:
\begin{equation}
\begin{split}\label{Dv_with_optimizer}
 v_{x}(t, x) & = e^{- \lambda (T - t)} \EE \Big[ \tfrac{\partial Z_{T}^{(t, x)}}{\partial x} \Big] \\
 &\quad + \lambda \int_{t}^{T} e^{- \lambda (s - t)} \EE \left[  Z_{s}^{(t, x)} L_x ( s, Y_{s}^{(t, x)} ) \tfrac{\partial Y_{s}^{(t, x)}}{\partial x} + \tfrac{\partial Z_{s}^{(t, x)}}{\partial x} L( s, Y_{s}^{(t, x)})  \right] d s. 
 \end{split}
\end{equation}
Furthermore, $v_{x}(t, 0) := \lim_{x \downarrow 0} v_{x}(t, x)$ and $v_{x}(t, 1) := \lim_{x \uparrow 1} v_{x}(t, x)$ are well-defined and finite.
\end{proposition}

\begin{proof}
(i) We observe that \eqref{L_origin} and \eqref{Supremum_term_L} imply
\begin{align}
L(t,x)=(1-\gamma) \sup_{y \in [0, 1]} \left( \tfrac{v(t, y)}{1 - \gamma} \left( \left( \tfrac{1 + \be x}{1 + \be y} \right)^{1 - \gamma} 1_{\left\{ y \in [x, 1] \right\}} + \left( \tfrac{1 - \ue x}{1 - \ue y} \right)^{1 - \gamma} 1_{\left\{ y \in [0, x) \right\}} \right) \right). \label{L_sup_exp}
\end{align}
Let $\tilde v(t,x):=e^{-\lambda t} \tfrac{ v(t,x)}{(1-x)^{1-\gamma}}$ for $(t,x)\in [0,T]\times [0,1)$. Then, \eqref{Main_PDE} and \eqref{L_sup_exp} imply
\begin{align} \label{tv_PDE}
  \begin{cases}
  0 = \tilde{v}_{t}(t, x) + x (1 - x) \left( \mu - \sigma^{2} x \right) \tilde{v}_{x}(t, x) + \frac{\sigma^{2}}{2} x^{2} (1 - x)^{2} \tilde{v}_{x x}(t, x)  + e^{-\lambda t} \tfrac{\lambda  L(t,x)}{(1-x)^{1-\gamma}},  \\
  \tfrac{e^{-\lambda T}}{(1-x)^{1-\gamma}} = \tilde v (T, x).
  \end{cases}
\end{align}
Let $x\in [0,1)$ fixed. We apply Ito's formula to $\tilde v(s,Y_s^{(t,x)})$ and use \eqref{tv_PDE} and \eqref{Y_SDE} to produce
\begin{align}
\tfrac{e^{-\lambda T}}{(1-Y_T^{(t,x)})^{1-\gamma}} + \int_t^T  \tfrac{\lambda e^{-\lambda s} L(s,Y_s^{(t,x)})}{(1-Y_s^{(t,x)})^{1-\gamma}} ds =\tilde v(t,x)+ \int_t^T\sigma (1-Y_s^{(t,x)})Y_s^{(t,x)} \tilde v_x (s, Y_s^{(t,x)}) dB_s.  \label{Kac_martingale}
\end{align}
Observe that the expectation of the stochastic integral term above is zero, because
\begin{align*}
&\EE\left[ \int_t^T \left(   (1-Y_s^{(t,x)})Y_s^{(t,x)} \tilde v_x (s, Y_s^{(t,x)})\right)^2 ds\right]\\
&=  \int_t^T \EE \left[ \left( e^{-\lambda s} \cdot  \tfrac{(1-Y_s^{(t,x)})Y_s^{(t,x)} Z_s^{(t,x)}  v_x (s, Y_s^{(t,x)}) +(1-\gamma)Y_s^{(t,x)} Z_s^{(t,x)} v(s, Y_s^{(t,x)})}{(1-x)^{1-\gamma}} \right)^2 \right] ds<\infty, 
\end{align*}
where the equality is due to the definition of $\tilde v$ and \eqref{YZ}, and the inequality is due to
 \lemref{v_classical_solution_and_properties} (iii) and \lemref{YZ_bound1}.
Therefore, \eqref{Kac_martingale} and \eqref{YZ} produce \eqref{v_with_optimizer} for $(t,x)\in [0,T]\times [0,1)$.
Since $v\in C \left( [0, T] \times [0, 1] \right)$, to complete the proof, it is enough to check that 
\begin{align*}
&\lim_{x\uparrow 1} \left( e^{- \lambda (T - t)} \EE \left[ Z_{T}^{(t, x)} \right] + \lambda \int_{t}^{T} e^{- \lambda (s - t)} \EE \left[ Z_{s}^{(t, x)} L( s, Y_{s}^{(t, x)}) \right] d s \right) \\
&\quad = e^{- \lambda (T - t)} \EE \left[ Z_{T}^{(t, 1)} \right] + \lambda \int_{t}^{T} e^{- \lambda (s - t)} \EE \left[ Z_{s}^{(t, 1)} L( s, Y_{s}^{(t, 1)}) \right] d s.
\end{align*}
Indeed, \lemref{YZ_bound1} and $\lVert  L \rVert_\infty<\infty$ allow us to apply the dominated convergence theorem above. 

\medskip

(ii) We differentiate \eqref{Supremum_term_L} with respect to $x$ and obtain \eqref{Lx_expression} for $x\in (0,1) \setminus \{ \uy(t),\by(t) \}$, and the continuity at $x\in \{ \uy(t),\by(t) \}$ is due to \lemref{Boundaries_of_NT_region}. By \lemref{Boundaries_of_NT_region} (iii) and \eqref{Range_of_v},
\begin{align}
\lVert  L_x \rVert_\infty  \leq C (\be + \ue) \label{Lx_bound}
\end{align}
for a constant $C>0$. We take derivative with respect to $x$ in \eqref{v_with_optimizer}, and put the derivative inside of the expectation (\lemref{YZ_bound1}, $\lVert  L \rVert_\infty<\infty$, and \eqref{Lx_bound} allow us to do this) to obtain \eqref{Dv_with_optimizer}. 

Finally, $\lim_{x \downarrow 0} v_{x}(t, x)$ and $\lim_{x \uparrow 1} v_{x}(t, x)$ are well-defined because $x \mapsto \tfrac{v(t, x)}{1 - \gamma}$ is strictly concave by \lemref{v_concave} (i), and these limits are finite due to \lemref{YZ_bound1}, $\lVert  L \rVert_\infty<\infty$ and \eqref{Lx_bound}.
\end{proof}

The next proposition presents properties about the boundaries of the no-trade region.

\begin{proposition} \label{Property_of_NT} Recall that we denote the Merton fraction as $y_M := \frac{\mu}{\gamma \sigma^{2}}$. 

\noindent (i) Let $t \in [0,T)$. If $0<y_M$, then $\by(t)>0$. If $y_M<1$, then $\uy(t)<1$.

\noindent (ii) If $0<y_M<1$ and at least one of $\be$ and $\ue$ is strictly positive, then $\uy(t) < \by(t)$ for $t\in [0,T)$.

\noindent (iii) If $0<y_M$ and $\underline{t} \in [0, T]$ is the solution of the equation (if a solution doesn't exist, we set $\underline{t}=0$)
 \begin{align}
  \be \, e^{\mu T} = \mu (1 + \be) \int_{\underline{t}}^{T} e^{(\mu + \lambda) s}\Big( e^{- \lambda T} + \lambda \int_{s}^{T} e^{- \lambda u} \tfrac{ v(u, \uy(u))}{( 1 + \be \uy(u) )^{1 - \gamma}} d u \Big)  d s, \label{ut_eq}
\end{align}
then $\uy(t)>0$ for $t\in [0,\underline{t})$ and $\uy(t)=0$ for $t\in [\underline{t},T)$.
 
If $y_M<1$ and $\overline{t} \in [0, T]$ is the solution of the equation (if a solution doesn't exist, we set $\overline{t}=0$)
\begin{align*}
  \ue & = (\gamma \sigma^{2} - \mu) \int_{\overline{t}}^{T} e^{(\gamma \mu - \frac{\gamma(1+\gamma)\sigma^2}{2})(T-s)} \Big( e^{- \beta (T-s)} 
  +  \lambda \int_{s}^{T} e^{- \beta (u-s)} v(u, \by(u)) \left( \tfrac{1 - \ue}{1 - \ue \by(u)} \right)^{1 - \gamma} d u \Big) d s,
\end{align*}
where $\beta:= \lambda-(1-\gamma)\mu + \tfrac{\gamma(1-\gamma) \sigma^2}{2}$, then $\by(t)<1$ for $t\in [0,\overline{t})$ and $\by(t)=1$ for $t\in [\overline{t},T)$.
\end{proposition}

Figure~\ref{NT} illustrates the no-trade region and $\overline t$ and $\underline t$ in \proref{Property_of_NT}.

\begin{figure}[t]
		\begin{center}
			\includegraphics[width=0.45\textwidth]{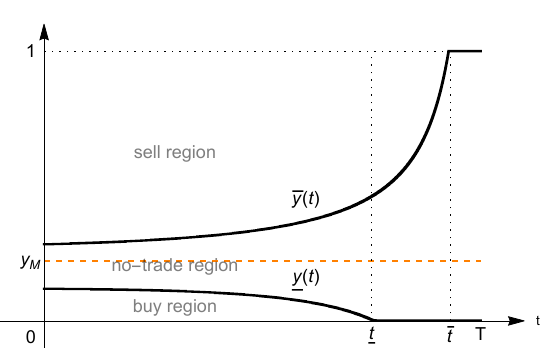}
	\end{center}
	 \caption{The graph shows $\underline{y}(t)$ and $\overline{y}(t)$ as functions of $t$. The dashed orange line is the Merton fraction $y_M=\frac{\mu}{\gamma \sigma^{2}}$. The parameters are $\mu=0.2, \, \sigma=1, \, \gamma=0.9, \, \lambda=3,\, \underline{\epsilon}=\overline{\epsilon}=0.05$ and $T=1$.
		 We generate the graph using \eqref{uy_oy_argmax} and the numerical solution of PDE \eqref{Main_PDE} obtained via the explicit finite difference method. The mesh sizes for the space and time variables are $1/10^3$ and $1/10^5$, respectively.} 
		\label{NT}
\end{figure}

\begin{proof}
(i) Assume that $0<y_M$ ($y_M<1$ case can be treated similarly).  \lemref{Boundaries_of_NT_region} implies that 
 \begin{align}
  \tfrac{v_{x}(t, 0)}{1 - \gamma} + \ue v(t, 0) > 0 \; \iff \; \by(t) > 0. \label{by>0}
  \end{align}
The expressions of $ Y_{s}^{(t,x)}$ and $Z_{s}^{(t, x)}$ in \eqref{AYZ}, $L$ in \eqref{Supremum_term_L} and $L_x$ in \eqref{Lx_expression} imply
\begin{align*}
&\lim_{x\downarrow 0}  Y_{s}^{(t,x)}=0, \quad \lim_{x\downarrow 0}   \tfrac{\partial Y_{s}^{(t, x)}}{\partial x}=A_{s,t},\quad \lim_{x\downarrow 0}   Z_{s}^{(t, x)} =1,\quad  \lim_{x\downarrow 0}   \tfrac{\partial Z_{s}^{(t, x)}}{\partial x} =(1-\gamma) \left(A_{s,t} - 1 \right),\\
&\lim_{x\downarrow 0} L(t,x)=  \tfrac{v(t, \uy(t))}{( 1 + \be \uy(t) )^{1 - \gamma}} 1_{\left\{ \uy(t) > 0 \right\}} + v(t,0) 1_{\left\{ \uy(t) = 0 \right\}},\\
&L_x(t, 0):=\lim_{x\downarrow 0} L_x(t, x) = \tfrac{\be(1-\gamma) v(t, \uy(t))}{( 1 + \be \uy(t) )^{1 - \gamma}} 1_{\left\{ \uy(t) > 0 \right\}} + v_{x} (t,0) 1_{\left\{ \uy(t) = 0 < \by(t) \right\}} - \ue (1-\gamma) v(t, 0) 1_{\left\{ \by(t) = 0 \right\}}, 
\end{align*}
where we use $v_{x}(t, 0) := \lim_{x \downarrow 0} v_{x}(t, x)$ in \proref{vvx_representation} (ii).
$v>0$ and \eqref{by>0} imply
\begin{align}
L(t,0) \geq v(t,0) 1_{\left\{ \uy(t) = 0 \right\}}, \quad  \tfrac{L_x(t, 0)}{1-\gamma} \geq -\ue v(t,0) 1_{\left\{ \uy(t) = 0 \right\}}.\label{Lx_0}
\end{align}
We take limit $x\downarrow 0$ in \eqref{v_with_optimizer} and \eqref{Dv_with_optimizer} and apply the dominated convergence theorem (justified by \lemref{YZ_bound1}) to obtain
\begin{align*}
   \tfrac{v_{x}(t, 0)}{1 - \gamma} + \ue v(t, 0) &= e^{- \lambda (T - t)} \EE [ A_{T,t} - 1+\ue ]   +  \lambda\int_{t}^{T} e^{- \lambda (s - t)} \Big( \tfrac{\EE[ A_{s,t}] L_x(s,0)}{1-\gamma} + \EE[ A_{s,t}-1+\ue ] L(s,0) \Big) ds \\
  & \geq e^{- \lambda (T - t)}  \big( e^{\mu (T - t)} - 1+\ue \big)  + \lambda\int_{t}^{T} e^{- \lambda (s - t)} \big( e^{\mu (s - t)} - 1 \big) (1 - \ue) v(s,0) 1_{\left\{ \uy(s) = 0 \right\}}  d s \\
  &> 0,
\end{align*}
where the inequalities are due to \eqref{Lx_0}, $v>0$ and $\mu>0$ (implied by $y_M>0$). Therefore, by \eqref{by>0} and the above inequality, we conclude that $\by(t)>0$.

\medskip

(ii) Suppose that $0<y_M<1$ and at least one of $\be$ and $\ue$ is strictly positive. By part (i) result, we have $\uy(t)<1$ and $\by(t)>0$. In case $\uy(t)=0$ or $\by(t)=1$, then $\uy(t)<\by(t)$. Therefore, it remains to consider the case that $0<\uy(t) \leq \by(t)<1$. By \lemref{Boundaries_of_NT_region} and $v>0$, we have
\begin{align}
\tfrac{v_{x}(t, \uy(t))}{1 - \gamma} = \tfrac{\be v(t, \uy(t))}{1 + \be \uy(t)} > - \tfrac{ \ue v(t, \by(t))}{1 - \ue \by(t)} = \tfrac{v_{x}(t, \by(t))}{1 - \gamma}.
\end{align}
The above inequality and the strict concavity of $x\mapsto \tfrac{v(t,x)}{1-\gamma}$ in \lemref{v_concave} imply $\uy(t)<\by(t)$.

\medskip

(iii) Assume that $0<y_M$ ($y_M<1$ case can be treated similarly). \lemref{Boundaries_of_NT_region} implies that 
 \begin{align}
  \tfrac{v_{x}(t, 0)}{1 - \gamma} - \be v(t, 0) > 0 \; \iff \; \uy(t) > 0. \label{uy>0}
  \end{align}
By the same way as in the proof of part (i), we take limit $x\downarrow 0$ in \eqref{v_with_optimizer} and \eqref{Dv_with_optimizer} and obtain
\begin{align}
   \tfrac{v_{x}(t, 0)}{1 - \gamma} - \be v(t, 0) &= e^{- \lambda (T - t)} \EE [ A_{T,t} - 1-\be ]   + \lambda \int_{t}^{T} e^{- \lambda (s - t)} \Big( \tfrac{\EE[ A_{s,t}] L_x(s,0)}{1-\gamma} + \EE[ A_{s,t}-1-\be] L(s,0) \Big) ds  \nonumber\\
  & = e^{- \lambda (T - t)}  \big( e^{\mu (T - t)} - 1-\be \big) + \lambda \int_{t}^{T} e^{- \lambda (s - t)} \big( e^{\mu (s - t)} - 1 \big) (1 + \be) \tfrac{ v(s, \uy(s))}{( 1 + \be \uy(s) )^{1 - \gamma}} d s \nonumber \\
  & \qquad + \lambda \int_{t}^{T} e^{- \lambda (s - t)} e^{\mu (s - t)} \Big( \tfrac{v_{x}(s, 0)}{1 - \gamma} - \be v(s, 0) \Big) 1_{\left\{ \uy(s) = 0 \right\}} d s,  \label{Equation_of_overline_D}
  \end{align}
where the second equality is due to $\by(t)>0$ by part (i). Let $f, g:[0,T]\to \R$ be defined as
$$
f(t):= e^{- \lambda t} \Big( \tfrac{v_{x}(t, 0)}{1 - \gamma} - \be v(t, 0) \Big), \quad g(t):=\lambda e^{-\lambda t} \tfrac{ v(t, \uy(t))}{( 1 + \be \uy(t) )^{1 - \gamma}}.
$$
Then, we can rewrite \eqref{Equation_of_overline_D} as
\begin{align*}
f(t)& = e^{- \lambda T}  \big( e^{\mu (T - t)} - 1-\be \big) +  (1 + \be)\int_{t}^{T}  \big( e^{\mu (s - t)} - 1 \big) g(s) d s  + \lambda \int_{t}^{T}  e^{\mu (s - t)} f(s) 1_{\left\{ f(s) \leq 0 \right\}} d s,
\end{align*}
where we use the equivalence of $f(t)\leq 0$ and $\uy(t)=0$ by \eqref{uy>0}. We differentiate above to obtain
\begin{align}\label{Equation_overline_D}
f'(t) = - \mu \Big(f(t) + (1 + \be) \Big( e^{- \lambda T} + \int_{t}^{T} g(s) d s \Big) \Big) - \lambda f(t) 1_{ \{ f(t) \leq 0 \} }.
\end{align}

We define $\underline t$ as
\begin{align}\label{underline_T}
  \underline{t} := \inf \{ t \in [0, T] :  f(s) \leq 0 \textrm{   for all   } s\in [t,T] \},
\end{align}
then the set in \eqref{underline_T} is non-empty because $f(T) = - \be e^{- \lambda T} \leq 0$. Since $\mu>0$ (due to $y_\infty>0$) and $g>0$, the form of ODE \eqref{Equation_overline_D} and definition of $\underline t$ above imply $f(t)>0$ for $t\in [0,\underline t)$. This observation and \eqref{uy>0} imply $\uy(t)>0$ for $t\in [0,\underline{t})$ and $\uy(t)=0$ for $t\in [\underline{t},T)$. To determine $\underline t$, it is enough to observe that the solution of ODE \eqref{Equation_overline_D} for $t\in [\underline t, T)$ is
\begin{align}
e^{(\mu + \lambda) t}  f(t) =   \mu (1 + \be) \int_{t}^{T}e^{(\mu + \lambda) s} \Big( e^{- \lambda T} + \int_{s}^{T} g(u) d u \Big)  d s - \be \, e^{\mu T}.
\end{align}
If there is no solution to \eqref{ut_eq}, then $f(t)> 0$ for $t\in [0,T)$ and $\underline t=0$. If there is a solution to \eqref{ut_eq}, then such a solution should be unique and $f(\underline t)=0$.
\end{proof}

\medskip

\section{Asymptotic results}

In this section, we provide asymptotic results to analyze the utility maximization problem when both transaction costs and search frictions are small. For convenience, we assume throughout this section that $\be = \ue=:\epsilon\in (0,1)$ and $0<y_M<1$. We focus on the asymptotics of the no-trade region and the value function as $\epsilon\downarrow 0$ and $\lambda \to \infty$ {\it simultaneously}, and then compare these results with the already-known asymptotic results in the benchmark cases of transaction costs only ($\lambda=\infty$) and search frictions only ($\epsilon=0$).

Our heuristic inspection using the results in \cite{GC23} implies that the {\it multivariable limits} as $(\epsilon,\lambda) \to (0,\infty)$ do not exist in general. It turns out that a specific scaling relation $\lambda=c \, \epsilon^{-\frac{2}{3}}$ for $c>0$ is relevant to consider, as explained below. For the log utility and fixed $\lambda<\infty$, Section 5 in \cite{GC23} provides asymptotic results as $\epsilon\downarrow 0$. According to Proposition 5.5 and Proposition 5.7 in \cite{GC23}, one can check the following limits:
\begin{equation}
\begin{split}\label{heuristic1}
&\lim_{\lambda \to \infty} \left( \lim_{\epsilon\downarrow 0} \tfrac{\textrm{``no-trade region width"}}{\lambda \epsilon}          \right)= \tfrac{2}{\sigma^2} \neq 0 = \lim_{\epsilon\downarrow 0}  \left( \lim_{\lambda \to \infty}\tfrac{\textrm{``no-trade region width"}}{\lambda \epsilon}\right),\\
&\lim_{\lambda \to \infty} \left( \lim_{\epsilon\downarrow 0} \tfrac{\textrm{``decrease of value"}}{\sqrt{\lambda} \epsilon}  \right)= \tfrac{\sigma y_M (1-y_M)(T-t)}{\sqrt{2}} \neq 0 = \lim_{\epsilon\downarrow 0}  \left( \lim_{\lambda \to \infty} \tfrac{\textrm{``decrease of value"}}{\sqrt{\lambda} \epsilon}\right).
\end{split}
\end{equation}  
Therefore, the {\it multivariable limits} $\lim_{(\epsilon,\lambda)\to (0,\infty)}\tfrac{\textrm{``no-trade region width"}}{\lambda \epsilon}  $ and $\lim_{(\epsilon,\lambda)\to (0,\infty)}\tfrac{\textrm{``decrease of value"}}{\sqrt{\lambda} \epsilon}  $ do not exist in general. 
On the other hand, it is well known in the literature (see \cite{JanShr04, GerMuhSch11, Cho11, Bic12}) that in the case of transaction costs only ($\lambda=\infty$), the asymptotics are as follows: 
\begin{align}
\textrm{``no-trade region width"} = O(\epsilon^{\frac{1}{3}}), \quad \textrm{``decrease of value"} =O( \epsilon^{\frac{2}{3}}). \label{heuristic2}
\end{align}
In the case of search frictions only ($\epsilon=0$, see \cite{matsumotopower, matsumoto2006}), the decrease of value is $O(\tfrac{1}{\lambda})$ and the width of the no-trade region is zero. We combine this observation with \eqref{heuristic1} and \eqref{heuristic2} and attempt to match the orders. We naturally conjecture that a suitable relation between $\epsilon$ and $\lambda$ for the asymptotics would satisfy\footnote{
\cite{EkrMuh2019} investigates a model incorporating both temporary and transient price impacts, presenting asymptotic results for the joint limit as both impacts diminish. To obtain nontrivial limits, \cite{EkrMuh2019} appropriately rescales the price impact parameters. In this sense, our consideration of \eqref{heuristic2} shares a similar motivation with the asymptotic setting in \cite{EkrMuh2019}.
}
\begin{align}
\lambda \epsilon \sim \epsilon^{\frac{1}{3}}\quad  \textrm{and} \quad \sqrt{\lambda} \epsilon  \sim  \epsilon^{\frac{2}{3}} \sim \tfrac{1}{\lambda} \quad \Longrightarrow \quad \lambda \sim \epsilon^{-\frac{2}{3}}.\label{curve_motive}
\end{align} 

Motivated by the above discussion, we make the following assumption, which holds throughout this section.

\begin{assumption}\label{ass}
(i) For $c>0$ and $\epsilon \in (0,1)$, $\be = \ue=\epsilon$ and $\lambda=c \, \epsilon^{-\frac{2}{3}}$. (ii) $y_M \in (0,1)$.
\end{assumption}

\begin{notation}\label{notation}\ \\
(i) Under \assref{ass}, to emphasize their dependence on $\epsilon\in (0,1)$ (with $\lambda$ dependence through the relation $\lambda=c \, \epsilon^{-\frac{2}{3}}$), we denote $v,v_x,v_{xx},v_t, \uy, \by, L, L_x$, etc. by $v^\epsilon,v_x^\epsilon,v_{xx}^\epsilon,v_t^\epsilon, \uy^\epsilon, \by^\epsilon, L^\epsilon, L_x^\epsilon$, etc. 

\noindent(ii) The case of the perfectly liquid market ($\epsilon=0$ and $\lambda=\infty$) corresponds to the classical Merton problem, and we denote the value function and optimal fraction as $v^0$ and $y_M$.

\noindent(iii) In the case of search frictions only (no transaction costs, $\epsilon=0$), we denote the value function and optimal fraction as $v^{SO,\lambda}$ and $\hat y^{SO, \lambda}$ to emphasize their dependence on $\lambda$.

\noindent(iv) In the case transaction costs only (no search frictions, $\lambda=\infty$), we denote the value function and the no-trade boundaries as $v^{TO,\epsilon}$, $\by^{TO, \epsilon}$ and $\uy^{TO, \epsilon}$ to emphasize their dependence on $\epsilon$.
\end{notation}

As benchmarks for our asymptotic results, we present the asymptotic results for the cases of  transaction costs only (see \cite{Bic12}) and search frictions only (see \cite{matsumotopower, matsumoto2006}). 

\begin{itemize}
\item In the case where there are no transaction costs or search frictions ($\epsilon = 0$ and $\lambda = \infty$), the utility maximization problem becomes the classical Merton problem \cite{Mer69, Mer71}. The explicit formula and HJB equation for the corresponding value function $v^{0}$ are:
\begin{align}
&v^{0}(t) = e^{Q(y_M) (T - t)}, \quad \begin{cases}v_{t}^{0}(t) + Q(y_M) v^{0}(t) = 0,\\
v^{0}(T)=1.
\end{cases} \label{merton_pde}
\end{align}
\item In the case where there are transaction costs only ($\epsilon\in (0,1)$ and $\lambda=\infty$), the utility maximization problem becomes the problem investigated in \cite{Bic12}. The asymptotic results are as follows:
\begin{equation}
\begin{split}\label{TO_benchmark}
\by^{TO, \epsilon}(t) &=y_M + \tfrac{1}{2} \Big( \tfrac{12y_M^2(1-y_M)^2}{\gamma}\Big)^{\frac{1}{3}}  \cdot \epsilon^{\frac{1}{3}} + o(\epsilon^{\frac{1}{3}}),\\
\uy^{TO, \epsilon}(t) &=y_M - \tfrac{1}{2} \Big( \tfrac{12y_M^2(1-y_M)^2}{\gamma}\Big)^{\frac{1}{3}}  \cdot \epsilon^{\frac{1}{3}} + o(\epsilon^{\frac{1}{3}}),\\
v^{TO,\epsilon}(t,y_M) &=v^0(t) - \tfrac{(1-\gamma)\gamma \sigma^2}{8} \Big( \tfrac{12y_M^2(1-y_M)^2}{\gamma}\Big)^{\frac{2}{3}} v^0(t)(T-t) \cdot \epsilon^{\frac{2}{3}} + o(\epsilon^{\frac{2}{3}}).
\end{split}
\end{equation} 
\item In the case where there are search frictions only ($\epsilon=0$ and $\lambda<\infty$), the utility maximization problem becomes the problem investigated in \cite{matsumotopower}. The asymptotic results are as follows: 
\begin{equation}
\begin{split}\label{SO_benchmark}
\hat y^{SO, \lambda}(t) &= y_M + \sigma^2 y_M(1-y_M)(2y_M-1) \cdot \tfrac{1}{\lambda} + o(\tfrac{1}{\lambda}),\\
v^{SO,\lambda}(t, y_M) &= v^0(t) - \tfrac{(1-\gamma)\gamma \sigma^4 y_M^2 (1-y_M)^2}{2}  v^0(t) (T-t) \cdot \tfrac{1}{\lambda} + o(\tfrac{1}{\lambda}).
\end{split}
\end{equation}  
Notice that the no-trade region vanishes in this case, $\hat{y}^{SO, \lambda}(t) =\by^{SO,\lambda}(t)=\uy^{SO,\lambda}(t)$.
\end{itemize}
\begin{remark}
As seen in \eqref{TO_benchmark} and \eqref{SO_benchmark}, when only one type of friction is present, the coefficients of the correction terms (with respect to the friction parameter) are explicitly determined.
In contrast, \cite{GC23} examines the case where both types of frictions are present, providing asymptotic results only for a small transaction cost parameter $\epsilon$, while keeping $\lambda<\infty$ fixed. Consequently, the asymptotic coefficients in \cite{GC23} depend on $\lambda$ in an implicit manner, as they are expressed in terms of solutions to partial differential equations (see Theorems 5.4 and 5.6 in \cite{GC23}). 

A key motivations for the asymptotic analysis in this paper is to develop a method for obtaining explicit coefficients by allowing both $\epsilon\to 0$ and $\lambda\to \infty$, in contrast to \cite{GC23}, where $\lambda<\infty$ remains fixed. As discussed around \eqref{heuristic1} and \eqref{curve_motive}, achieving this requires selecting an appropriate curve along which the limits are taken.
\end{remark}

The following theorem is the main result of this paper. Along the parametric curve $\lambda=c \, \epsilon^{-\frac{2}{3}}$ (see the discussion for \eqref{curve_motive}), the boundaries of the no-trade region and the value function have asymptotic expansions in terms of $\epsilon^{\frac{1}{3}}$ and $\epsilon^{\frac{2}{3}}$, respectively.

\begin{theorem} \label{Joint_limit_of_Wnt_and_Vd}
Let \assref{ass} hold and $a_1, a_2: (0, \infty) \rightarrow (0, \infty)$ be defined as
\begin{equation}
\begin{split}\label{a1a2_def}
a_1(c) & := \tfrac{ \sigma y_M (1 - y_M)}{ \sqrt{2\,c}} \Big( \Big( \tfrac{3 \sqrt{2} \,c^{\frac{3}{2}}}{\gamma \sigma^{3} y_M (1 - y_M)} + 1 \Big)^{\frac{1}{3}} - 1 \Big), \\
  a_2(c) & := \tfrac{\gamma (1-\gamma) \sigma^{4} y_M^{2} (1 - y_M)^{2}}{4 \,c} \Big( \Big( \tfrac{3 \sqrt{2} \,c^{\frac{3}{2}}}{\gamma \sigma^{3} y_M (1 - y_M)} + 1 \Big)^{\frac{2}{3}} + 1 \Big).
\end{split}
\end{equation}
Then, for $t \in [0, T)$, 
\begin{equation}
\begin{split}\label{asymptotics_ntr}
    \by^{\epsilon}(t)  &=y_M+  a_1(c)\cdot \epsilon^{\frac{1}{3}} + o(\epsilon^{\frac{1}{3}}), \\
    \uy^{\epsilon}(t) &=y_M -  a_1(c)\cdot\epsilon^{\frac{1}{3}} + o(\epsilon^{\frac{1}{3}}), \\
    v^{\epsilon}(t,y_M) &= v^{0}(t)  - a_2(c)v^{0}(t) (T - t) \cdot\epsilon^{\frac{2}{3}} + o(\epsilon^{\frac{2}{3}}).
\end{split}
\end{equation}
Alternatively, due to the relation $\lambda=c \, \epsilon^{-\frac{2}{3}}$, the above asymptotics can be written in terms of $\lambda$:
\begin{align*}
    \by^{\epsilon}(t) &=y_M+  \sqrt{c} \,a_1(c)\cdot \tfrac{1}{\sqrt{\lambda}}  + o(\tfrac{1}{\sqrt{\lambda}}), \\
    \uy^{\epsilon}(t) &=y_M-  \sqrt{c} \,a_1(c)\cdot \tfrac{1}{\sqrt{\lambda}}  + o(\tfrac{1}{\sqrt{\lambda}}), \\
    v^{\epsilon}(t,y_M) &= v^{0}(t)  - c\, a_2(c)v^{0}(t) (T - t) \cdot \tfrac{1}{\lambda} + o(\tfrac{1}{\lambda}).
\end{align*}
\end{theorem}

The proof of the theorem is postponed to Section~\ref{proof_main}. 
One of the main difficulties in the analysis is the rigorous treatment of subtle limiting behaviors that do not appear in the benchmark cases. For example, $\lim_{\epsilon\downarrow 0} v_{xx}^\epsilon(t,x)/\epsilon^{\frac{2}{3}}$ depends on the choice of $x\in [\uy^\epsilon(t),\by^\epsilon(t)]$ in our model, whereas $\lim_{\epsilon\downarrow 0} v_{xx}^{TO,\epsilon}(t,x)/\epsilon^{\frac{2}{3}}=0$ for $x\in [\uy^\epsilon(t),\by^\epsilon(t)]$. We present \lemref{really_used} and \lemref{v_xx_conv_lem} to address these subtle limiting behaviors. Figure~\ref{NTW_VD} illustrates the asymptotics in \thmref{Joint_limit_of_Wnt_and_Vd}.
\begin{figure}[t]
		\begin{center}$
			\begin{array}{cc}
			\includegraphics[width=0.45\textwidth]{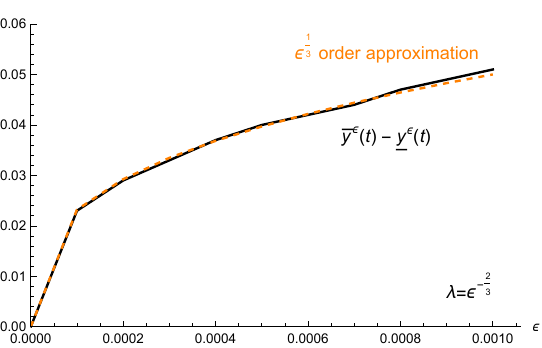} \,\,\, & \,\,\,
 			\includegraphics[width=0.45\textwidth]{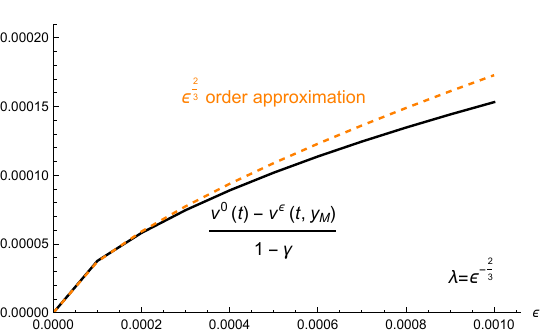} 		
			\end{array}$
	\end{center}
	 \caption{The graphs describe ``width of the no-trade region" and ``value decrease" as functions of $\epsilon$ along the parametric curve $\lambda =c\, \epsilon^{- \frac{2}{3}}$, where the dashed orange lines are their approximations $2a_1(c) \cdot \epsilon^{\frac{1}{3}}$ and $\tfrac{a_2(c)v^0(t)(T-t)}{1-\gamma} \cdot \epsilon^{\frac{2}{3}}$ in \thmref{Joint_limit_of_Wnt_and_Vd}.
The parameters are $c=1,\, \mu=0.2, \, \sigma=1, \, \gamma=0.9$, $t=0.75$ and $T=1$. 
We generate the graphs using \eqref{uy_oy_argmax} and the numerical solution of the PDE \eqref{Main_PDE}, obtained via the explicit finite difference method. The mesh sizes for the space and time variables are $1/10^3$ and $1/10^5$, respectively. The graphs are plotted using 10 points of $\epsilon$, ranging from $0.0001$ to $0.001$ in increments of $0.0001$.} 
		\label{NTW_VD}
\end{figure}

From \eqref{a1a2_def}, direct computations using L'Hopital's rule produce the following limits:
\begin{align}
& \lim_{c \to \infty} a_1(c) =\tfrac{1}{2} \Big( \tfrac{12y_M^2(1-y_M)^2}{\gamma}\Big)^{\frac{1}{3}}, \quad\lim_{c \to \infty} a_2(c)=\tfrac{(1-\gamma)\gamma \sigma^2}{8} \Big( \tfrac{12y_M^2(1-y_M)^2}{\gamma}\Big)^{\frac{2}{3}}, \label{a1_limit}\\
& \lim_{c \to 0} \sqrt{c} \, a_1(c) =0,\quad  \lim_{c \to 0} c\, a_2(c) =\tfrac{(1-\gamma)\gamma \sigma^4 y_M^2 (1-y_M)^2}{2} . 
\end{align}
Using the above limits, we can rephrase \eqref{TO_benchmark} and \eqref{SO_benchmark} to clarify the connection of our asypmtotics in \thmref{Joint_limit_of_Wnt_and_Vd} with the benchmark asymptotic results:
\begin{equation}
\begin{split}\label{asymptotic_connection}
    \by^{TO, \epsilon}(t)  &=y_M+  \left( \lim_{c \to \infty} a_1(c)\right)\cdot  \epsilon^{\frac{1}{3}} + o(\epsilon^{\frac{1}{3}}), \\
    \uy^{TO, \epsilon}(t) &=y_M -  \left( \lim_{c \to \infty} a_1(c)\right)\cdot \epsilon^{\frac{1}{3}} + o(\epsilon^{\frac{1}{3}}), \\
    v^{TO, \epsilon}(t,y_M) &= v^{0}(t)  - \left( \lim_{c \to \infty} a_2(c)\right)  v^{0}(t) (T - t) \cdot \epsilon^{\frac{2}{3}} + o(\epsilon^{\frac{2}{3}}),\\
    \hat y^{SO,\lambda}(t) &= y_M + \left(\lim_{c\to 0} \sqrt{c} \,a_1(c) \right) \cdot \tfrac{1}{\sqrt{\lambda}} + o(\tfrac{1}{\sqrt{\lambda}}),  \quad (\because\,  \lim_{c \to 0} \sqrt{c} \, a_1(c) =0)\\
v^{SO, \lambda}(t,y_M) &= v^{0}(t)  - \left( \lim_{c \to 0} c\, a_2(c)\right)  v^{0}(t) (T - t) \cdot \tfrac{1}{\lambda} + o(\tfrac{1}{\lambda}).
\end{split}
\end{equation}
Indeed, \thmref{Joint_limit_of_Wnt_and_Vd} bridges the benchmark asymptotics in \eqref{TO_benchmark} and \eqref{SO_benchmark} through the parametric relation $\lambda=c \, \epsilon^{-\frac{2}{3}}$, where the case $c=\infty$ corresponds to \eqref{TO_benchmark} and the case $c=0$ corresponds to \eqref{SO_benchmark}. In this sense, our approach of using $\lambda=c \, \epsilon^{-\frac{2}{3}}$ unifies the asymptotics for transaction costs and search frictions.

Section 5 in \cite{GC23} contains asymptotics for $\epsilon$, which differ from \thmref{Joint_limit_of_Wnt_and_Vd}: we let $\epsilon\downarrow 0$ and $\lambda\to \infty$ at the same time through the relation $\lambda=c \, \epsilon^{-\frac{2}{3}}$ in \thmref{Joint_limit_of_Wnt_and_Vd}, whereas $\lambda$ is fixed in Theorem 5.4 and Theorem 5.6 in \cite{GC23}. It is also worth noting that the correction terms in Theorem 5.4 and Theorem 5.6 in \cite{GC23} are not explicit (in terms of solutions of some PDEs), whereas $a_1(c)$ and $a_2(c)$ in \thmref{Joint_limit_of_Wnt_and_Vd} are explicit in terms of the model parameters. Therefore, given model parameters $\mu,\sigma,\gamma,\lambda,\epsilon$, one can compute the auxiliary parameter $c=\lambda \epsilon^{\frac{2}{3}}$ and use the formulas in  \thmref{Joint_limit_of_Wnt_and_Vd} to estimate the optimal trading strategy and value.

\medskip

\section{Impact of Search Frictions on Trading Decisions}

While transaction costs have been extensively studied in optimal trading models, the combined effect of transaction costs and search frictions remains less well understood. In many financial markets, investors cannot trade continuously but must wait for suitable liquidity events, leading to execution delays. These search frictions, modeled as trading opportunities arriving at random times, introduce an additional source of illiquidity that may meaningfully influence trading behavior.

Search frictions are present in a variety of financial markets. In highly liquid environments, such as major stock exchanges, trading opportunities arise frequently, and the impact of search frictions is likely to be minimal. However, in less liquid markets, trading is constrained by counterparty search and market frictions, leading to sporadic transactions. A notable example is the over-the-counter (OTC) market, where trading does not occur on centralized exchanges but rather through decentralized dealer networks. In OTC markets, the time required to execute a trade depends on the ability to find a suitable counterparty, potentially introducing substantial search frictions. For instance, \cite{Ang13} analyzes OTC stock market data and reports that the proportion of days with no trading volume ranges from 0.01 to 0.94 (see their Table 3). 

Similarly, Table 1 in \cite{Andrew14} highlights asset classes characterized by extended intervals between trades, including corporate and municipal bonds, emerging market equities, and private equity, where the typical time between transactions may range from several hours to several weeks. 

To account for these variations in search frictions, in the numerical examples below, we set $\lambda$ (interpreted as the annual expected number of trading opportunities) to range from $4$ to $500$: 
\begin{align}
\lambda=4, \, 20,\, 100,\, 500. \label{lambda_values}
\end{align}

We impose Assumption~\ref{ass} and Notation~\ref{notation} throughout this section. Our numerical experiments aim to evaluate whether search frictions materially impact trading decisions. If transaction costs alone primarily determine investor behavior, then models that exclude search frictions may provide a sufficiently accurate representation of trading strategies. However, if search frictions substantially alter the no-trade region, ignoring them could lead to significant misestimation of optimal trading behavior. To quantify the influence of search frictions relative to transaction costs, we introduce a measure based on the width of the no-trade region (WNTR), which characterizes the investor's optimal trading strategy. Specifically, we define:
\begin{equation}
\begin{split}\label{RR}
R(\lambda) :=&\frac{\textrm{``WNTR in a market with both search frictions and transaction costs"}}{\textrm{``WNTR in a market with only transaction costs and no search frictions"}} \\
&= \frac{\by^\epsilon(0)-\uy^\epsilon(0)}{\by^{TO,\epsilon}(0)-\uy^{TO,\epsilon}(0)}.
\end{split}
\end{equation}
The ratio $R(\lambda)$ serves as a measure of the extent to which search frictions affect the optimal trading strategy. When $R(\lambda)$ is close to $1$, the impact of search frictions is limited, and trading behavior is primarily dictated by transaction costs. Conversely, when $R(\lambda)$ deviates significantly from $1$, search frictions play a more substantial role in shaping the no-trade region, indicating that they should not be ignored in trading models.

\renewcommand{\arraystretch}{1.5}

\begin{table}[t]
    \centering
    \adjustbox{max width=\textwidth}{
{\tiny
    \begin{tabular}{ c  c  | c || c c c c || c c c c || c c c c }
    &   & & \multicolumn{4}{c||}{$\sigma=20\%$} &  \multicolumn{4}{c||}{$\sigma=50\%$} &  \multicolumn{4}{c}{$\sigma=100\%$} \\ 
    &   $\lambda$ & $c=\epsilon^{\frac{2}{3}} \lambda$  & $\by$ & $\uy$ & $\by-\uy$ & $R(\lambda)$  &  $\by$ & $\uy$  & $\by-\uy$  & $R(\lambda)$ & $\by$ & $\uy$  & $\by-\uy$  & $R(\lambda)$ \\  
        \hline \hline
   \multirow{10}{*}{$\epsilon=0.5\%$ \,\,  }   &  \multirow{2}{*}{4} & \multirow{2}{*}{0.12} &  $63.8\%$   &  $56.4\%$ & $7.4\%$ & $  0.71 $ &  $62.3\%$   &  $58.3\%$ & $4.0\%$ & $  0.38  $   &  $61.6\%$  &  $59.9\%$  &  $1.7\%$  & $0.17$  \\
     & &  &  $(63.6\%)$ & $(56.4\%)$ & $(7.2\%)$  & $(0.69)$  & $(61.8\%)$ & $(58.2\%)$ & $(3.6\%)$ & $(0.35)$  &  $(60.6\%)$  &  $(59.4\%)$  &  $(1.2\%)$  &  $(0.12)$ \\
       \cline{2-15}
   & \multirow{2}{*}{20}  & \multirow{2}{*}{0.58}  & $64.6\%$   &  $55.5\%$  &  $9.1\%$  &  $ 0.88 $  &   $63.6\%$  &  $56.6\%$  &  $7.0\%$  &  $ 0.67 $    &  $62.4\%$  &  $58.1\%$  &  $4.3\%$  &  $0.42$  \\
   &  &  &  $(64.5\%)$  &  $(55.5\%)$ &  $(9.0\%)$  &  $(0.86)$      & $(63.4\%)$  &  $(56.6\%)$  & $(6.9\%)$  &   $(0.66)$  &  $(62.0\%)$  &  $(58.0\%)$  &  $(4.1\%)$  &  $(0.40)$  \\   
         \cline{2-15} 
   & \multirow{2}{*}{100}  & \multirow{2}{*}{2.92} & $65.0\%$   &  $55.0\%$  &  $10.0\%$  &  $ 0.96 $  &  $64.5\%$  &  $55.6\%$  &  $8.9\%$  &  $ 0.86 $    &  $63.8\%$  &  $56.4\%$  &  $7.4\%$  &  $0.73$   \\
   &  & &  $(64.9\%)$  &  $(55.1\%)$  &  $(9.8\%)$  &  $(0.94)$      &  $(64.4\%)$  &  $(55.6\%)$  & $(8.8\%)$  &  $(0.85)$  &  $(63.6\%)$  &  $(56.4\%)$  &  $(7.2\%)$  &  $(0.71)$  \\   
         \cline{2-15} 
   & \multirow{2}{*}{500}  & \multirow{2}{*}{14.6} &  $65.1\%$   &  $54.9\%$  &  $10.2\%$  &  $ 0.98 $  &   $64.9\%$  &  $55.1\%$  &  $9.8\%$  &  $ 0.94 $    &  $64.6\%$  &  $55.5\%$  &  $9.1\%$  &  $0.89$   \\
   &  &   &  $(65.1\%)$  &  $(54.9\%)$  &  $(10.2\%)$  &  $(0.98)$      & $(64.9\%)$  &  $(55.1\%)$  & $(9.7\%)$  &   $(0.94)$  &  $(64.5\%)$  &  $(55.5\%)$  &  $(9.0\%)$  &  $(0.88)$  \\   
         \cline{2-15} 
  &  \multirow{2}{*}{$10^4$}  & \multirow{2}{*}{$292$}  &  $65.2\%$   & $54.8\%$  & $10.4\%$ &   $1$ &    $65.2\%$  & $54.8\%$  &  $10.4\%$ & $1$    &  $65.1\%$  &  $54.9\%$  &  $10.2\%$  &  $1$ \\
  &   &   &  $(65.2\%)$ &  $(54.8\%)$  &  $(10.4\%)$ &  $(1)$   &  $(65.2\%)$ & $(54.8\%)$ & $(10.3\%)$ & $(1)$  &  $(65.1\%)$  &  $(54.9\%)$  &  $(10.1\%)$  &  $(1)$ \\     
       \hline 
       \hline 
   \multirow{10}{*}{$\epsilon=1\%$ \,\,  }   &  \multirow{2}{*}{4} & \multirow{2}{*}{0.19}  &  $65.2\%$   &  $54.9\%$ & $10.3\%$ & $0.77$ &  $63.5\%$   &  $57.1\%$ & $6.4\%$ & $ 0.48 $   &  $62.3\%$  &  $59.2\%$  &  $3.1\%$  & $0.24$  \\
     & &  &  $(64.9\%)$ & $(55.1\%)$ & $(9.9\%)$  & $(0.75)$  & $(62.9\%)$ & $(57.1\%)$ & $(5.8\%)$ & $(0.44)$  &  $(61.2\%)$  &  $(58.8\%)$  &  $(2.3\%)$  &  $(0.18)$ \\
       \cline{2-15}
   & \multirow{2}{*}{20}  & \multirow{2}{*}{0.93}  & $66.0\%$   &  $54.0\%$  &  $12.0\%$  &  $0.90$  &   $65.0\%$  &  $55.2\%$  &  $9.8\%$  &  $ 0.74 $    &  $63.7\%$  &  $56.8\%$  &  $6.9\%$  &  $0.53$  \\
   &  &  &  $(65.8\%)$  &  $(54.2\%)$ &  $(11.7\%)$  &  $(0.89)$      & $(64.8\%)$  &  $(55.2\%)$  & $(9.5\%)$  &   $(0.73)$  &  $(63.2\%)$  &  $(56.8\%)$  &  $(6.4\%)$  &  $(0.50)$  \\   
         \cline{2-15} 
   & \multirow{2}{*}{100}  & \multirow{2}{*}{4.64}  & $66.4\%$   &  $53.6\%$  &  $12.8\%$  &  $0.96$  &  $65.9\%$  &  $54.1\%$  &  $11.8\%$  &  $  0.89  $    &  $65.2\%$  &  $55.0\%$  &  $10.2\%$  &  $0.78$   \\
   &  & &  $(66.3\%)$  &  $(53.7\%)$  &  $(12.5\%)$  &  $(0.95)$      &  $(65.8\%)$  &  $(54.2\%)$  & $(11.5\%)$  &  $(0.88)$  &  $(64.9\%)$  &  $(55.1\%)$  &  $(9.9\%)$  &  $(0.77)$  \\   
         \cline{2-15} 
   & \multirow{2}{*}{500}  & \multirow{2}{*}{23.2}  &  $66.6\%$   &  $53.4\%$  &  $13.2\%$  &  $0.99$  &   $66.4\%$  &  $53.7\%$  &  $12.7\%$  &  $ 0.96 $    &  $66.0\%$  &  $54.0\%$  &  $12.0\%$  &  $0.92$   \\
   &  &   &  $(66.5\%)$  &  $(53.5\%)$  &  $(12.9\%)$  &  $(0.98)$      & $(66.2\%)$  &  $(53.8\%)$  & $(12.4\%)$  &   $(0.95)$  &  $(65.8\%)$  &  $(54.2\%)$  &  $(11.7\%)$  &  $(0.91)$  \\   
         \cline{2-15} 
  &  \multirow{2}{*}{$10^4$}  & \multirow{2}{*}{$464$}   &  $66.7\%$   & $53.3\%$  & $13.4\%$ &   $1$ &    $66.6\%$  & $53.4\%$  &  $13.2\%$ & $1$    &  $66.5\%$  &  $53.4\%$  &  $13.1\%$  &  $1$ \\
  &   &   &  $(66.6\%)$ &  $(53.4\%)$  &  $(13.1\%)$ &  $(1)$   &  $(66.5\%)$ & $(53.5\%)$ & $(13.0\%)$ & $(1)$  &  $(66.4\%)$  &  $(53.6\%)$  &  $(12.9\%)$  &  $(1)$ \\     
       \hline 
       \hline        
   \multirow{10}{*}{$\epsilon=3\%$ \,\,  }   &  \multirow{2}{*}{4} &  \multirow{2}{*}{0.39} &  $70.3\%$   &  $49.9\%$ & $20.4\%$ & $0.89$ &  $66.5\%$   &  $54.2\%$ & $12.3\%$ & $0.63 $   &  $64.7\%$  &  $57.0\%$  &  $7.7\%$  & $0.40$  \\
     & &  &  $(67.8\%)$ & $(52.2\%)$ & $(15.7\%)$  & $(0.83)$  & $(65.6\%)$ & $(54.4\%)$ & $(11.1\%)$ & $(0.59)$  &  $(62.9\%)$  &  $(57.1\%)$  &  $(5.8\%)$  &  $(0.31)$ \\
       \cline{2-15}
   & \multirow{2}{*}{20}  &   \multirow{2}{*}{1.93} & $70.9\%$   &  $49.2\%$  &  $21.7\%$  &  $0.95$  &   $68.2\%$  &  $52.0\%$  &  $16.2\%$  &  $0.83$    &  $66.8\%$  &  $53.8\%$  &  $13.0\%$  &  $0.67$  \\
   &  &  &  $(68.8\%)$  &  $(51.2\%)$ &  $(17.5\%)$  &  $(0.92)$      & $(67.7\%)$  &  $(52.3\%)$  & $(15.3\%)$  &   $(0.81)$  &  $(65.9\%)$  &  $(54.1\%)$  &  $(11.9\%)$  &  $(0.63)$  \\   
         \cline{2-15} 
   & \multirow{2}{*}{100}  &  \multirow{2}{*}{9.65} & $71.2\%$   &  $48.8\%$  &  $22.4\%$  &  $0.98$  &  $69.1\%$  &  $51.0\%$  &  $18.1\%$  &  $0.93$    &  $68.4\%$  &  $51.8\%$  &  $16.6\%$  &  $0.86$   \\
   &  & &  $(69.2\%)$  &  $(50.8\%)$  &  $(18.4\%)$  &  $(0.97)$      &  $(68.7\%)$  &  $(51.3\%)$  & $(17.4\%)$  &  $(0.92)$  &  $(67.8\%)$  &  $(52.2\%)$  &  $(15.7\%)$  &  $(0.84)$  \\   
         \cline{2-15} 
   & \multirow{2}{*}{500}  &  \multirow{2}{*}{48.3} &  $71.4\%$   &  $48.6\%$  &  $22.8\%$  &  $1.00$  &   $69.5\%$  &  $50.5\%$  &  $19.0\%$  &  $0.97$    &  $69.2\%$  &  $50.9\%$  &  $18.3\%$  &  $0.94$   \\
   &  &   &  $(69.4\%)$  &  $(50.6\%)$  &  $(18.7\%)$  &  $(0.99)$      & $(69.1\%)$  &  $(50.9\%)$  & $(18.3\%)$  &   $(0.97)$  &  $(68.8\%)$  &  $(51.2\%)$  &  $(17.5\%)$  &  $(0.94)$  \\   
         \cline{2-15} 
  &  \multirow{2}{*}{$10^4$}  & \multirow{2}{*}{$965$}   &  $71.4\%$   & $48.5\%$  & $22.9\%$ &   $1$ &    $69.7\%$  & $50.2\%$  &  $19.5\%$ & $1$    &  $69.7\%$  &  $50.3\%$  &  $19.4\%$  &  $1$ \\
  &   &   &  $(69.5\%)$ &  $(50.5\%)$  &  $(19.0\%)$ &  $(1)$   &  $(69.4\%)$ &  $(50.6\%)$  &  $(18.7\%)$ &  $(1)$  &   $(69.4\%)$ &  $(50.6\%)$  &  $(18.7\%)$ &  $(1)$ \\     
       \hline 
       \hline        
    \end{tabular}}
}      
    \caption{
   The parameters are as in \eqref{lambda_values}, \eqref{sigma_values}, and \eqref{other_parameters}. 
The values for $\by=\by^\epsilon(0)$ and $\uy=\uy^\epsilon(0)$ are obtained by numerically solving the PDE in \eqref{Main_PDE}, with values in parentheses corresponding to the asymptotic approximations from \eqref{asymptotics_ntr}. To compute $R(\lambda)$ in \eqref{RR}, we approximate $\by^{TO,\epsilon}(0)$ and $\uy^{TO,\epsilon}(0)$ using $\by^\epsilon(0)$ and $\uy^\epsilon(0)$ for $\lambda=10^4$. The values in parentheses in the column for $R(\lambda)$ are computed using the corresponding asymptotic approximations from \eqref{asymptotics_ntr}.
     }
    \label{table1}
\end{table}

Table~\ref{table1} evaluates the impact of search frictions across different market environments by computing the boundaries of the no-trade region ($\by^\epsilon(0)$ and $\uy^\epsilon(0)$), WNTR, and $R(\lambda)$ for varying levels of annualized volatility $\sigma$ and transaction costs $\epsilon$.\footnote{In our model, the proportional transaction cost $\epsilon$ applies symmetrically to both buying and selling. Consequently, our setup is comparable to the model in \cite{Guasoni}, where transaction costs are imposed only on sales, but with adjusted cost level of $2\epsilon$. In this sense, $\epsilon=3\%$ in our model corresponds to a one-way transaction cost of $6\%$ in \cite{Guasoni}, which is within the range observed in real estate markets.}
The values of $\sigma$ and $\epsilon$ are specified as follows:
\begin{align}
\sigma=20\%, \, 50\%, \, 100\% \quad \textrm{and} \quad \epsilon=0.5\%,\, 1\%, 3\%.
\label{sigma_values}
\end{align}
The remaining market parameters are set as follows: 
\begin{align}
T=3 \textrm{ years}, \quad \gamma=3, \quad y_M=60\%, \quad \mu=\gamma \sigma^2 y_M.
\label{other_parameters}
\end{align}

The values for $\by=\by^\epsilon(0)$ and $\uy=\uy^\epsilon(0)$ in Table~\ref{table1} are obtained by numerically solving the PDE in \eqref{Main_PDE}, with values in parentheses corresponding to the asymptotic approximations from \eqref{asymptotics_ntr}. To compute $R(\lambda)$ in \eqref{RR}, we approximate $\by^{TO,\epsilon}(0)$ and $\uy^{TO,\epsilon}(0)$ using $\by^\epsilon(0)$ and $\uy^\epsilon(0)$ for a sufficiently large $\lambda$, setting $\lambda=10^4$. This choice is justified by the fact that $\by^{TO,\epsilon}(0)$ and $\uy^{TO,\epsilon}(0)$ in \eqref{TO_benchmark} represent the boundaries of the no-trade region in the absence of search frictions. The values in parentheses in the column for $R(\lambda)$ are computed using the corresponding asymptotic approximations from \eqref{asymptotics_ntr}.

Table~\ref{table1} reveals several key findings from the numerical analysis. First, WNTR increases with $\lambda$. Transaction costs discourage frequent trading, while search frictions limit the timing of trades.
When $\lambda$ is low, trading opportunities are infrequent, requiring the investor to make large portfolio adjustments whenever a trade is possible. This urgency results in a smaller WNTR. As $\lambda$ increases, the investor anticipates more frequent trading opportunities, reducing the need for aggressive rebalancing. Consequently, WNTR expands, reflecting greater patience in trading executing trades. This observation aligns with the finding that $R(\lambda)$ in Table~\ref{table1} increases with $\lambda$ and remains bounded above by $1$. 

Second, WNTR decreases with $\sigma$. Higher volatility increases the risk of large portfolio deviations. When trading opportunities are limited, a higher $\sigma$ amplifies the risk of the portfolio deviating significantly from the Merton fraction $y_M$ due to extended waiting times before the next trading opportunity. In response, the investor reduces WNTR, opting to rebalance more aggressively to maintain a position closer to $y_M$ when an opportunity arises. 

Third, $R(\lambda)$ decreases with $\sigma$. The reasoning behind WNTR decreasing with $\sigma$ applies here as well. Since $R(\lambda)$ measures the extent to which search frictions affect the optimal trading strategy and is bounded above by $1$ in Table~\ref{table1}, a smaller $R(\lambda)$ implies a greater effect of search frictions. Thus, Table~\ref{table1} suggests that the ``narrowing effect" of search frictions on WNTR is more pronounced for larger $\sigma$.

Lastly, $R(\lambda)$ tends to increase with $c$ for fixed $\sigma$. For instance, when $\sigma=100\%$, some of the $(c,R(\lambda))$ pairs in Table~\ref{table1} are $(0.12, 0.17), \, (0.19,0.24),\, (0.39, 0.40), \, (0.58, 0.42), \, (0.93,0.53)$, etc.
We also observe that the values obtained from the asymptotics in \eqref{TO_benchmark} and \eqref{asymptotics_ntr} closely approximate those obtained by numerically solving the PDE in  \eqref{Main_PDE}, particularly for large $\lambda$ and small $\epsilon$.

\begin{figure}[t]
		\begin{center}
			\includegraphics[width=0.45\textwidth]{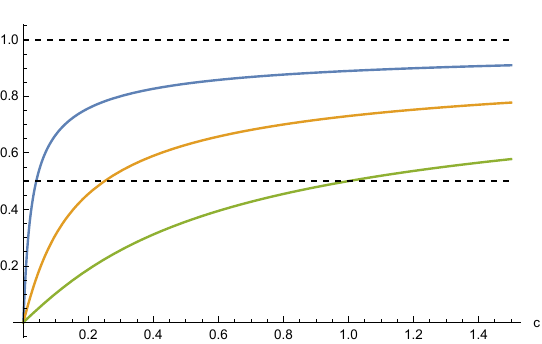}
	\end{center}
	 \caption{The graph shows the expression in \eqref{R_approx} as functions of $c$ for varying $\sigma$ (\textcolor{myblue}{blue}, \textcolor{myorange}{orange}, \textcolor{mygreen}{green} lines for $\sigma=0.2,\, 0.5,\, 1$). The common parameters are $y_M=0.6$ and $\gamma=3$. 	 
			}
		\label{figure_R}
\end{figure}

The asymptotics \eqref{TO_benchmark} and \eqref{asymptotics_ntr} provide a natural approximation for $R(\lambda)$:
  \begin{align}
R(\lambda) \approx  \frac{2\cdot \tfrac{ \sigma y_M (1 - y_M)}{ \sqrt{2\,c}} \Big( \Big( \tfrac{3 \sqrt{2} \,c^{\frac{3}{2}}}{\gamma \sigma^{3} y_M (1 - y_M)} + 1 \Big)^{\frac{1}{3}} - 1 \Big) }{ \Big( \tfrac{12y_M^2(1-y_M)^2}{\gamma}\Big)^{\frac{1}{3}} }.
\label{R_approx}
\end{align}

Figure~\ref{figure_R} illustrates the approximated values of $R(\lambda)$ in \eqref{R_approx} as a function of $c$ for the varying levels of $\sigma$. 
The approximation suggests that $R(\lambda)$ increases with $c$ and converges to 1, consistent with  \eqref{a1_limit}. This result provides a quantitative measure of the impact of search frictions on WNTR.
For instance, the figure indicates that WNTR is reduced to less than half due to search frictions when approximately $c=\epsilon^{\frac{2}{3}} \lambda<0.04$ for $\sigma=20\%$, $c<0.25$ for $\sigma=50\%$, and $c<1$ for $\sigma=100\%$.

\medskip

\section{Proof of \thmref{Joint_limit_of_Wnt_and_Vd}}\label{proof_main}

We prove \thmref{Joint_limit_of_Wnt_and_Vd} in this section, starting with some technical lemmas. We use Notation~\ref{notation}.
 Under \assref{ass}, by \proref{Property_of_NT}, there are $\underline{t}^{\epsilon}, \overline{t}^{\epsilon} \in [0,T)$ such that
\begin{align}
\begin{cases}
0<\underline{y}^{\epsilon}(t)<1 &\textrm{if} \quad t\in [0,\underline{t}^{\epsilon})\\
\underline{y}^{\epsilon}(t)=0 &\textrm{if} \quad t\in [\underline{t}^{\epsilon},T)
\end{cases}, \quad
\begin{cases}
0<\overline{y}^{\epsilon}(t)<1 &\textrm{if} \quad t\in [0,\overline{t}^{\epsilon})\\
\overline{y}^{\epsilon}(t)=1 &\textrm{if} \quad t\in [\overline{t}^{\epsilon},T)
\end{cases}. \label{ts}
\end{align}

The following lemma estimates the location of $\underline{t}^{\epsilon}$ and $\overline{t}^{\epsilon}$.

\begin{lemma}\label{NT_boundaries_hitting_times_bound}
Let \assref{ass} hold. Then, there are $\epsilon_0>0$ and $\overline{C}> \underline{C} > 0$ such that 
\begin{align*}
\underline{t}^{\epsilon}, \, \overline{t}^{\epsilon}\in  [ T - \overline{C} \epsilon, \, T - \underline{C} \epsilon] \quad \textrm{for} \quad \epsilon \in (0,\epsilon_0].
\end{align*}
\end{lemma}

\begin{proof}
We prove the inequalities for $\underline{t}^{\epsilon}$ ($\overline{t}^{\epsilon}$ can be treated by the same way).
Considering the lower bound of $v^\epsilon$ in \eqref{Range_of_v}, for small enough $\epsilon$, there exists $\underline{t}^{\epsilon} \in [0,T)$ that satisfies \eqref{ut_eq}. Hence, 
\begin{align*}
  \epsilon \, e^{\mu T} & = \mu (1 + \epsilon) \int_{\underline{t}^{\epsilon}}^{T}  e^{(\mu + \lambda) s} \Big( e^{- \lambda T} + \lambda \int_{s}^{T} e^{- \lambda u} \tfrac{v^{\epsilon}(u, \uy^{\epsilon}(u))}{( 1 + \epsilon \uy^{\epsilon}(u) )^{1 - \gamma}} d u \Big) d s  \\
  & \leq C \mu (1 + \epsilon) \int_{\underline{t}^{\epsilon}}^{T} e^{(\mu + \lambda) s} \left( e^{- \lambda T} + \lambda \int_{s}^{T} e^{- \lambda u} d u \right) d s = C(1+\epsilon) \big( e^{\mu T} - e^{\mu \underline{t}^\epsilon} \big), 
\end{align*}
where $C>0$ is a constant independent of $\epsilon$. This implies that there exists $\epsilon_0>0$ such that
\begin{align*}
\underline{t}^{\epsilon} & \leq T + \tfrac{1}{\mu}\ln \left( 1 - \tfrac{\epsilon}{C (1 + \epsilon)} \right) \leq T- \overline C \epsilon \quad \textrm{for}\quad \epsilon \in (0,\epsilon_0],
\end{align*}
where $\overline{C}>0$ is a constant independent of $\epsilon$ and the second inequality is due to $\lim_{x \downarrow 0} \tfrac{\ln (1 - x)}{x} = - 1$.

Similarly, for small enough $\epsilon$, \eqref{ut_eq} produces 
\begin{align*}
   \epsilon \, e^{\mu T}   & \geq \mu (1 + \epsilon) \int_{\underline{t}^{\epsilon}}^{T} e^{(\mu + \lambda) s} e^{- \lambda T} d s = \tfrac{\mu(1+\epsilon)}{\mu+\lambda} e^{-\lambda T} \big(e^{(\mu+\lambda)T} - e^{(\mu+\lambda) \underline{t}^\epsilon} \big).
\end{align*}
This implies that there exists $\epsilon_0>0$ and $\underline{C}>0$ such that
\begin{align*}
\underline{t}^{\epsilon,\lambda} & \geq T + \tfrac{1}{\mu+\lambda} \ln \left( 1 - \tfrac{(\mu + \lambda) \epsilon}{\mu (1 + \epsilon)} \right) \geq T-\underline{C} \epsilon \quad \textrm{for}\quad \epsilon\in (0,\epsilon_0],
  \end{align*}
where the second inequality is due to $\lim_{x \downarrow 0} \tfrac{\ln (1 - x)}{x} = - 1$ and the relation $\lambda=c \, \epsilon^{-\frac{2}{3}}$. 
\end{proof}

Before we start to prove \thmref{Joint_limit_of_Wnt_and_Vd}, we list three technical lemmas whose proofs are provided in the appendices.

\begin{lemma}\label{Merton_fraction_inside_NT}
Let \assref{ass} hold.

\noindent (i) There exist a constant $C>0$ independent of $(t,\epsilon)\in [0,T)\times (0,1)$ such that
\begin{align}
&\big\vert \uy^{\epsilon}(t) -y_M \big\vert, \, \big\vert \by^{\epsilon}(t) -y_M \big\vert, \, 
\big\vert \by^{\epsilon}(t) -\uy^{\epsilon}(t) \big\vert \leq \tfrac{C \epsilon^{\frac{1}{3}} }{\min \{ 1, \lambda (T - t) \}}. \label{pre_y1}
\end{align}

\noindent(ii) For fixed $t\in [0,T)$, $\li_{\epsilon \downarrow 0} \tfrac{ \by^{\epsilon}(t) - \uy^{\epsilon}(t)}{\epsilon^{\frac{1}{3}}}>0$.

\noindent(iii) There is $\epsilon_0>0$ such that for $(t,\epsilon)\in [0,T]\times (0,\epsilon_0]$, we have $\uy^{\epsilon}(t)<y_M<\by^{\epsilon}(t)$.
\end{lemma}
\begin{proof}
See Appendix~\ref{lemma 5.6}.
\end{proof}

\begin{lemma}\label{really_used}
Let \assref{ass} hold.

\noindent(i) For $(t,\epsilon, x)\in [0,T)\times (0,1)\times [ \uy^{\epsilon}(t), \by^{\epsilon}(t)]$, there is $C>0$ independent of $(t,\epsilon, x)$ such that
\begin{align}
&\left| v_x^{\epsilon}(t, x) \right|  \leq C \epsilon,\label{vx_epsilon}\\
&\left| v^{\epsilon}(t, x) - v^{0}(t)  \right|  \leq C \epsilon^{\frac{2}{3}}, \label{v_epsilon}\\
&\left\vert v_{t}^{\epsilon}(t, x) - v_{t}^{0}(t) \right\vert \leq C \big( \epsilon^{\frac{2}{3}}+ (x-y_M)^2 \big). \label{vt_ep_vt_0_difference_on_NT} 
\end{align}

\noindent(ii) For fixed $t\in [0,T)$,
\begin{align}
&\lim_{\epsilon \downarrow 0}\bigg(\sup_{x_1,x_2 \in  [ \uy^{\epsilon}(t), \by^{\epsilon}(t)]} \Big| \tfrac{v^{\epsilon}(t, x_1)-v^{\epsilon}(t, x_2)}{ \epsilon^{\frac{2}{3}}} \Big| \bigg) =0, \label{Uniform_limit_of_v_ep_v_0_difference}\\
&\lim_{\epsilon \downarrow 0} \bigg(\sup_{x_1,x_2 \in  [ \uy^{\epsilon}(t), \by^{\epsilon}(t)]}\Big| \tfrac{v_t^{\epsilon}(t, x_1)- v_t^{\epsilon}(t, x_2)}{ \epsilon^{\frac{2}{3}}}\Big| \bigg)=0.\label{Uniform_limit_of_vt_ep_vt_0_difference}
\end{align}
\end{lemma}
\begin{proof}
See Appendix~\ref{lemma 5.7}.
\end{proof}

\begin{lemma}\label{v_xx_conv_lem}
Let \assref{ass} hold and $G^{\epsilon} (t, x):= x^{2} (1 - x)^{1 + \gamma} \tfrac{\lambda v_{x x}^{\epsilon}(t, x)}{1 - \gamma}$.
Let $t\in [0,T)$ be fixed, $h(z):=\tfrac{e^z}{1+e^z}$, and $x_\epsilon\in [ \uy^{\epsilon}(t), \by^{\epsilon}(t)]$. Then,
\begin{align*}
G^{\epsilon} (t, x_\epsilon)  + y_M^2(1-y_M)^{1+\gamma} \gamma \sigma^2 v^{0}(t) 
- \int_{\uz^{\epsilon}(t)}^{\bz^{\epsilon}(t)} G^{\epsilon} (t,h(z)) \tfrac{\sqrt{2 \lambda} }{2 \sigma}\,e^{- \frac{\sqrt{2 \lambda}}{\sigma} \left\vert z - z_\epsilon \right\vert}  dz \stackrel{\epsilon \downarrow 0}{\longrightarrow}0,
\end{align*}
where $\uz^{\epsilon}(t):=h^{-1}(\uy^{\epsilon}(t))$, $\bz^{\epsilon}(t):=h^{-1}(\by^{\epsilon}(t))$ and $z_\epsilon:= h^{-1}(x_\epsilon)$.
\end{lemma}
\begin{proof}
See Appendix~\ref{lemma 5.8}.
\end{proof}

\medskip

We now proceed to prove \thmref{Joint_limit_of_Wnt_and_Vd}. The proof consists of three steps. 
Throughout this proof, $C>0$ is a generic constant independent of $(t,s, x,\epsilon)\in [0,T)\times [t,T)\times(0,1) \times (0,1)$ (also independent of $\lambda$ due to relation $\lambda=c \, \epsilon^{-\frac{2}{3}}$) that may differ line by line. 

Let $t\in [0,T)$ be fixed. Since $0<y_M<1$, \lemref{Merton_fraction_inside_NT} implies that for small enough $\epsilon>0$,
$$0 < \uy^{\epsilon}(t) < y_M < \by^{\epsilon}(t) < 1.$$
In the end, we are interested in the limiting behaviors as $\epsilon\downarrow 0$. Hence, we assume that $\epsilon>0$ is small enough and the above inequalities hold.

For $x\in [ \uy^{\epsilon}(t), \by^{\epsilon}(t)]$,
using $L^{\epsilon}(t,x)=v^{\epsilon}(t,x)$ and \eqref{merton_pde}, we rewrite \eqref{Main_PDE} as
\begin{align}
  0  & =   v_t^{\epsilon}(t, x) - v_t^{0}(t)  + Q(x) \big( v^{\epsilon}(t, x) - v^{0}(t) \big) + v^{0}(t)(Q(x) - Q(y_M)) \nonumber \\
  & \quad + x (1 - x) \left( \mu - \gamma \sigma^{2} x \right) v_{x}^{\epsilon}(t, x) + \tfrac{\sigma^{2}x^{2} (1 - x)^{2}}{2}  v_{x x}^{\epsilon}(t, x) \quad \textrm{for} \quad  x\in [ \uy^{\epsilon}(t), \by^{\epsilon}(t)]. \label{PDE_on_NT_region}
\end{align}

{\bf Step 1.}  We define $I^{\epsilon}(t, x)$ as 
\begin{align*}
  I^{\epsilon}(t, x) & := G^{\epsilon} (t, x) - \gamma (1 - y_M)^{\gamma - 1} v^{0}(t) \lambda (x - y_M)^{2},
\end{align*}
where $G^{\epsilon}$ is as in \lemref{v_xx_conv_lem}.
We multiply $\tfrac{\lambda}{1-\gamma}$ to \eqref{PDE_on_NT_region} and obtain
\begin{align}
  0 & =   \tfrac{  \lambda (v_t^{\epsilon}(t, x) - v_t^{0}(t) )}{1 - \gamma} + Q(x) \tfrac{\lambda ( v^{\epsilon}(t, x) - v^{0}(t))}{1 - \gamma}  + \tfrac{\gamma \sigma^{2}}{2} \big( \big( \tfrac{1 - x}{1 - y_M} \big)^{1 - \gamma} - 1 \big) v^{0}(t) \lambda ( x - y_M)^{2} \nonumber \\
  & \quad + x (1 - x) ( \mu - \gamma \sigma^{2} x ) \tfrac{\lambda v_{x}^{\epsilon}(t, x)}{1 - \gamma} + \tfrac{\sigma^{2}}{2} (1 - x)^{1 - \gamma} I^{\epsilon}(t, x)\quad \textrm{for} \quad  x\in [ \uy^{\epsilon}(t), \by^{\epsilon}(t)]. \label{Uniform_limit_of_U}
\end{align}
The above equality with $x=y_M$, together with \lemref{Merton_fraction_inside_NT} (i) and \lemref{really_used}, produces
\begin{align}
& \tfrac{  \lambda (v_t^{\epsilon}(t, y_M) - v_t^{0}(t) )}{1 - \gamma} + Q(y_M) \tfrac{\lambda ( v^{\epsilon}(t,y_M) - v^{0}(t))}{1 - \gamma}  + \tfrac{\sigma^{2}}{2} (1 - y_M)^{1 - \gamma} I^{\epsilon}(t, y_M)\stackrel{\epsilon\downarrow 0 } {\longrightarrow} 0,\label{lim1}
 \\
&  \sup_{x_1,x_2 \in  [ \uy^{\epsilon}(t), \by^{\epsilon}(t)]} \left|  I^{\epsilon} \left( t, x_1 \right) - I^{\epsilon} \left( t, x_2 \right) \right| \stackrel{\epsilon\downarrow 0 } {\longrightarrow} 0.\label{U_Uniform_limit}
\end{align}
Let $h, x_\epsilon, z_\epsilon, \uz^{\epsilon}(t),\bz^{\epsilon}(t)$ be as in \lemref{v_xx_conv_lem}. 
Note that $\uy^{\epsilon}(t) \leq x_\epsilon \leq \by^{\epsilon}(t)$ and $\uz^{\epsilon}(t) \leq z_\epsilon \leq \bz^{\epsilon}(t)$.
Since $z\in [ \uz^{\epsilon}(t),\bz^{\epsilon}(t)]$ is equivalent to $h(z)\in [ \uy^{\epsilon}(t),\by^{\epsilon}(t)]$, the convergence in \eqref{U_Uniform_limit} produces
\begin{align}
  &  \int_{\uz^{\epsilon}(t)}^{\bz^{\epsilon}(t)} \left\vert I^{\epsilon} ( t, h(z)) - I^{\epsilon} (t,x_\epsilon) \right\vert \tfrac{\sqrt{2 \lambda}}{2 \sigma} e^{- \frac{\sqrt{2 \lambda}}{\sigma} \left\vert z - z_\epsilon \right\vert} d z \nonumber\\
  & \leq \sup_{x_1,x_2 \in  [ \uy^{\epsilon}(t), \by^{\epsilon}(t)]} \big|  I^{\epsilon} \left( t, x_1 \right) - I^{\epsilon} \left( t, x_2 \right) \big| \cdot \int_{\uz^{\epsilon}(t)}^{\bz^{\epsilon}(t)} \tfrac{\sqrt{2 \lambda}}{2 \sigma} e^{- \frac{\sqrt{2 \lambda}}{\sigma} \left\vert z - z_\epsilon \right\vert} d z \stackrel{\epsilon\downarrow 0 } {\longrightarrow} 0, \label{I1}
\end{align}
where we also use the following observation for the convergence part above:
\begin{align}
\int_{\uz^{\epsilon}(t)}^{\bz^{\epsilon}(t)} \tfrac{\sqrt{2 \lambda}}{2 \sigma} e^{- \frac{\sqrt{2 \lambda}}{\sigma} \left\vert z -z_\epsilon \right\vert} d z = 1- \tfrac{1}{2}\left( e^{- \frac{\sqrt{2 \lambda}}{\sigma} \left(z_\epsilon-\uz^{\epsilon}(t)\right)} + e^{- \frac{\sqrt{2 \lambda}}{\sigma} \left(\bz^{\epsilon}(t)-z_\epsilon\right)}  \right) <1. \label{I2}
\end{align}
We combine \lemref{v_xx_conv_lem} and \eqref{I1} to obtain
\begin{align}
&I^{\epsilon} (t, x_\epsilon) + \gamma(1-y_M)^{\gamma-1}v^{0}(t) \lambda (x_\epsilon-y_M)^2 + y_M^2(1-y_M)^{1+\gamma} \gamma \sigma^2 v^{0}(t) \nonumber \\
&- \int_{\uz^{\epsilon}(t)}^{\bz^{\epsilon}(t)}\Big( I^{\epsilon}(t,x_\epsilon) +  \gamma (1 - y_M)^{\gamma - 1} v^{0}(t) \lambda (h(z) - y_M)^{2} \Big) \tfrac{\sqrt{2 \lambda} }{2 \sigma}\,e^{- \frac{\sqrt{2 \lambda}}{\sigma} \left\vert z - z_\epsilon \right\vert}  dz 
\stackrel{\epsilon \downarrow 0}{\longrightarrow}0. \label{I3}
\end{align}
Using the explicit form of the integral in \eqref{I2}, the above convergence can be written as
\begin{align}
&\tfrac{1}{2}\left( e^{- \frac{\sqrt{2 \lambda}}{\sigma} \left(z_\epsilon -\uz^{\epsilon}(t)\right)} + e^{- \frac{\sqrt{2 \lambda}}{\sigma} \left(\bz^{\epsilon}(t)-z_\epsilon \right)}  \right)I^{\epsilon} (t, x_\epsilon)  + \gamma(1-y_M)^{\gamma-1}v^{0}(t) J^{\epsilon}(t,x_\epsilon)
\stackrel{\epsilon \downarrow 0}{\longrightarrow}0, \label{I4}
\end{align}
where $J^{\epsilon}(t,x)$ is defined as
\begin{align}
J^{\epsilon}(t,x):= \lambda (x -y_M)^2 - \int_{\uz^{\epsilon}(t)}^{\bz^{\epsilon}(t)} \lambda (h(z) - y_M)^{2}  \tfrac{\sqrt{2 \lambda} }{2 \sigma}\,e^{- \frac{\sqrt{2 \lambda}}{\sigma} \left\vert z - h^{-1}(x) \right\vert}  dz + \sigma^2 y_M^2(1-y_M)^2. \label{J_def}
\end{align}
In \eqref{I4}, we substitute $x_\epsilon =\by^{\epsilon}(t)$ and $x_\epsilon=\uy^{\epsilon}(t)$, then subtract the resulting expressions to obtain the following equation:   
\begin{align}
&\tfrac{1}{2}\left( e^{- \frac{\sqrt{2 \lambda}}{\sigma} \left(\bz^{\epsilon}(t)-\uz^{\epsilon}(t)\right)} +1\right) \left( I^{\epsilon} (t, \by^{\epsilon}(t))- I^{\epsilon} (t, \uy^{\epsilon}(t))\right)\nonumber\\
& +\gamma(1-y_M)^{\gamma-1}v^{0}(t) \left( J^{\epsilon} (t, \by^{\epsilon}(t))- J^{\epsilon} (t, \uy^{\epsilon}(t))\right) \stackrel{\epsilon \downarrow 0}{\longrightarrow}0. \label{I5}
\end{align}
We combine \eqref{U_Uniform_limit} and \eqref{I5} and conclude that
\begin{align}
J^{\epsilon} (t, \by^{\epsilon}(t))- J^{\epsilon} (t, \uy^{\epsilon}(t)) \stackrel{\epsilon \downarrow 0}{\longrightarrow}0. \label{J0}
\end{align}
Let $z_M:=h^{-1}(y_M)$.
Observe that \eqref{pre_y1}, $\by^{\epsilon}(t)=h(\bz^{\epsilon}(t))$ and $\uy^{\epsilon}(t)=h(\uz^{\epsilon}(t))$ imply
\begin{align}
\left|\bz^{\epsilon}(t) - \uz^{\epsilon}(t) \right| \leq \tfrac{C \epsilon^{\frac{1}{3}} }{\min \{ 1, \lambda (T - t) \}}. \label{pre_z1}
\end{align}
Then, $h(z_M)=y_M$, $h'(z)=h(z)(1-h(z))$, $\lambda=c \, \epsilon^{-\frac{2}{3}}$ and \eqref{pre_z1} imply
\begin{align}
& \sup_{z_1, z_2\in [ \uz^{\epsilon}(t),\bz^{\epsilon}(t)]} \sqrt{\lambda}\,\Big| h(z_1) - h(z_2) -   y_M(1-y_M) \left( z_1 - z_2 \right) \Big| \stackrel{\epsilon \downarrow 0}{\longrightarrow} 0 , \label{J1}\\
& \sup_{z \in [ \uz^{\epsilon}(t),\bz^{\epsilon}(t)]} \lambda \, \Big| ( h(z)- y_M )^2 -y_M^2(1-y_M)^2 (z - z_M )^2 \Big| \stackrel{\epsilon \downarrow 0}{\longrightarrow} 0. \label{J2}
\end{align}
The limit in \eqref{J2} and the bound in \eqref{I2} produce
\begin{align}
& \int_{\uz^{\epsilon}(t)}^{\bz^{\epsilon}(t)} \lambda \Big((h(z) - y_M)^{2}- y_M^2(1-y_M)^2  (z- z_M )^2 \Big) \tfrac{\sqrt{2 \lambda} }{2 \sigma}\,e^{- \frac{\sqrt{2 \lambda}}{\sigma} \left\vert z - z_\epsilon \right\vert}  dz 
 \stackrel{\epsilon \downarrow 0}{\longrightarrow}0. \label{J3}
\end{align}
From \eqref{J0}, we obtain the following:
\begin{align}
0&=\lim_{\epsilon \downarrow 0} \left( J^{\epsilon} (t, \by^{\epsilon}(t))- J^{\epsilon} (t, \uy^{\epsilon}(t)) \right) \nonumber\\
&=\lim_{\epsilon \downarrow 0} \Big( \lambda(\by^{\epsilon}(t)-y_M)^2 -\lambda(\uy^{\epsilon}(t)-y_M)^2 \nonumber \\
&\qquad \qquad -   y_M^2(1-y_M)^2 \int_{\uz^{\epsilon}(t)}^{\bz^{\epsilon}(t)} \lambda (z- z_M )^2  \tfrac{\sqrt{2 \lambda} }{2 \sigma} \left( e^{- \frac{\sqrt{2 \lambda}}{\sigma} \left\vert z - \bz^{\epsilon}(t) \right\vert} - e^{- \frac{\sqrt{2 \lambda}}{\sigma} \left\vert z - \uz^{\epsilon}(t) \right\vert} \right) dz \Big)\nonumber\\
&=\lim_{\epsilon \downarrow 0}  \Big(
\lambda \big(\by^{\epsilon}(t)-\uy^{\epsilon}(t)\big) \big(\by^{\epsilon}(t)+\uy^{\epsilon}(t)- 2y_M \big) - \lambda y_M^2 (1-y_M)^2 \big(\bz^{\epsilon}(t)-\uz^{\epsilon}(t)\big) \big(\bz^{\epsilon}(t)+\uz^{\epsilon}(t)- 2z_M \big) \nonumber\\
&\qquad\qquad+ y_M^2(1-y_M)^2 \sqrt{\lambda}\big( \bz^{\epsilon}(t)+\uz^{\epsilon}(t)- 2z_M \big) \Big( 1 - e^{- \frac{\sqrt{2 \lambda}}{\sigma} \left( \bz^{\epsilon}(t) - \uz^{\epsilon}(t) \right)} \Big)\Big(\tfrac{\sigma}{\sqrt{2}} + \tfrac{\sqrt{\lambda}(\bz^{\epsilon}(t) - \uz^{\epsilon}(t) )}{2}  \Big) \Big) \nonumber\\
&=\lim_{\epsilon \downarrow 0}  
\Big(1 - e^{- \frac{\sqrt{2 \lambda}}{\sigma} \left( \bz^{\epsilon}(t) - \uz^{\epsilon}(t) \right)} \Big) \sqrt{\lambda} \big(\by^{\epsilon}(t)+\uy^{\epsilon}(t)- 2y_M \big)
\Big( \tfrac{\sqrt{\lambda} \left(\by^{\epsilon}(t)-\uy^{\epsilon}(t)\right)}{2} + \tfrac{\sigma y_M (1-y_M)}{\sqrt{2}} \Big), \label{JJ}
\end{align}
where the second equality is due to \eqref{J3}, the third equality is due to integration parts, and the last equality is due to \eqref{J1}. \lemref{Merton_fraction_inside_NT} (ii), \eqref{J1} and $\lambda=c \, \epsilon^{-\frac{2}{3}}$ imply $\li_{\epsilon \downarrow 0}\sqrt{\lambda} \left(\bz^{\epsilon}(t)-\uz^{\epsilon}(t)\right)>0$. 
Therefore, \eqref{JJ} implies that $\sqrt{\lambda} \big(\by^{\epsilon}(t)+\uy^{\epsilon}(t)- 2y_M \big)\to 0$ as $\epsilon \downarrow 0$. We rephrase this observation as
\begin{align}
\lim_{\epsilon \downarrow 0} \left( \frac{\by^{\epsilon}(t)-y_M}{\epsilon^{\frac{1}{3}}} -\frac{y_M-\uy^{\epsilon}(t)}{\epsilon^{\frac{1}{3}}}  \right)=0. \label{J4}
\end{align}
In other words, asymptotically, the no-trade region is symmetric around the Merton line.

We substitute $x=y_M$ in \eqref{J_def} and $x_\epsilon =y_M$ in \eqref{J3} and use integration by parts to obtain 
\begin{align}
0&=\lim_{\epsilon \downarrow 0} \bigg(J^{\epsilon}(t,y_M)+ y_M^2(1-y_M)^2\int_{\uz^{\epsilon}(t)}^{\bz^{\epsilon}(t)} \lambda (z-z_M)^{2}  \tfrac{\sqrt{2 \lambda} }{2 \sigma}\,e^{- \frac{\sqrt{2 \lambda}}{\sigma} \left\vert z - z_M \right\vert}  dz - \sigma^2 y_M^2(1-y_M)^2 \bigg) \nonumber\\
&=\lim_{\epsilon \downarrow 0} \bigg( J^{\epsilon} (t, y_M) -  y_M^2 (1-y_M)^2 \Big(
 e^{- \frac{\sqrt{2 \lambda}}{\sigma} \left( z_M - \uz^{\epsilon}(t) \right)} \Big( \tfrac{\lambda(z_M -   \uz^{\epsilon}(t))^2}{2} +\tfrac{\sigma\sqrt{\lambda}(z_M -   \uz^{\epsilon}(t))}{\sqrt{2}} + \tfrac{\sigma^2}{2}
\Big) \nonumber\\
&\qquad\qquad +e^{- \frac{\sqrt{2 \lambda}}{\sigma} \left( \bz^{\epsilon}(t) -z_M \right)} \Big( \tfrac{\lambda(\bz^{\epsilon}(t) -z_M)^2}{2} +\tfrac{\sigma\sqrt{\lambda}(\bz^{\epsilon}(t) -z_M)}{\sqrt{2}} + \tfrac{\sigma^2}{2}
\Big)\Big)\bigg) \nonumber\\
&=\lim_{\epsilon \downarrow 0} \bigg( J^{\epsilon} (t, y_M) - \Big(e^{- \frac{\sqrt{2 \lambda}}{\sigma} \left( z_M - \uz^{\epsilon}(t) \right)} +  e^{- \frac{\sqrt{2 \lambda}}{\sigma} \left( \bz^{\epsilon}(t) -z_M \right)} \Big) \nonumber\\
&\qquad\qquad \qquad
\cdot \Big(  \tfrac{\lambda( \by^{\epsilon}(t) - \uy^{\epsilon}(t))^2}{8} +\tfrac{\sigma y_M(1-y_M)\sqrt{\lambda}( \by^{\epsilon}(t) - \uy^{\epsilon}(t))}{2\sqrt{2}} + \tfrac{\sigma^2y_M^2 (1-y_M)^2}{2} \Big) \bigg), \label{J5}
\end{align}
where we apply \eqref{J1}, \eqref{J2} and \eqref{J4} to obtain the last equality.
We substitute $x_\epsilon =y_M$ in \eqref{I4} and combine with \eqref{J5} to obtain
\begin{align*}
I^{\epsilon} (t, y_M) +\gamma(1-y_M)^{\gamma-1}v^{0}(t) \Big( \Big( 
\tfrac{\sqrt{\lambda}( \by^{\epsilon}(t) - \uy^{\epsilon}(t))}{2} + \tfrac{\sigma y_M(1-y_M)}{\sqrt{2}}\Big)^2
+\tfrac{ \sigma^2 y_M^2 (1-y_M)^2}{2} \Big)  \stackrel{\epsilon \downarrow 0}{\longrightarrow}0.
\end{align*}
By \eqref{lim1}, the above limit and $\lambda=c \, \epsilon^{-\frac{2}{3}}$, we conclude
\begin{align}
0&=\lim_{\epsilon \downarrow 0} \bigg( \tfrac{  c (v_t^{\epsilon}(t, y_M) - v_t^{0}(t) )}{(1 - \gamma)\epsilon^{\frac{2}{3}}} +  \tfrac{c\, Q(y_M)( v^{\epsilon}(t,y_M) - v^{0}(t))}{(1 - \gamma)\epsilon^{\frac{2}{3}}} \nonumber\\
&\qquad\qquad - \tfrac{\gamma \sigma^2 v^{0}(t)}{2}\Big( \Big( 
\tfrac{\sqrt{c}( \by^{\epsilon}(t) - \uy^{\epsilon}(t))}{2\epsilon^{\frac{1}{3}}} + \tfrac{\sigma y_M(1-y_M)}{\sqrt{2}}\Big)^2
+\tfrac{ \sigma^2 y_M^2 (1-y_M)^2}{2} \Big)  \bigg).
 \label{lim2}
\end{align}

{\bf Step 2.} We inspect the integrals of the terms in \eqref{PDE_on_NT_region} with respect to $x$, from $\uy^{\epsilon}(t)$ to $\by^{\epsilon}(t)$. By the mean value theorem, there exist $x_\epsilon^*, x_\epsilon^{**} \in [\uy^{\epsilon}(t),\by^{\epsilon}(t)]$ such that
\begin{align}
&\int_{\uy^{\epsilon}(t)}^{\by^{\epsilon}(t)} \tfrac{ v_t^{\epsilon}(t, x) - v_t^{0}(t)}{(1-\gamma)\epsilon} dx - \tfrac{v_t^{\epsilon}(t, y_M) - v_t^{0}(t)}{(1-\gamma)\epsilon^{\frac{2}{3}} }\cdot  \tfrac{\by^{\epsilon}(t) - \uy^{\epsilon}(t)}{\epsilon^{\frac{1}{3}} }  =\tfrac{v_t^{\epsilon}(t, x_\epsilon^*) - v_t^{\epsilon}(t, y_M)}{(1-\gamma)\epsilon^{\frac{2}{3}} }  \cdot  \tfrac{\by^{\epsilon}(t) - \uy^{\epsilon}(t)}{\epsilon^{\frac{1}{3}} } \stackrel{\epsilon \downarrow 0}{\longrightarrow}0, \label{lim3}\\
&\int_{\uy^{\epsilon}(t)}^{\by^{\epsilon}(t)}  \tfrac{Q(x) \left( v^{\epsilon}(t, x) - v^{0}(t)\right)}{(1-\gamma)\epsilon}  dx  -  \tfrac{Q(y_M)\left(v^{\epsilon}(t, y_M) - v^{0}(t) \right)}{(1-\gamma)\epsilon^{\frac{2}{3}} } \cdot  \tfrac{\by^{\epsilon}(t) - \uy^{\epsilon}(t)}{\epsilon^{\frac{1}{3}} }  \nonumber\\
&\qquad\qquad\qquad\qquad=\left(\tfrac{Q(x_\epsilon^{**})\left(v^{\epsilon}(t, x_\epsilon^{**}) - v^{0}(t) \right)}{(1-\gamma)\epsilon^{\frac{2}{3}} } - \tfrac{Q(y_M)\left(v^{\epsilon}(t, y_M) - v^{0}(t) \right)}{(1-\gamma)\epsilon^{\frac{2}{3}} } \right)  \tfrac{\by^{\epsilon}(t) - \uy^{\epsilon}(t)}{\epsilon^{\frac{1}{3}} }\stackrel{\epsilon \downarrow 0}{\longrightarrow}0, \label{lim4}
\end{align}
where the convergences are due to \lemref{Merton_fraction_inside_NT} (i) and \lemref{really_used}. By \eqref{J4}, we have
\begin{align}
&\int_{\uy^{\epsilon}(t)}^{\by^{\epsilon}(t)} \tfrac{Q(x) - Q(y_M)}{(1-\gamma)\epsilon} dx + \tfrac{\gamma \sigma^2}{24}\Big(\tfrac{\by^{\epsilon}(t)-\uy^{\epsilon}(t)}{\epsilon^{\frac{1}{3}} }\Big)^3 \nonumber\\
& = - \tfrac{\gamma \sigma^2}{6} \left( \Big(\tfrac{\by^{\epsilon}(t)-y_M}{\epsilon^{\frac{1}{3}} }\Big)^3 -  \Big(\tfrac{\uy^{\epsilon}(t)-y_M}{\epsilon^{\frac{1}{3}} }\Big)^3 - \tfrac{1}{4}\Big(\tfrac{\by^{\epsilon}(t)-\uy^{\epsilon}(t)}{\epsilon^{\frac{1}{3}} }\Big)^3 \right)
\stackrel{\epsilon \downarrow 0}{\longrightarrow} 0. \label{lim5}
\end{align}
By \eqref{vx_epsilon} and \lemref{Merton_fraction_inside_NT} (i), we obtain
\begin{align}
&\Big|\int_{\uy^{\epsilon}(t)}^{\by^{\epsilon}(t)}
\tfrac{x (1 - x) \left( \mu - \gamma \sigma^{2} x \right) v_{x}^{\epsilon}(t, x)}{(1-\gamma)\epsilon} dx \Big| \leq C \left| \by^{\epsilon}(t)-\uy^{\epsilon}(t) \right| \stackrel{\epsilon \downarrow 0}{\longrightarrow} 0. \label{lim6}
\end{align}
By integration by parts and \lemref{Boundaries_of_NT_region},
\begin{align}
\int_{\uy^{\epsilon}(t)}^{\by^{\epsilon}(t)}
\tfrac{\sigma^2 x^2 (1 - x)^2 v_{x x}^{\epsilon, \lambda}(t, x)}{2(1-\gamma)\epsilon}   dx &=
- \tfrac{\sigma^2\by^{\epsilon}(t)^2(1-\by^{\epsilon}(t))^2  v^{\epsilon}(t, \by^{\epsilon}(t))}{2(1-\epsilon \by^{\epsilon}(t))} - \tfrac{\sigma^2\uy^{\epsilon}(t)^2(1-\uy^{\epsilon}(t))^2  v^{\epsilon}(t, \uy^{\epsilon}(t))}{2(1+\epsilon \uy^{\epsilon}(t))}\nonumber\\
&\quad -\int_{\uy^{\epsilon}(t)}^{\by^{\epsilon}(t)}
 \tfrac{\sigma^{2}x (1 - x)(1-2x)  v_{x}^{\epsilon}(t, x)}{(1-\gamma)\epsilon} dx \,\,\stackrel{\epsilon \downarrow 0}{\longrightarrow}\,\, - \sigma^2 y_M^2(1-y_M)^2 v^{0}(t), \label{lim7} 
\end{align}
where the convergence is due to \lemref{Merton_fraction_inside_NT} (i) and \lemref{really_used} (i). 

Now we integrate the right-hand side of \eqref{PDE_on_NT_region} with respect to $x$ from $\uy^{\epsilon}(t)$ to $\by^{\epsilon}(t)$ and multiply it by $\tfrac{c}{(1-\gamma)\epsilon}$, then apply \eqref{lim3}-\eqref{lim7} to obtain the following:
\begin{align}
&0=\lim_{\epsilon \downarrow 0} \bigg( \Big(\tfrac{  c(v_t^{\epsilon}(t, y_M) - v_t^{0}(t) )}{(1 - \gamma)\epsilon^{\frac{2}{3}}} +  \tfrac{c\, Q(y_M)( v^{\epsilon}(t,y_M) - v^{0}(t))}{(1 - \gamma)\epsilon^{\frac{2}{3}}}\Big) \tfrac{ \by^{\epsilon}(t) - \uy^{\epsilon}(t)}{\epsilon^{\frac{1}{3}}}\nonumber \\
&\qquad\qquad\qquad - \tfrac{c\gamma \sigma^2 v^{0}(t)}{24}\Big( \tfrac{ \by^{\epsilon}(t) - \uy^{\epsilon}(t)}{\epsilon^{\frac{1}{3}}}\Big)^3 -c \sigma^2 y_M^2(1-y_M)^2 v^{0}(t)  \bigg).\label{lim8}
\end{align}

{\bf Step 3.} 
We multiply $\tfrac{ \by^{\epsilon}(t) - \uy^{\epsilon}(t)}{\epsilon^{\frac{1}{3}}}$ to \eqref{lim2} and subtract \eqref{lim8} to obtain
\begin{align*}
0=\lim_{\epsilon \downarrow 0} \Big( -\tfrac{c \gamma}{12} \Big(\tfrac{ \by^{\epsilon}(t) - \uy^{\epsilon}(t)}{\epsilon^{\frac{1}{3}}}+ \tfrac{\sqrt{2} \sigma y_M(1-y_M)}{\sqrt{c}} \Big)^3 + c y_M^2(1-y_M)^2 + \tfrac{\gamma \sigma^3 y_M^3(1-y_M)^3}{3\sqrt{2 c }}\Big). 
\end{align*} 
The above equation implies that
\begin{align}
\lim_{\epsilon\downarrow 0} \tfrac{ \by^{\epsilon}(t) - \uy^{\epsilon}(t)}{\epsilon^{\frac{1}{3}}} = \tfrac{\sqrt{2} \sigma y_M(1-y_M)}{\sqrt{c}} \Big( \Big( \tfrac{3\sqrt{2} \,c^{\frac{3}{2}}}{\gamma \sigma^3 y_M(1-y_M)} +1\Big)^{\frac{1}{3}}-1 \Big)=2a_1(c). \label{cubic_solution}
\end{align}
We conclude the desired asymptotic result for $\by^{\epsilon}(t)$ and $\uy^{\epsilon}(t)$ by the above equation and \eqref{J4}.  

It remains to prove the asymptotic result for $v^{\epsilon}(t,y_M)$. Using \eqref{merton_pde}, we rewrite \eqref{lim2} as
\begin{align}
 \tfrac{\partial}{\partial t} \Big(\tfrac{ e^{Q(y_M)t} ( v^{\epsilon}(t, y_M) - v^{0}(t) ) }{(1 - \gamma)\epsilon^{\frac{2}{3}}} \Big)  - \tfrac{\gamma \sigma^2 e^{Q(y_M)T}}{2}\Big( \Big( 
\tfrac{ \by^{\epsilon}(t) - \uy^{\epsilon}(t)}{2\epsilon^{\frac{1}{3}}} + \tfrac{\sigma y_M(1-y_M)}{\sqrt{2\,c }}\Big)^2
+\tfrac{ \sigma^2y_M^2 (1-y_M)^2}{2\, c} \Big)\stackrel{\epsilon \downarrow 0}{\longrightarrow} 0.  \label{lim9}
\end{align}
Then, \eqref{cubic_solution} and \eqref{lim9} imply 
\begin{align*}
\lim_{\epsilon \downarrow 0}  \tfrac{\partial}{\partial t} \Big(\tfrac{ e^{Q(y_M)t}(v^{\epsilon}(t, y_M) - v^{0}(t) ) }{(1 - \gamma)\epsilon^{\frac{2}{3}}} \Big) 
=\tfrac{e^{Q(y_M)T} a_2(c)}{1-\gamma } .
\end{align*}
The bounds in \lemref{really_used} (i) enable us to use the dominated convergence theorem as below:
\begin{align*}
\tfrac{e^{Q(y_M)T} a_2(c)}{1-\gamma } (T-t)
&=\int_t^T \lim_{\epsilon \downarrow 0}  \tfrac{\partial}{\partial s} \Big( \tfrac{  e^{Q(y_M)s} (v^{\epsilon}(s, y_M) - v^{0}(s)) }{(1 - \gamma)\epsilon^{\frac{2}{3}}} \Big) ds =\lim_{\epsilon \downarrow 0}  \int_t^T \tfrac{\partial}{\partial s} \Big( \tfrac{ e^{Q(y_M)s}  (v^{\epsilon}(s, y_M) - v^{0}(s)) }{(1 - \gamma)\epsilon^{\frac{2}{3}}} \Big) ds\\
&=- e^{Q(y_M)t}\cdot\lim_{\epsilon \downarrow 0}  \tfrac{   v^{\epsilon}(t, y_M) - v^{0}(t) }{(1 - \gamma)\epsilon^{\frac{2}{3}}}, 
\end{align*}
where the last equality is due to $v^{\epsilon}(T, y_M) = v^{0}(T) = 1$. We conclude the desired asymptotic result for $v^{\epsilon}(t,y_M)$.

\section{conclusion}

This paper investigates the optimal investment problem in a market with two types of illiquidity: transaction costs and search frictions. Building on the framework established by \cite{GC23}, we extend the analysis to a power-utility maximization problem.
Our main contribution is the development of a novel asymptotic framework applicable when both transaction costs and search frictions are small ($\epsilon \ll 1$ and $\tfrac{1}{\lambda} \ll 1$). We derive explicit asymptotics for the no-trade region and the value function along the parametric curve $\lambda = c \, \epsilon^{-\frac{2}{3}}$ for $c>0$. This approach unifies the existing asymptotic results for models with only transaction costs or only search frictions, providing a coherent methodology for handling both types of illiquidity simultaneously. Additionally, our framework offers explicit expressions for the correction terms, facilitating practical computation of the optimal trading strategy and value. Our asymptotic analysis provides insights into the limiting behaviors not present in models with only one source of illiquidity.
As a future research, we plan to extend our results to a multi-asset model.

\bibliographystyle{siam}

\appendix

\medskip

\section{Proof of \lemref{v_classical_solution_and_properties}}\label{appendix_A}

In the proof, we assume $0<\gamma<1$ (the case of $\gamma>1$ can be treated similarly). 
Let $a:=\max_{x\in [0,1]} Q(x)$ where $Q$ is in \eqref{Q_def}, $h(z) := \frac{e^{z}}{1 + e^{z}}$, and $C_b([0,T]\times \R)$ be the set of all bounded  (with the uniform norm) continuous functions. For $f\in C_b([0,T]\times \R)$, we define $\phi(f)$ as 
\begin{align}
  \phi(f)(t, z) & := \EE \left[ e^{a T } e^{ \int_t^T \left( Q(h(\Upsilon_u^{(t,z)}))-\lambda-a \right) du} + \lambda \int_{t}^{T} e^{ \int_t^s \left( Q(h(\Upsilon_u^{(t,z)}))-\lambda-a \right) du}  K_f ( s, \Upsilon_{s}^{(t, z)} ) d s \right], \label{Stochastic_representation_v_contraction}
\end{align}
where $K_f$ is
\begin{align}
  K_f(t, z) & := \sup_{\zeta \in \R} \Big( f(t, \zeta)  g(z,\zeta) \Big) \quad \textrm{with} \quad g(z,\zeta):=\big( \tfrac{1 + \be h(z)}{1 + \be h(\zeta)} \big)^{1 - \gamma} 1_{\left\{ \zeta \geq z \right\}} + \big( \tfrac{1 - \ue h(z)}{1 - \ue h(\zeta)} \big)^{1 - \gamma} 1_{\left\{ \zeta<z \right\}}  \label{Define_Kf}
\end{align}
and $\Upsilon_{s}^{(t, z)}$ for $(s,z)\in [t,T]\times \R$ is the solution of the following SDE: 
\begin{align}
d\Upsilon_s^{(t,z)}= \left( \mu-\tfrac{\sigma^2}{2} + (1-\gamma)\sigma^2 h(\Upsilon_s^{(t,z)})\right) ds + \sigma dB_s, \quad \Upsilon_t^{(t,z)}=z. \label{Upsilon_SDE}
\end{align}
Since $K_f$ and $Q \circ h$ are bounded and continuous, one can check $\phi(f)\in C_b([0,T]\times \R)$ by the dominated convergence theorem. 
From the definition of $a$, we observe that for $z\in \R$, 
\begin{align}
-\infty< \min_{x\in [0,1]} Q(x)  - \lambda - a  \leq Q(h(z))-\lambda- a \leq - \lambda. \label{exp_bounds}
\end{align}

We check that $\phi$ is a contraction map: for $f_1, f_2 \in C_b([0,T]\times \R)$, 
\begin{align*}
\lVert \phi(f_1)- \phi(f_2) \rVert_\infty &\leq \lambda \int_t^T e^{-\lambda(s-t)} 
\cdot \sup_{\zeta\in \R} \big| f_1(s,\zeta)-f_2(s,\zeta) \big|  ds \leq (1-e^{-\lambda T} ) \lVert f_1-f_2\rVert_\infty,
\end{align*}
where the first inequality is due to \eqref{exp_bounds} and $\lVert g\lVert_\infty \leq 1$.
Therefore, by the Banach fixed point theorem, there exists a unique function $\hat f \in C_b([0,T]\times \R)$ such that
$\phi(\hat f) = \hat f$. 

\medskip

\begin{lemma}\label{K_lemma}
$K_{\hat f} \in C^{\frac{1}{2}, 1} ([0,T]\times \R)$.
\end{lemma}
\begin{proof} Throughout the proof of this claim, $C>0$ is a generic constant independent of $(t,s,z, \delta)\in [0,T]\times [t,T] \times \R \times [0,1]$ and paths. For $\delta\in [0,1]$ and $t\in [0,T]$,
\begin{align}
\left|  K_{\hat f} (t,z+\delta) - K_{\hat f} (t,z)    \right|  &\leq   \lVert \hat f   \rVert_\infty \cdot \sup_{\zeta\in \R} \left| g(z+\delta,\zeta) - g(z,\zeta)\right| \leq C \, \delta, \label{K_space}
\end{align}
where the second inequality is due to $\lVert \tfrac{\partial g}{\partial z}  \rVert_\infty <\infty$. By SDE \eqref{Upsilon_SDE}, for $s\in [t,T]$, we have
\begin{align*}
 \left| \Upsilon_s^{(t+\delta,z)}- \Upsilon_s^{(t,z) }  \right| &\leq \left| \sigma(B_t - B_{t+\delta}) 
 - \int_t^{t+\delta}  \left( \mu-\tfrac{\sigma^2}{2} + (1-\gamma)\sigma^2 h(\Upsilon_u^{(t,z)})\right) du \right| \\
 &\qquad\qquad + \left| \int_{t+\delta}^s (1-\gamma)\sigma^2 \left( h(\Upsilon_u^{(t+\delta,z)})-h(\Upsilon_u^{(t,z)}) \right) du \right| \\
 &\leq  C \left( \left| B_{t+\delta} - B_t \right| + \delta + \int_{t+\delta}^s  \left| \Upsilon_u^{(t+\delta,z)}- \Upsilon_u^{(t,z) }  \right| du  \right). 
\end{align*}
We apply Gronwall's inequality (see \lemref{Gronwall_lemma}) to above and obtain  
\begin{align}
\big| \Upsilon_s^{(t+\delta,z)}- \Upsilon_s^{(t,z) }  \big| \leq C \left( \left| B_{t+\delta} - B_t \right| + \delta \right).  \label{Upsilon_estimate}
\end{align}
Using \eqref{exp_bounds} and \eqref{Upsilon_estimate}, we produce the following estimate: for $s\in [t,T]$ and $\delta\in [0,1]$,
\begin{align}
&\EE \Big[ \Big|  e^{ \int_t^s \left( Q(h(\Upsilon_u^{(t,z)}))-\lambda-a \right) du} - e^{ \int_{t+\delta}^s \left( Q(h(\Upsilon_u^{(t+\delta,z)}))-\lambda-a \right) du}  \Big| \Big] \nonumber\\
&\leq C \, \EE \Big[\Big| \int_t^s \left( Q(h(\Upsilon_u^{(t,z)}))-\lambda-a \right) du -  \int_{t+\delta}^s \left( Q(h(\Upsilon_u^{(t+\delta,z)}))-\lambda-a \right) du  \Big| \Big] \leq C\sqrt{\delta} ,  \label{delta_t}
\end{align}
where we also used $\lVert \tfrac{d}{d z}Q(h(z))  \rVert_\infty <\infty$ for the last inequality. Similarly, \eqref{K_space} and \eqref{Upsilon_estimate} produce
\begin{align}
\EE \Big[\Big| K_{\hat f} ( s, \Upsilon_{s}^{(t+\delta, z)} )- K_{\hat f} ( s, \Upsilon_{s}^{(t, z)} ) \Big| \Big]
&\leq  C \, \EE \Big[ \Big|  \Upsilon_s^{(t+\delta,z)}- \Upsilon_s^{(t,z) }  \Big| \Big] \leq C\sqrt{\delta}. \label{K_f_estimate}
\end{align} 
For $\delta\in [0,1]$ and $(t,z)\in [0,T]\times \R$, \eqref{exp_bounds}, \eqref{delta_t} and \eqref{K_f_estimate} produce
\begin{align}
\left| \phi(\hat f)(t+\delta,z)-\phi(\hat f) (t,z)  \right| \leq C \sqrt{\delta}. \label{est1}
\end{align}
The above inequality implies that
\begin{align*}
\left| K_{\hat f}(t+\delta,z)- K_{\hat f}(t,z)\right| \leq \sup_{\zeta \in \R} \left| \hat f(t+\delta,\zeta)-\hat f (t,\zeta)  \right| = \sup_{\zeta \in \R} \left| \phi(\hat f)(t+\delta,\zeta)-\phi(\hat f) (t,\zeta) \right| \leq C \sqrt{\delta}. 
\end{align*}
We conclude $K_{\hat f} \in C^{\frac{1}{2}, 1} ([0,T]\times \R)$ by the above inequality and \eqref{K_space}. 
\end{proof}

Let $\alpha\in (0,1)$ be fixed. \lemref{K_lemma} and Theorem 9.2.3 in \cite{krylov1996lectures} guarantee that there exists a unique solution $\tilde f\in C^{1 + \frac{\alpha}{2}, 2 + \alpha}([0, T] \times \mathbb{R})$ of the following PDE:
\begin{align} \label{Region_I}
  \begin{cases}
  0 = f_{t}(t, z) +\left( \mu-\tfrac{\sigma^2}{2} + (1-\gamma)\sigma^2 h(z)\right) f_z(t, z) + \tfrac{\sigma^{2}}{2}  f_{z z}(t, z)  \\
\qquad  \qquad + \big(Q(h(z)) - \lambda-a \big) f(t, z) + \lambda  K_{\hat{f}}(t, z) \\
  f(T, z)=e^{aT}
  \end{cases}
\end{align}
By the Feynman-Kac formula (i.e., see Theorem 5.7.6 in \cite{karatzas2014brownian}), we have $\tilde f=\phi(\hat{f})=\hat f$. 

Observe that $h$ can be continuously extended to $z=\pm\infty$ as $h(\infty):=1$ and $h(-\infty):=0$. Then,
\begin{align*}
\lim_{z\to  \infty} e^{-at}  \tilde f(t,z) &= \lim_{z\to  \infty} e^{-at} \phi(\tilde f)(t,z) \\
&= 
e^{ (Q(1)-\lambda)(T-t)}  + \lambda \int_t^T e^{ (Q(1)-\lambda) (s-t)} \sup_{\zeta\in \R} \left(e^{-as} \tilde f(s,\zeta) \left( \tfrac{1-\ue}{1-\ue h(\zeta)} \right)^{1-\gamma} \right) ds,
\end{align*}
where the last equality is due to the dominated convergence theorem. Therefore, we continuously extend $\tilde f$ to $z=+ \infty$ and the above equality produces 
\begin{align}
0=\tfrac{\partial}{\partial t} \left( e^{-at} \tilde f(t, \infty)\right) + (Q(1)-\lambda) \left( e^{-at} \tilde f(t, \infty)\right)  + \lambda \sup_{\zeta\in\R} \left( e^{-at} \tilde f(t, \zeta) \left( \tfrac{1-\ue}{1-\ue h(\zeta)} \right)^{1-\gamma}  \right). \label{ODE1}
\end{align} 
We can treat $\lim_{z\to  -\infty} e^{-at}  \tilde f(t,z)$ by the same way and obtain
\begin{align}
0=\tfrac{\partial}{\partial t} \left( e^{-at} \tilde f(t, -\infty)\right) + (Q(0)-\lambda) \left( e^{-at} \tilde f(t, -\infty)\right)  + \lambda \sup_{\zeta\in\R} \left( e^{-at} \tilde f(t, \zeta) \left( \tfrac{1}{1+\be h(\zeta)} \right)^{1-\gamma}  \right). \label{ODE2}
\end{align} 

We set $v(t,x):=e^{-at} \tilde f(t,h^{-1}(x))$ for $(t,x)\in [0,T]\times [0,1]$. The PDE for $\tilde f$ in \eqref{Region_I} implies that $v$ satisfies \eqref{Main_PDE} for $(t,x)\in [0,T)\times (0,1)$. We check (ii) by \eqref{ODE1} and \eqref{ODE2}. To check (iii), since $\tilde f\in C^{1 + \frac{\alpha}{2}, 2 + \alpha}([0, T] \times \mathbb{R})$, it is enough to observe that for $z=h^{-1}(x)$,
\begin{align*}
&v_t(t,x)=e^{-at}(\tilde f_t(t,z) - a \tilde f(t,z)),\\
&x(1-x)v_x(t,x)= e^{-at} \tilde f_z(t,z),\\
&x^2(1-x)^2 v_{xx}(t,x)=e^{-at}\left(\tilde f_{zz}(t,z)-(1-2x)\tilde f_z(t,z)  \right).
\end{align*}

\medskip

\section{Preliminary analysis for Section 5}

This appendix is devoted to presenting and proving preliminary asymptotic results used in the proof of the lemmas in Section 5. 
As in Section 5, we set $\epsilon = \be = \ue\in (0,1)$ and assume that $y_M \in (0, 1)$.
Recall that $A_{s, t}$, $Y_{s}^{(t, x)}$ and $Z_{s}^{(t, x)}$ are defined in \eqref{AYZ}.

\begin{lemma} \label{Boundedness_of_various_expectation} 
Let $C>0$ be a generic constant independent of $(t, s, x,\lambda)\in [0,T]\times [t,T]\times (0,1)\times [1,\infty)$ that may differ line by line.

\noindent(i) For nonnegative integers $n, m, k, l$, 
\begin{align}
&\left\vert \mathbb{E} \left[ \left( \tfrac{\partial^{m} Z_{s}^{(t, x)}}{\partial x^{m}} \right)^{k} \left( \tfrac{\partial^{n} Y_{s}^{(t, x)}}{\partial x^{n}} \right)^{l}\right] \right\vert  \leq   C (s - t) 
+ \begin{cases} 
0, \quad m\geq 1 \textrm{  or  } n \geq 2\\
x^l, \quad m=0 \textrm{  and  } n=0\\
1, \quad m=0 \textrm{  and  } n=1
\end{cases}
\label{Expectation_Time_order}
\end{align} 
In particular, 
\begin{equation}
\begin{split}\label{Boundedness_of_Expectation_of_Multiplied_term}
& \mathbb{E} \left[ \left| Z_{s}^{(t, x)} \right| \right]  \leq 1+C(s-t), \quad  \mathbb{E} \left[ \left| \tfrac{\partial^2 Z_{s}^{(t, x)}}{\partial x^2} \right| \right]  \leq C(s-t), \quad  \mathbb{E} \left[ \left| \tfrac{\partial Z_{s}^{(t, x)}}{\partial x} \right| \right]  \leq C\sqrt{s-t}.
\end{split}
\end{equation}

\noindent(ii) Let $n\in \N$ and $F: [0,T]\times (0,1) \to \R$. Suppose that $y\mapsto F(t,y)$ is piecewise continuous for each $t\in [0,T]$ and $\sup_{(t,y)\in [0,T]\times (0,1)}\left|(y(1-y))^{n-1}F(t,y)\right|=C_F<\infty$.
Then,
\begin{align}
\left| \tfrac{\partial}{\partial x}\left(\int_t^T \lambda e^{-\lambda(s-t)} \mathbb{E} \Big[ Z_{s}^{(t, x)}  \left(\tfrac{\partial Y_{s}^{(t, x)}}{\partial x}\right)^n F( s, Y_{s}^{(t, x)}) \Big] ds \right) \right|  \leq \tfrac{C \cdot C_F\sqrt{\lambda}}{(x (1 - x))^{n}}.
\end{align}

\noindent(iii) Let $F: (0,1) \to \R$ be a continuous function. Suppose that $F'$ is continuous on $(0,1)$, $F''$ is continuous on $(0,1)$ except finitely many points, the left and right limits of $F''$ exist at the discontinuous points and
\begin{align}
\esssup_{y\in (0,1)} \Big( \big| F(y)  \big| + \big| y(1-y) F'(y)\big| + \big|y^2(1-y)^2 F''(y)\big| \Big) <\infty.   \label{F_bounds}
\end{align}
Then, the derivative below exists and 
\begin{align}
&  \left\vert \tfrac{\partial}{\partial t} \EE \left[ \tfrac{\partial Z_{s}^{(t, x)}}{\partial x} F (Y_{s}^{(t, x)}) \right] \right\vert \leq C \Big( \big| (y_M -x) F(x) \big| + \big| x(1-x) F'(x) \big| \Big)  + C (s-t),\label{Derivative_of_expectation_process_1} \\
&  \left\vert  \EE \left[ \tfrac{\partial Z_{s}^{(t, x)}}{\partial x} F (Y_{s}^{(t, x)}) \right] \right\vert \leq C \Big( \big| (y_M -x) F(x) \big| + \big| x(1-x) F'(x) \big| \Big) (s-t)  + C (s-t)^2. \label{Derivative_of_expectation_process_2} 
\end{align}

\end{lemma}

\begin{proof}
Throughout this proof, $C>0$ is a generic constant independent of $(t,s,x,\lambda)\in [0,T]\times [t,T]\times (0,1) \times [1,\infty)$  that may differ line by line.

(i) To obtain \eqref{Expectation_Time_order}, we apply Ito's lemma to $\left( \tfrac{\partial^{m} Z_{s}^{(t, x)}}{\partial x^{m}} \right)^{k} \left( \tfrac{\partial^{n} Y_{s}^{(t, x)}}{\partial x^{n}} \right)^{l}$ using the expression in \eqref{YZ_derivative} and the SDE for $A_s^{(t,x)}$ in \eqref{A_SDE}, then we apply the inequalities in \eqref{A_ineq}. 

Since $Z_{s}^{(t, x)}> 0$, we obtain the first inequality in \eqref{Boundedness_of_Expectation_of_Multiplied_term} by \eqref{Expectation_Time_order} (with $m=n=0$, $k=1$ and $l=0$). 
The expression in \eqref{YZ_derivative} and \eqref{Expectation_Time_order} (with $m=2$, $k=1$ and $l=0$) imply 
$$ \mathbb{E} \left[ \left| \tfrac{\partial^2 Z_{s}^{(t, x)}}{\partial x^2} \right| \right] = \textrm{sgn}(\gamma-1) \cdot \mathbb{E} \left[  \tfrac{\partial^2 Z_{s}^{(t, x)}}{\partial x^2}  \right] \leq C(s-t).$$
H$\ddot{o}$lder's inequality and \eqref{Expectation_Time_order} (with $m=1$, $k=2$ and $l=0$) produce the last inequality in \eqref{Boundedness_of_Expectation_of_Multiplied_term}.

(ii) The probability density function $\varphi(y;s - t, x)$ of $Y_{s}^{(t, x)}$ is calculated as 
\begin{align}
  \varphi(y;s - t, x) & := \tfrac{\partial}{\partial y} \mathbb{P} \left( Y_{s}^{(t, x)} \leq y \right) = \tfrac{\exp \left( - \frac{1}{2 \sigma^{2} (s - t)} \left( (\frac{\sigma^2}{2}-\mu) (s - t) + \ln \left( \frac{y (1 - x)}{(1 - y) x} \right) \right)^{2} \right)}{\sigma y (1 - y) \sqrt{2 \pi (s - t)}} \label{density_form}\\
  \Longrightarrow \quad \varphi_{x}(y;s - t, x) & = \varphi(y;s - t, x) \cdot \tfrac{1}{\sigma^{2} x (1 - x)} \left(\tfrac{\sigma^2}{2}-\mu + \tfrac{1}{s-t} \ln \left( \tfrac{y (1 - x)}{(1 - y) x} \right) \right). \nonumber
\end{align}

We use the expression above to obtain
\begin{align}
  &\int_{0}^{1} \left| \tfrac{\partial}{\partial x} \left( \left( \tfrac{1 - x}{1 - y} \right)^{1 - \gamma} \left( \tfrac{y (1 - y)}{x (1 - x)} \right)^{n} F(s, y) \varphi(y;s - t, x) \right) \right| d y \nonumber \\
  & = \int_{0}^{1} \left( \tfrac{1 - x}{1 - y} \right)^{1 - \gamma} \left( \tfrac{y (1 - y)}{x (1 - x)} \right)^{n} \left| F(s, y)\right| \nonumber \\
  &\qquad\qquad\qquad \cdot  \left| \tfrac{1}{\sigma^{2} x (1 - x)} \left(\tfrac{\sigma^2}{2}-\mu+ \tfrac{1}{s-t} \ln \left( \tfrac{y (1 - x)}{(1 - y) x} \right) \right) - \tfrac{n (1 - 2 x) + x (1 - \gamma)}{x (1 - x)}\right| \cdot \varphi(y;s - t, x)  d y \nonumber \\
  & =\tfrac{1}{(x(1-x))^n} \,  \mathbb{E} \left[ Z_{s}^{(t, x)}  \tfrac{\partial Y_{s}^{(t, x)}}{\partial x}   (Y_{s}^{(t, x)} (1-Y_{s}^{(t, x)} ))^{n-1} \left| F ( s, Y_{s}^{(t, x)} ) \right| \cdot \left| \tfrac{B_{s} - B_{t}}{\sigma (s - t)} - n - x (1 - \gamma - 2 n) \right| \right] \nonumber\\
&\leq \tfrac{C\cdot C_F}{(x (1 - x))^{n}\sqrt{s-t}}, 
 \label{D_Density_function_integrability}
\end{align}
where the inequality is due to \eqref{Expectation_Time_order} and H$\ddot{o}$lder's inequality. This implies that 
\begin{align}
  &\int_{0}^{1}  \tfrac{\partial}{\partial x} \left( \left( \tfrac{1 - x}{1 - y} \right)^{1 - \gamma} \left( \tfrac{y (1 - y)}{x (1 - x)} \right)^{n} F(s, y) \varphi(y;s - t, x) \right)  d y \nonumber \\
  & =\tfrac{1}{(x(1-x))^n} \mathbb{E} \Big[ Z_{s}^{(t, x)}  \tfrac{\partial Y_{s}^{(t, x)}}{\partial x}  \big(Y_{s}^{(t, x)} (1-Y_{s}^{(t, x)} )\big)^{n-1} F( s, Y_{s}^{(t, x)} ) \left( \tfrac{B_{s} - B_{t}}{\sigma (s - t)} - n - x (1 - \gamma - 2 n) \right) \Big] \nonumber
  \end{align}
is well-defined and continuous in $x$. Therefore, together with the estimate \eqref{D_Density_function_integrability}, we validate the interchange of integration and differentiation below:
\begin{align*}
&\left| \tfrac{\partial}{\partial x}\left(\int_t^T \lambda e^{-\lambda(s-t)} \mathbb{E} \Big[ Z_{s}^{(t, x)}  \left(\tfrac{\partial Y_{s}^{(t, x)}}{\partial x}\right)^n F( s, Y_{s}^{(t, x)}) \Big] ds \right) \right|  \\
&=\left| \int_t^T \lambda e^{-\lambda(s-t)}  \int_{0}^{1} \tfrac{\partial}{\partial x} \left( \left( \tfrac{1 - x}{1 - y} \right)^{1 - \gamma} \left( \tfrac{y (1 - y)}{x (1 - x)} \right)^{n} F(s, y) \varphi(y;s - t, x) \right)  d y \, ds  \right|  
\leq   \tfrac{C\cdot C_F \sqrt{\lambda}}{(x (1 - x))^{n}},
\end{align*}
where the inequality is due to \eqref{D_Density_function_integrability} and \lemref{Order_of_exponential_times}. Especially, when $n=1$, we have
\begin{align}
&\tfrac{\partial}{\partial x}\left(\int_t^T \lambda e^{-\lambda(s-t)} \mathbb{E} \Big[ Z_{s}^{(t, x)}  \tfrac{\partial Y_{s}^{(t, x)}}{\partial x} F ( s, Y_{s}^{(t, x)} ) \Big] ds \right) \nonumber\\
 & =\tfrac{1}{x(1-x)} \int_t^T \lambda e^{-\lambda(s-t)}  \mathbb{E} \left[ Z_{s}^{(t, x)}  \tfrac{\partial Y_{s}^{(t, x)}}{\partial x}   F( s, Y_{s}^{(t, x)}) \left( \tfrac{B_{s} - B_{t}}{\sigma (s - t)} - 1 + x (1 + \gamma ) \right) \right] ds . \label{F_expression}
\end{align}

(iii) By \eqref{YZ_derivative}, we have $\tfrac{\partial Z_{s}^{(t, x)}}{\partial x} F(Y_{s}^{(t, x)})=f(A_{s,t})$ with  $f(a):=(1-\gamma)(a-1)(xa+1-x)^{-\gamma} F(\tfrac{xa}{xa+1-x})$. The conditions for $F$ allow us to apply Ito's lemma (see 6.24 Problem in p.215 of \cite{karatzas2014brownian}) and obtain
\begin{align*}
f(A_{s-t,0})= \int_0^{s-t} g(A_{u,0}) du  +  \int_0^{s-t} f'(A_{u,0}) \sigma A_{u,0} dB_u, 
\end{align*}
where $g(a):=f'(a) \mu a + \tfrac{1}{2} f''(a) \sigma^2 a^2$. The inequalities in \eqref{A_ineq} and \eqref{F_bounds} imply that the stochastic integral term above is a square integrable martingale. Therefore,
\begin{align*}
\EE \left[ \tfrac{\partial Z_{s}^{(t, x)}}{\partial x} F(Y_{s}^{(t, x)})\right] &= \EE[f(A_{s,t})] =\EE[f(A_{s-t,0})] =\int_0^{s-t} \EE [g(A_{u,0})] du. 
\end{align*}
We differentiate above with respect to $t$, apply Ito's lemma and use the inequalities in \eqref{A_ineq}, then
\begin{align}
 \left| \tfrac{\partial}{\partial t} \EE \left[ \tfrac{\partial Z_{s}^{(t, x)}}{\partial x} F (Y_{s}^{(t, x)}) \right] \right|  = \Big| \EE [g(A_{s-t,0})] \Big| \leq  \left | g(1) \right| + C(s-t). \label{hform}
 \end{align}
Since $g(1)= \gamma(1-\gamma)\sigma^2(y_M -x) F(x) + (1-\gamma)\sigma^2 x(1-x) F'(x)$, 
\eqref{hform} implies \eqref{Derivative_of_expectation_process_1} and \eqref{Derivative_of_expectation_process_2}.
\end{proof}

\medskip

Recall Notation~\ref{notation} (ii) and (iii). 
When $\epsilon=0$ and $\lambda=\infty$, the value function is $v^{0}(t)$ in \eqref{merton_pde} and optimal fraction is $y_M=\tfrac{\mu}{\gamma \sigma^2}$. When $\epsilon=0$ and $\lambda<\infty$, we denote the value function as $v^{SO,\lambda}(t,x)$ and optimal fraction as $\hat y^{SO, \lambda}(t)$. Similarly, we denote $L^{SO,\lambda}(t,x):=L(t,x)|_{\epsilon=0}$, where $L$ is defined in \eqref{Supremum_term_L}. Observe that 
\begin{align}
\hat y^{SO, \lambda}(t)= \uy(t)\big|_{\epsilon=0}= \by(t) \big|_{\epsilon=0} \quad \textrm{and} \quad 
L^{SO, \lambda}(t, x) = v^{SO, \lambda}(t, \hat{y}^{SO, \lambda}(t)) \quad \textrm{for all $x\in [0,1]$}. \label{SO_eq}
\end{align}
The following results can be found in \cite{matsumotopower}: there is a constant $C>0$ independent of $(t,\lambda)\in [0,T)\times [1,\infty)$ such that
\begin{align}
    &  \left\vert \hat y^{SO, \lambda}(t) - y_M \right\vert \leq \tfrac{C}{\lambda}, \label{Order_of_hat_y0_and_mean_variance_ratio} \\
  &  \left\vert v^{SO,\lambda}(t,\hat y^{SO, \lambda}(t)) - v^{0}(t) \right\vert \leq \tfrac{C}{\lambda},\label{Order_of_v_0_lambda_and_v_infinity}\\
  &  \lim_{\lambda \rightarrow \infty} \tfrac{\lambda \left( v^{0}(t) - v^{SO, \lambda}(t, y_M) \right)}{1 - \gamma} = \tfrac{\gamma \sigma^{4} y_M^{2} (1 - y_M)^{2}}{2} \cdot e^{Q(y_M) (T - t)} (T - t). \label{Liquid_Illiquid_difference}
\end{align}

For the later analysis, we provide more estimates in this direction.

\begin{lemma}[Asymptotics for $\epsilon=0$ case]
The functions $v^{SO, \lambda}(t, x)$, $v_{x}^{SO, \lambda}(t, x)$, $v_{x x}^{SO, \lambda}(t, x)$ obtained by substitution $\epsilon = 0$ can be continuously extended to $x = 0$ and $x = 1$. 
There exist positive constants $C, \underline{C}, \overline{C}$ independent of $(t, x, \lambda)\in [0,T]\times (0,1)\times [1,\infty)$ such that 
\begin{align}
  &    \left|  \tfrac{\partial^n}{\partial x^n} v^{SO, \lambda}(t, x) \right|  \leq \tfrac{C}{\lambda} \quad \textrm{for} \quad n\in \N, \label{Dv_0_lambda_order} \\
  &  \underline{C}\cdot \min \{ 1, \lambda (T - t) \} \leq -\tfrac{\lambda v_{x x}^{SO, \lambda}(t, x)}{1 - \gamma} \leq \overline{C} \cdot \min \{ 1, \lambda (T - t) \}. \label{DDv_0_lambda_order}
\end{align} 
Furthermore, for $(t,x) \in [0, T)\times (0,1)$, 
\begin{align}
 & \lim_{\lambda \rightarrow \infty} \tfrac{\lambda v_{x x}^{SO, \lambda}(t, x)}{1 - \gamma} = - \gamma \sigma^{2} v^{0}(t).  \label{Limit_of_lambda_DDv}
\end{align}
\end{lemma}

\begin{proof}
Throughout this proof, $C>0$ is a generic constant independent of $(t,x,\lambda)\in [0,T)\times (0,1) \times [1,\infty)$  that may differ line by line.
Using \eqref{SO_eq}, the representation in \eqref{v_with_optimizer} for $\epsilon=0$ becomes
\begin{align}
  v^{SO,\lambda} (t, x) & = e^{- \lambda (T - t)} \EE \left[ Z_{T}^{(t, x)} \right] + \lambda \int_{t}^{T} e^{- \lambda (s - t)} v^{SO, \lambda}(s, \hat{y}^{SO, \lambda}(s))  \EE \left[ Z_{s}^{(t, x)} \right] d s. \label{v_with_optimizer_0} 
\end{align}
We take derivative with respect to $x$ above. Using \lemref{YZ_bound1} and the dominated convergence theorem, we put the derivative inside of the integrals: for $n\in \N$,
\begin{align}
    \tfrac{\partial^n}{\partial x^n} v^{SO, \lambda}(t, x) & =  e^{- \lambda (T - t)} \EE \Big[ \tfrac{\partial^n Z_{T}^{(t, x)}}{\partial x^n} \Big] + \lambda \int_{t}^{T} e^{- \lambda (s - t)} v^{SO, \lambda}(s, \hat{y}^{SO, \lambda}(s)) \EE \left[ \tfrac{\partial^n Z_{s}^{(t, x)}}{\partial x^n} \right] d s.\label{v_derivative_0_orders}
\end{align}
The above equality, \lemref{YZ_bound1}, \eqref{YZ_derivative} and the dominated convergence theorem enable us to conclude that
$v^{SO, \lambda}(t, x)$, $v_{x}^{SO, \lambda}(t, x)$, $v_{x x}^{SO, \lambda}(t, x)$ can be continuously extended to $x=0$ and $x=1$.

Observe that 
\begin{align*}
&x e^{-x} \leq \min\{ 1, x\} \quad \textrm{and}\quad  0\leq 1- e^{-x} - x e^{-x} \leq \min \{ 1, x\} \quad \textrm{for}\quad x\geq 0, \\ 
&x e^{-x} \geq \tfrac{x}{e} \quad \textrm{for} \quad 0\leq x \leq 1, \quad 1- e^{-x} - x e^{-x}\geq 1-\tfrac{2}{e}>0 \quad \textrm{for} \quad x\geq 1. 
\end{align*}
By the above inequalities, for positive constants $c_1$ and $c_2$, we can find positive constants $\underline{c}$ and $\overline{c}$ such that
\begin{align}
\underline{c} \cdot \min \{ 1, x\}  \leq  c_1 x e^{-x} + c_2 \left(1- e^{-x} - x e^{-x} \right) \leq \overline{c} \cdot \min \{ 1, x\} \quad \textrm{for}\quad x\geq 0. \label{ineq_1}
\end{align}

We apply \eqref{Expectation_Time_order} to \eqref{v_derivative_0_orders} and obtain
\begin{align*}
  \left|  \tfrac{\partial^n}{\partial x^n} v^{SO, \lambda}(t, x) \right|  &\leq  C(T-t) e^{- \lambda (T - t)} + C\lambda \int_t^T e^{-\lambda(s-t)}(s-t)ds
  \leq \tfrac{C}{\lambda},
\end{align*}
where the second inequality is due to \eqref{ineq_1}.
Thus, we conclude \eqref{Dv_0_lambda_order}.

The expressions in \eqref{AYZ} and \eqref{YZ_derivative} imply
\begin{align*}
& \min \{ A_{s, t}^{- 1 - \gamma}, 1 \} \leq \left( x A_{s, t} + 1 - x \right)^{- 1 - \gamma}  \leq \max \{ A_{s, t}^{- 1 - \gamma}, 1 \} \quad \textrm{for}\quad  0\leq x\leq 1,\\
&\lim_{s\downarrow t}  \EE \left[ \min \{ A_{s, t}^{- 1 - \gamma}, 1 \} \left( \tfrac{A_{s, t} - 1}{\sqrt{s - t}} \right)^{2} \right] = \lim_{s\downarrow t} \EE \left[ \max \{ A_{s, t}^{- 1 - \gamma}, 1 \} \left( \tfrac{A_{s, t} - 1}{\sqrt{s - t}} \right)^{2} \right] =\sigma^2>0,\\
&-\tfrac{1}{(1 - \gamma) (s - t)} \EE \left[ \tfrac{\partial^{2} Z_{s}^{(t, x)}}{\partial x^{2}} \right] = \gamma \, \EE \left[ \left( x A_{s, t} + 1 - x \right)^{- 1 - \gamma} \left( \tfrac{A_{s, t} - 1}{\sqrt{s - t}} \right)^{2} \right]. 
\end{align*}
We combine the above inequalities, limits and expression to conclude that
\begin{align}
\lim_{s\downarrow t} \left(-\tfrac{1}{(1 - \gamma) (s - t)} \EE \left[ \tfrac{\partial^{2} Z_{s}^{(t, x)}}{\partial x^{2}} \right] \right) = \gamma \sigma^2>0, \label{Zxx_limit}
\end{align}
and there exist positive constants $\underline{c}$ and $\overline{c}$ independent of $(t,s,x,\lambda)$ such that 
\begin{align}
\underline{c} \leq -\tfrac{1}{(1 - \gamma) (s - t)} \EE \left[ \tfrac{\partial^{2} Z_{s}^{(t, x)}}{\partial x^{2}} \right] \leq \overline{c}. \label{Zxx_bounds}
\end{align}

From \eqref{v_derivative_0_orders} for $n=2$, we obtain the following expression:
\begin{align*}
  - \tfrac{\lambda v_{x x}^{SO, \lambda}(t, x)}{1 - \gamma} & = \lambda (T - t) e^{- \lambda (T - t)} \cdot \left(  -\tfrac{1}{(1 - \gamma) (T - t)} \EE \Big[ \tfrac{\partial^{2} Z_{T}^{(t, x)}}{\partial x^{2}} \Big]  \right) \\
  &\quad + \lambda^2 \int_{t}^{T} e^{- \lambda (s - t)} (s - t) \cdot v^{SO, \lambda}(s, \hat{y}^{SO, \lambda}(s)) \cdot  \left(  -\tfrac{1}{(1 - \gamma) (s - t)} \EE \left[ \tfrac{\partial^{2} Z_{s}^{(t, x)}}{\partial x^{2}} \right] \right) d s 
 \end{align*}
We apply the inequalities in \lemref{v_concave} (ii), \eqref{ineq_1} and \eqref{Zxx_bounds} to the above expression to conclude \eqref{DDv_0_lambda_order}. Finally, in the above expression, we substitute $u=\lambda(s-t)$ and let $\lambda\to \infty$:
\begin{align*}
\lim_{\lambda\to \infty}   - \tfrac{\lambda v_{x x}^{SO, \lambda}(t, x)}{1 - \gamma} &= \lim_{\lambda\to \infty}  \int_0^{\lambda (T-t)} e^{-u} u   \cdot v^{SO, \lambda}(\tfrac{u}{\lambda}+t, \hat{y}^{SO, \lambda}(\tfrac{u}{\lambda}+t)) \cdot  \Big(  -\tfrac{1}{(1 - \gamma) \frac{u}{\lambda}} \EE \Big[ \tfrac{\partial^{2} Z_{u/\lambda+t}^{(t, x)}}{\partial x^{2}} \Big] \Big) du\\
&=\int_0^\infty e^{-u} u \,  du  \cdot v^{0}(t) \cdot \gamma \sigma^2 = \gamma \sigma^2  v^{0}(t),
\end{align*}
where the second equality is due to \eqref{Order_of_v_0_lambda_and_v_infinity}, \eqref{Zxx_limit}, \eqref{Zxx_bounds} and the dominated convergence theorem. 
\end{proof}

\begin{lemma}[Estimates of $v^\epsilon, v_x^\epsilon,v_{xx}^\epsilon$ and $v_{xxx}^\epsilon$] \label{Various_orders}
Let \assref{ass} hold. 
Let $C>0$ be a generic constant independent of $(t,s, x,\epsilon)\in [0,T)\times [t,T)\times(0,1) \times (0,1)$ (also independent of $\lambda$ due to relation $\lambda=c \, \epsilon^{-\frac{2}{3}}$) that may differ line by line. Then, the followings hold:
\begin{align}
&x (1 - x) \left\vert v_{x x}^{\epsilon}(t, x) - v_{x x}^{SO, \lambda}(t, x) \right\vert \leq C \epsilon^{\frac{2}{3}},  \label{DDv_ep_DDv_0_difference}\\
&\left\vert v_{x}^{\epsilon}(t, x) - v_{x}^{SO, \lambda}(t, x) \right\vert \leq C \epsilon, \label{Dv_ep_Dv_0_difference}\\
&\left\vert v^{\epsilon}(t, x) - v^{SO, \lambda}(t, x) \right\vert \leq C  \epsilon^{\frac{2}{3}}, \label{v_ep_v_0_difference_more_strong}\\
&-\tfrac{ v_{x x}^{\epsilon}(t, x) }{1 - \gamma} \geq C(\epsilon^{\frac{2}{3}}-\epsilon)\,  \min \{ \lambda (T - t), 1 \}, \label{Upper_bound_of_DDv}\\
&x^2 (1 - x)^2 \left\vert v_{xxx}^{\epsilon}(t, x) - v_{xxx}^{SO, \lambda}(t, x) \right\vert \leq C \epsilon^{\frac{1}{3}}, \label{DDDv_ep_DDDv_0_difference}
\end{align}
where $v_{xxx}^{\epsilon}(t, x)$ exists and is continuous in $(t,x)\in [0,T) \times (0,1)$.
\end{lemma}

\begin{proof}
For convenience, let $g^\epsilon:[0,1]^2 \to \R$ be
\begin{align}
g^\epsilon(x,y):=\Big( \big( \tfrac{1 + \epsilon x}{1 + \epsilon y} \big)^{1 - \gamma} - 1 \Big) 1_{ \left\{ y>x \right\} }  + \Big( \big( \tfrac{1 - \epsilon x}{1 - \epsilon y} \big)^{1 - \gamma} - 1 \Big) 1_{ \left\{ y<x \right\} }. \label{g_def}
\end{align}
Then, the mean value theorem (consider $\epsilon$ as a variable) produces
\begin{align}
\big| g^\epsilon(x,y) \big| \leq C |x-y| \,\epsilon \quad \textrm{for}\quad x,y\in [0,1].  \label{g_bound}
\end{align}
From the expression of $L$ in \eqref{L_sup_exp}, we obtain
\begin{align}
   \left\vert L^{\epsilon} \left( t, x \right) - L^{SO, \lambda} \left( t, x \right) \right\vert & \leq \sup_{y \in [0, 1]} \left\vert v^{\epsilon}(t, y) g^\epsilon (x,y)\right\vert + \sup_{y \in [0, 1]} \big\vert v^{\epsilon}(t, y) - v^{SO, \lambda}(t, y) \big\vert \nonumber\\  
  & \leq C  \epsilon + \sup_{y \in [0, 1]} \big\vert v^{\epsilon}(t, y) - v^{SO, \lambda}(t, y) \big\vert,   \label{Difference_between_L_in_epsilon}
\end{align}
where the second inequality is due to \eqref{Range_of_v} and \eqref{g_bound}.
The above inequality and \eqref{v_with_optimizer} produce
\begin{align}
\left| v^{\epsilon}(t, x) - v^{SO, \lambda}(t, x)\right| & =\left| \lambda \int_{t}^{T} e^{- \lambda (s - t)} \EE \left[ Z_{s}^{(t, x)} \left( L^{\epsilon} ( s, Y_{s}^{(t, x)} ) - L^{SO, \lambda} ( s, Y_{s}^{(t, x)} )\right)   \right]ds \right|  \nonumber\\
& \leq  \lambda \int_{t}^{T} e^{- \lambda (s - t)} \EE \left[ Z_{s}^{(t, x)}\right] \Big( C  \epsilon + \sup_{y \in [0, 1]} \left\vert v^{\epsilon}(s, y) - v^{SO, \lambda}(s, y) \right\vert  \Big)   ds \nonumber \\
  & \leq \lambda \int_{t}^{T} e^{(C - \lambda) (s - t)} \Big( C\epsilon +\sup_{y \in [0, 1]} \left\vert v^{\epsilon}(s, y) - v^{SO, \lambda}(s, y) \right\vert \Big) d s,\label{v_ep_v_0_difference_upper_bound}
\end{align}
where the last inequality is due to $\EE [ Z_{s}^{(t, x)} ] \leq 1+ C(s-t) \leq e^{C (s - t)}$ by \eqref{Expectation_Time_order}. For convenience, we define $f(t):=e^{(C-\lambda)t} \sup_{y \in [0, 1]} \left\vert v^{\epsilon}(t, y) - v^{SO, \lambda}(t, y) \right\vert$. Then, inequality \eqref{v_ep_v_0_difference_upper_bound} can be written as
\begin{align*}
f(t) \leq C \epsilon \int_t^T \lambda e^{(C-\lambda)s}ds + \int_t^T \lambda f(s) ds.
\end{align*}
Since $f$ is measurable due to \lemref{meas_lem}, we apply Gronwall's inequality (see \lemref{Gronwall_lemma}) to obtain
\begin{align*}
f(t) &\leq C \epsilon \int_t^T \lambda e^{(C-\lambda)s}ds + \int_t^T \left(C \epsilon \int_s^T \lambda e^{(C-\lambda)u}du \right) \lambda e^{\lambda(s-t)}ds\\
&=\tfrac{C \epsilon \lambda}{\lambda - C} \cdot \left( e^{(C - \lambda) t} - e^{(C - \lambda) T} \right) + \tfrac{C \epsilon \lambda e^{- \lambda t}}{\lambda - C} \cdot \left( \tfrac{\lambda \left( e^{C T} - e^{C t} \right)}{C} - e^{C T} \left( 1 - e^{- \lambda (T - t)} \right) \right).
\end{align*}
We apply $\lambda=c \, \epsilon^{-\frac{2}{3}}$ to the above inequality and conclude that 
\begin{align}
\sup_{y \in [0, 1]} \left\vert v^{\epsilon}(t, y) - v^{SO, \lambda}(t, y) \right\vert = e^{(\lambda-C) t} f(t) \leq C \epsilon^{\frac{1}{3}}. \label{v_estimate_weak}
\end{align}
Notice that \eqref{v_ep_v_0_difference_more_strong} is not obtained yet. 

Since the value function $V$ in \eqref{Definition_of_value_function} should decrease in $\epsilon$, we have $\tfrac{v^{\epsilon}(t, x)}{1 - \gamma} \leq \tfrac{v^{SO, \lambda}(t, x)}{1 - \gamma}$. Therefore, the expression of $L$ in \eqref{L_sup_exp}, together with \eqref{Difference_between_L_in_epsilon} and \eqref{v_estimate_weak}, implies that 
\begin{align}
-C \epsilon^{\frac{1}{3}}\leq \tfrac{L^{\epsilon}(t, x)-L^{SO, \lambda}(t, x)}{1 - \gamma} \leq 0. \label{L_ep_lambda-L_0_lambda}
\end{align}

{\bf (Proof of \eqref{DDv_ep_DDv_0_difference})}

We take derivative with respect to $x$ in \eqref{Dv_with_optimizer}, and put the derivative inside of the expectation (\lemref{YZ_bound1}, $\lVert  L \rVert_\infty<\infty$ and \eqref{Lx_bound} allow us to do this) to obtain the following expression: 
\begin{align}
  v_{x x}^{\epsilon}(t, x) & = e^{- \lambda (T - t)} \EE \Big[ \tfrac{\partial^{2} Z_{T}^{(t, x)}}{\partial x^{2}} \Big]+ \lambda \int_{t}^{T} e^{- \lambda (s - t)} \EE \Big[ \tfrac{\partial^{2} Z_{s}^{(t, x)}}{\partial x^{2}} L^{\epsilon} ( s, Y_{s}^{(t, x)} ) + \tfrac{\partial Z_{s}^{(t, x)}}{\partial x} L^{\epsilon}_{x} ( s, Y_{s}^{(t, x)} ) \tfrac{\partial Y_{s}^{(t, x)}}{\partial x} \Big]  d s \nonumber\\
  &\quad + \tfrac{\partial}{\partial x}\left(\lambda\int_t^T  e^{-\lambda(s-t)} \mathbb{E} \left[ Z_{s}^{(t, x)} L^{\epsilon}_x ( s, Y_{s}^{(t, x)} ) \tfrac{\partial Y_{s}^{(t, x)}}{\partial x} \right] ds \right). \label{DDv_with_optimizer}
\end{align}
This implies that 
\begin{align}
   \left| v_{x x}^{\epsilon}(t, x) - v_{x x}^{SO, \lambda}(t, x) \right| & \leq  \left| \lambda \int_{t}^{T} e^{- \lambda (s - t)} \EE \left[ \tfrac{\partial^{2} Z_{s}^{(t, x)}}{\partial x^{2}} \left( L^{\epsilon} ( s, Y_{s}^{(t, x)} ) - L^{SO,\lambda} ( s, Y_{s}^{(t, x)} )\right)  \right]  d s \right|  \nonumber\\
  & \quad +  \left| \lambda \int_{t}^{T} e^{- \lambda (s - t)} \EE \left[  \tfrac{\partial Z_{s}^{(t, x)}}{\partial x} L^{\epsilon}_{x} ( s, Y_{s}^{(t, x)} ) \tfrac{\partial Y_{s}^{(t, x)}}{\partial x} \right]  d s \right|  \nonumber\\
  &\quad  + \left| \tfrac{\partial}{\partial x}\left(\lambda\int_t^T  e^{-\lambda(s-t)} \mathbb{E} \left[ Z_{s}^{(t, x)} L^{\epsilon}_x ( s, Y_{s}^{(t, x)} ) \tfrac{\partial Y_{s}^{(t, x)}}{\partial x} \right] ds \right) \right|. \nonumber
\end{align}
In the right side of the above inequality, the first term is bounded by $C \lambda \int_{t}^{T} e^{- \lambda (s - t)} \epsilon^{\frac{1}{3}}(s-t)  d s $ due to \eqref{Boundedness_of_Expectation_of_Multiplied_term} and \eqref{L_ep_lambda-L_0_lambda}, the second term is bounded by $C \lambda \int_{t}^{T} e^{- \lambda (s - t)} \epsilon \,d s $ due to  \lemref{YZ_bound1} and \eqref{Lx_bound} and the third term is bounded by $\tfrac{C}{x(1-x)} \epsilon^{\frac{2}{3}}$ due to \lemref{Boundedness_of_various_expectation} (ii) (with $n=1$ and $F=L_x^{\epsilon}$) and \eqref{Lx_bound}. These bounds and $\lambda=c \, \epsilon^{-\frac{2}{3}}$ produce \eqref{DDv_ep_DDv_0_difference}.

\medskip

{\bf (Proof of \eqref{Dv_ep_Dv_0_difference})}

The expression of $L_{x}$ in \eqref{Lx_expression} implies that $L_{x}^\epsilon(t,x)$ is continuously differentiable with respect to $x$, except $x=\uy^{\epsilon}(t)$ and $x=\by^{\epsilon}(t)$. To be specific,
\begin{align}
   L_{xx}^{\epsilon}(t, x) = \begin{cases}
  - \tfrac{\gamma (1 - \gamma) \epsilon^{2} v^{\epsilon}(t, \uy^{\epsilon}(t)) }{(1 + \epsilon x)^{2}} \left( \tfrac{1 + \epsilon x}{1 + \epsilon \uy^{\epsilon}(t)} \right)^{1 - \gamma} &\textrm{if}\quad x \in (0, \uy^{\epsilon}(t)), \\
    v_{x x}^{\epsilon}(t,x) &\textrm{if}\quad x \in ( \uy^{\epsilon}(t), \by^{\epsilon}(t) ), \\
   - \tfrac{\gamma (1 - \gamma) \epsilon^{2} v^{\epsilon}(t, \by^{\epsilon}(t)) }{(1 - \epsilon x)^{2}} \left( \tfrac{1 - \epsilon x}{1 - \epsilon \by^{\epsilon}(t)} \right)^{1 - \gamma} &\textrm{if}\quad x \in (\by^{\epsilon}(t), 1). 
   \end{cases} \label{Lxx_expression}
\end{align}
\lemref{v_concave} (i), \eqref{Dv_0_lambda_order} and \eqref{DDv_ep_DDv_0_difference} imply 
\begin{align}
- \tfrac{C \epsilon^{\frac{2}{3}}}{x(1-x)} \leq  \tfrac{v_{x x}^{\epsilon} \left( t, x \right)}{1-\gamma} \leq 0. \label{vxx_twosides}
\end{align}
We combine \lemref{v_concave} (ii), \eqref{Lxx_expression} and \eqref{vxx_twosides} to obtain
\begin{align}
-\tfrac{C \epsilon^{\frac{2}{3}}}{x(1-x)}\leq \tfrac{ L_{xx}^{\epsilon}(t, x)}{1-\gamma} \leq 0. \label{Lxx_twosides}
\end{align}

To apply \lemref{Boundedness_of_various_expectation} (iii), we set $F(y)=L^{\epsilon} ( s, y) - L^{SO, \lambda} ( s,y)$ and observe that $F'(y)=L^{\epsilon}_x ( s, y)$ and $F''(y)=L^{\epsilon}_{xx} ( s, y)$ except $y\in \{\uy^\epsilon(s), \by^\epsilon(s)\}$. 
We check that \eqref{F_bounds} is satisfied due to \eqref{Range_of_v}, \eqref{Lx_bound} and \eqref{Lxx_twosides}.
Therefore, \lemref{Boundedness_of_various_expectation} (iii) is applicable and \eqref{Derivative_of_expectation_process_2}, together with \eqref{L_ep_lambda-L_0_lambda}, \eqref{Lx_bound} and \eqref{Lxx_twosides}, produces
\begin{align}
\left| \EE \left[ \tfrac{\partial Z_{s}^{(t, x)}}{\partial x} \left( L^{\epsilon} ( s, Y_{s}^{(t, x)}) - L^{SO, \lambda} ( s, Y_{s}^{(t, x)}) \right) \right] \right| \leq C(\epsilon^{\frac{1}{3}} +\epsilon )(s-t)+ C(s-t)^2.  \label{L_bound_1/3}
\end{align}
We use the stochastic representation \eqref{Dv_with_optimizer} to obtain 
\begin{align*}
  \left| v_{x}^{\epsilon}(t, x) - v_{x}^{SO, \lambda}(t, x)  \right| & \leq \left| \lambda \int_{t}^{T} e^{- \lambda (s - t)} \EE \left[ \tfrac{\partial Z_{s}^{(t, x)}}{\partial x} \left( L^{\epsilon} ( s, Y_{s}^{(t, x)} ) - L^{SO, \lambda} 
  ( s, Y_{s}^{(t, x)} 
  ) \right) \right] d s \right| \nonumber \\
  & \quad + \left| \lambda \int_{t}^{T} e^{- \lambda (s - t)} \EE \left[ Z_{s}^{(t, x)} L_{x}^{\epsilon} ( s, Y_{s}^{(t, x)} ) \tfrac{\partial Y_{s}^{(t, x)}}{\partial x} \right] ds \right|. \nonumber
  \end{align*}
The first term in the right-hand side is bounded by $ C\lambda \int_{t}^{T} e^{- \lambda (s - t)}  \big((\epsilon^{\frac{1}{3}} +\epsilon ) (s-t)+(s-t)^2\big) ds $ due to \eqref{L_bound_1/3} and the second term is bounded by $ C  \lambda \int_{t}^{T} e^{- \lambda (s - t)} \epsilon\, ds$ due to \lemref{YZ_bound1} and \eqref{Lx_bound}. These bounds and $\lambda=c \, \epsilon^{-\frac{2}{3}}$ produce \eqref{Dv_ep_Dv_0_difference}.

\medskip

{\bf (Proof of \eqref{v_ep_v_0_difference_more_strong})}

Using \eqref{SO_eq} and  \eqref{L_ep_lambda-L_0_lambda}, we obtain
\begin{align}
  0&\leq  \tfrac{1}{1-\gamma} \left( L^{SO, \lambda} ( s, Y_{s}^{(t, \hat{y}^{SO, \lambda}(t))}) - L^{\epsilon} ( s, Y_{s}^{(t, \hat{y}^{SO, \lambda}(t))} ) \right)\nonumber\\
  & = \tfrac{1}{1-\gamma}\left( v^{SO, \lambda}(s, \hat{y}^{SO, \lambda}(s)) - L^{\epsilon} ( s, Y_{s}^{(t, \hat{y}^{SO, \lambda}(t))}) \right) \nonumber\\
  &\leq \tfrac{1}{1-\gamma}  \left( v^{SO, \lambda}(s, \hat{y}^{SO, \lambda}(s))-v^{\epsilon}(s, \hat{y}^{SO, \lambda}(s))  -   v^{\epsilon}(s, \hat{y}^{SO, \lambda}(s)) g^\epsilon \big(Y_{s}^{(t, \hat{y}^{SO, \lambda}(t))},\hat{y}^{SO, \lambda}(s) \big)  \right)\nonumber  \\
  & \leq C  \epsilon \, \left| Y_{s}^{(t, \hat{y}^{SO, \lambda}(t))}-\hat{y}^{SO, \lambda}(s) \right| + \tfrac{v^{SO, \lambda}(s, \hat{y}^{SO, \lambda}(s)) - v^{\epsilon}(s, \hat{y}^{SO, \lambda}(s))}{1 - \gamma}  \nonumber\\
    & \leq C  \epsilon \, \left(\left| Y_{s}^{(t, \hat{y}^{SO, \lambda}(t))}-\hat{y}^{SO, \lambda}(t) \right| + \tfrac{1}{\lambda} \right)  + \tfrac{v^{SO, \lambda}(s, \hat{y}^{SO, \lambda}(s)) - v^{\epsilon}(s, \hat{y}^{SO, \lambda}(s))}{1 - \gamma} , \label{L_bound_strong}
  \end{align}
where the second inequality is due to \eqref{L_sup_exp} and \eqref{g_def}, the third inequality is due to \eqref{g_bound} and the last inequality is due to \eqref{Order_of_hat_y0_and_mean_variance_ratio}. 

By the same way as we prove \lemref{Boundedness_of_various_expectation} (i), we can check that $\EE[(Y_s^{(t,x)}-x)^2] \leq C(s-t)$. Hence, by \lemref{Boundedness_of_various_expectation} (i) and H$\ddot{o}$lder's inequality, 
we produce
\begin{align}
\EE \left[  Z_{s}^{(t, x)} \cdot \big| Y_{s}^{(t, x)}-x \big|  \right] &\leq \sqrt{(1+C(s-t)) C(s-t)} \leq C e^{C(s-t)} \sqrt{s-t}. \label{Z(Y-x)}  
\end{align}
The stochastic representation in \eqref{v_with_optimizer} produces
\begin{align*}
0&\leq\tfrac{1}{1-\gamma}\Big(v^{SO, \lambda}(t, \hat{y}^{SO, \lambda}(t)) - v^{\epsilon}(t, \hat{y}^{SO, \lambda}(t))\Big)\nonumber \\
& = \lambda \int_{t}^{T} e^{- \lambda (s - t)} \EE \left[ Z_{s}^{(t, \hat{y}^{SO, \lambda}(t))}  \tfrac{1}{1-\gamma} \left( L^{SO, \lambda} ( s, Y_{s}^{(t, \hat{y}^{SO, \lambda}(t))} ) - L^{\epsilon} ( s, Y_{s}^{(t, \hat{y}^{SO, \lambda}(t))} )\right)   \right]ds   \nonumber\\
& \leq \lambda \int_{t}^{T} e^{- \lambda (s - t)} \EE \left[ Z_{s}^{(t, \hat{y}^{SO, \lambda}(t))}   C  \epsilon \, \left(\left| Y_{s}^{(t, \hat{y}^{SO, \lambda}(t))}-\hat{y}^{SO, \lambda}(t) \right| + \tfrac{1}{\lambda} \right)    \right] ds  \nonumber\\
& \quad  +\lambda \int_{t}^{T} e^{- \lambda (s - t)} \EE \left[ Z_{s}^{(t, \hat{y}^{SO, \lambda}(t))}\right]   \tfrac{v^{SO, \lambda}(s, \hat{y}^{SO, \lambda}(s)) - v^{\epsilon}(s, \hat{y}^{SO, \lambda}(s))}{1 - \gamma}     ds \nonumber \\
  &\leq C \lambda \int_{t}^{T} e^{(C - \lambda) (s - t)} \left( \epsilon \left( \sqrt{s - t} + \tfrac{1}{\lambda} \right) +\tfrac{v^{SO, \lambda}(s, \hat{y}^{SO, \lambda}(s)) - v^{\epsilon}(s, \hat{y}^{SO, \lambda}(s))}{1 - \gamma} \right) d s,
\end{align*}
where the first and second inequalities are due to \eqref{L_bound_strong} and the last inequality is due to \eqref{Z(Y-x)} and \eqref{YZ_bound1}.
Since the map $t\mapsto v^{SO, \lambda}(t, \hat{y}^{SO, \lambda}(t)) - v^{\epsilon}(t, \hat{y}^{SO, \lambda}(t))$ is measurable due to \lemref{meas_lem},
by the same way as we treated \eqref{v_ep_v_0_difference_upper_bound}, we apply Gronwall's inequality (see \lemref{Gronwall_lemma}), \lemref{Order_of_exponential_times} and $\lambda=c \, \epsilon^{-\frac{2}{3}}$ to the above inequality to obtain
\begin{align*}
\left\vert v^{SO, \lambda}(t, \hat{y}^{SO, \lambda}(t)) - v^{\epsilon}(t, \hat{y}^{SO, \lambda}(t)) \right\vert \leq C \epsilon^{\frac{2}{3}}.
\end{align*}
The above inequality and \eqref{Dv_ep_Dv_0_difference} produce \eqref{v_ep_v_0_difference_more_strong}: for $(t,x)\in [0,T)\times (0,1)$,
\begin{align*}
\left|  v^{\epsilon}(t, x) - v^{SO, \lambda}(t, x) \right| &= \Big|v^{\epsilon}(t, \zeta) - v^{SO, \lambda}(t, \zeta) + \int_{\zeta}^{x} \left( v_{x}^{\epsilon}(t, z) - v_{x}^{SO, \lambda}(t, z) \right) d z \Big| \bigg|_{\zeta=\hat{y}^{SO, \lambda}(t)} \leq C \epsilon^{\frac{2}{3}}.
\end{align*}

{\bf (Proof of \eqref{Upper_bound_of_DDv})}

By \eqref{Difference_between_L_in_epsilon} and \eqref{v_ep_v_0_difference_more_strong}, we have
\begin{align}
\left| L^{\epsilon} ( t, x) - L^{SO,\lambda} ( t, x )\right| \leq C \epsilon^{\frac{2}{3}}. \label{L_difference_estimate}
\end{align}
Since $\PP\big(Y_{s}^{(t, x)}\notin  \{\uy^{\epsilon}(s), \by^{\epsilon}(s)\} \big) =1$ for $t<s$,  \lemref{YZ_bound1} and \eqref{Lxx_twosides} and the continuity of $x\mapsto L^\epsilon_x(t,x)$ allow us to put the derivative inside of the expectation in \eqref{DDv_with_optimizer}:
\begin{align}
  v_{x x}^{\epsilon}(t, x) & =  \lambda \int_{t}^{T} e^{- \lambda (s - t)} \EE \left[ \tfrac{\partial^{2} Z_{s}^{(t, x)}}{\partial x^{2}} L^{\epsilon} ( s, Y_{s}^{(t, x)} ) + \left(2\tfrac{\partial Z_{s}^{(t, x)}}{\partial x} \tfrac{\partial Y_{s}^{(t, x)}}{\partial x} + Z_{s}^{(t, x)} \tfrac{\partial^2 Y_{s}^{(t, x)}}{\partial x^2}\right)L^{\epsilon}_{x} ( s, Y_{s}^{(t, x)} )  \right]  d s \nonumber\\
&\quad + e^{- \lambda (T - t)} \EE \left[ \tfrac{\partial^{2} Z_{T}^{(t, x)}}{\partial x^{2}} \right] + \lambda\int_t^T  e^{-\lambda(s-t)} \mathbb{E} \left[ Z_{s}^{(t, x)} L^{\epsilon}_{xx} ( s, Y_{s}^{(t, x)} ) \left(\tfrac{\partial Y_{s}^{(t, x)}}{\partial x}\right)^2 \right] ds.
\label{DDv_with_optimizer2}
\end{align}
The above expression produces the following equality:
\begin{align}
  \tfrac{ v_{x x}^{\epsilon}(t, x) - v_{x x}^{SO, \lambda}(t, x)}{1 - \gamma}  & =\lambda\int_t^T  e^{-\lambda(s-t)} \mathbb{E} \left[ Z_{s}^{(t, x)} \tfrac{L^{\epsilon}_{xx} ( s, Y_{s}^{(t, x)} )}{1-\gamma}\left(\tfrac{\partial Y_{s}^{(t, x)}}{\partial x}\right)^2 \right] ds \nonumber \\
  &\quad + \lambda \int_{t}^{T} e^{- \lambda (s - t)} \EE \left[ \tfrac{\partial^{2} Z_{s}^{(t, x)}}{\partial x^{2}} \tfrac{L^{\epsilon} ( s, Y_{s}^{(t, x)} )- L^{SO,\lambda} ( s, Y_{s}^{(t, x)} )}{1-\gamma}\right] ds  \nonumber \\
  &\quad +\lambda \int_{t}^{T} e^{- \lambda (s - t)} \EE \left[  \left(2\tfrac{\partial Z_{s}^{(t, x)}}{\partial x} \tfrac{\partial Y_{s}^{(t, x)}}{\partial x} + Z_{s}^{(t, x)} \tfrac{\partial^2 Y_{s}^{(t, x)}}{\partial x^2}\right)\tfrac{L^{\epsilon}_{x} ( s, Y_{s}^{(t, x)} )}{1-\gamma}  \right]  d s. \label{vxx_difference}
\end{align}
In the right-hand side of the above equality, the first term is bounded above by $0$ due to \eqref{Lxx_twosides}, the second term is bounded above by $C\lambda\int_t^T  e^{-\lambda(s-t)}  \epsilon^{\frac{2}{3}}(s-t)ds$ due to \eqref{Boundedness_of_Expectation_of_Multiplied_term} and \eqref{L_difference_estimate} and the third term is bounded above by $C  \lambda\int_t^T  e^{-\lambda(s-t)}\epsilon\, ds$ due to and \eqref{Lx_bound} and \lemref{YZ_bound1}. These bounds and $\lambda=c \, \epsilon^{-\frac{2}{3}}$, together with $1-e^{-\lambda(T-t)} \leq  \min \{ \lambda (T - t), 1 \}$, produce
\begin{align*}
- \tfrac{ v_{x x}^{\epsilon}(t, x)}{1 - \gamma}
\geq - \tfrac{ v_{x x}^{SO, \lambda}(t, x)}{1 - \gamma}- C \epsilon \min \{ \lambda (T - t), 1 \}.
\end{align*}
We combine the above inequality and \eqref{DDv_0_lambda_order} to conclude \eqref{Upper_bound_of_DDv}. 

\medskip

{\bf (Proof of \eqref{DDDv_ep_DDDv_0_difference})}

By the same way as we obtain \eqref{DDv_with_optimizer}, we take derivative with respect to $x$ in  \eqref{DDv_with_optimizer2}:
\begin{align*}
& v_{xxx}^{\epsilon}(t, x)  =   e^{- \lambda (T - t)} \EE \left[ \tfrac{\partial^{3} Z_{T}^{(t, x)}}{\partial x^{3}} \right] + \tfrac{\partial}{\partial x} \left(\lambda\int_t^T  e^{-\lambda(s-t)} \mathbb{E} \left[ Z_{s}^{(t, x)} L^{\epsilon}_{xx} ( s, Y_{s}^{(t, x)} ) \left(\tfrac{\partial Y_{s}^{(t, x)}}{\partial x}\right)^2 \right] ds\right)\\
& + \lambda \int_{t}^{T} e^{- \lambda (s - t)} \EE \bigg[ \tfrac{\partial^{3} Z_{s}^{(t, x)}}{\partial x^{3}} L^{\epsilon} ( s, Y_{s}^{(t, x)} ) + \left(2\tfrac{\partial Z_{s}^{(t, x)}}{\partial x} \tfrac{\partial Y_{s}^{(t, x)}}{\partial x} + Z_{s}^{(t, x)} \tfrac{\partial^2 Y_{s}^{(t, x)}}{\partial x^2}\right)L^{\epsilon}_{xx} ( s, Y_{s}^{(t, x)} ) \tfrac{\partial Y_{s}^{(t, x)}}{\partial x} \\
&\qquad\qquad\qquad\qquad  + \left(3\tfrac{\partial^2 Z_{s}^{(t, x)}}{\partial x^2} \tfrac{\partial Y_{s}^{(t, x)}}{\partial x} 
+ 3\tfrac{\partial Z_{s}^{(t, x)}}{\partial x} \tfrac{\partial^2 Y_{s}^{(t, x)}}{\partial x^2}
+ Z_{s}^{(t, x)} \tfrac{\partial^3 Y_{s}^{(t, x)}}{\partial x^3} \right) L^{\epsilon}_{x} ( s, Y_{s}^{(t, x)} )   \bigg]  d s.
\end{align*}
This implies that
\begin{align*}
& \left| v_{xxx}^{\epsilon}(t, x)-v_{xxx}^{SO, \lambda}(t, x) \right|  
\leq \left| \tfrac{\partial}{\partial x} \left(\lambda\int_t^T  e^{-\lambda(s-t)} \mathbb{E} \left[ Z_{s}^{(t, x)} L^{\epsilon}_{xx} ( s, Y_{s}^{(t, x)} ) \left(\tfrac{\partial Y_{s}^{(t, x)}}{\partial x}\right)^2 \right] ds\right) \right|\\
&\qquad\qquad\qquad\qquad\qquad +\bigg| \lambda \int_{t}^{T} e^{- \lambda (s - t)} \EE \bigg[ \tfrac{\partial^{3} Z_{s}^{(t, x)}}{\partial x^{3}} \left( L^{\epsilon} ( s, Y_{s}^{(t, x)} )- L^{SO,\lambda} ( s, Y_{s}^{(t, x)} )\right) \\
&\qquad\qquad\qquad\qquad\qquad\qquad + \left(2\tfrac{\partial Z_{s}^{(t, x)}}{\partial x} \tfrac{\partial Y_{s}^{(t, x)}}{\partial x} + Z_{s}^{(t, x)} \tfrac{\partial^2 Y_{s}^{(t, x)}}{\partial x^2}\right)L^{\epsilon}_{xx} ( s, Y_{s}^{(t, x)} ) \tfrac{\partial Y_{s}^{(t, x)}}{\partial x} \\
&\qquad\qquad\qquad\qquad\qquad\qquad  + \left(3\tfrac{\partial^2 Z_{s}^{(t, x)}}{\partial x^2} \tfrac{\partial Y_{s}^{(t, x)}}{\partial x} 
+ 3\tfrac{\partial Z_{s}^{(t, x)}}{\partial x} \tfrac{\partial^2 Y_{s}^{(t, x)}}{\partial x^2}
+ Z_{s}^{(t, x)} \tfrac{\partial^3 Y_{s}^{(t, x)}}{\partial x^3} \right) L^{\epsilon}_{x} ( s, Y_{s}^{(t, x)} )   \bigg]  d s \bigg|.
\end{align*}
In the right-hand side of the above inequality, the first integral is bounded by $\tfrac{C \epsilon^{\frac{2}{3}}\sqrt{\lambda}}{x^2(1-x)^2}$
due to \lemref{Boundedness_of_various_expectation} (ii) (with $n=2$ and $F=L_{xx}^{\epsilon}$) and \eqref{Lxx_twosides} and the second integral is bounded by $C \lambda \int_{t}^{T} e^{- \lambda (s - t)} (\epsilon^{\frac{2}{3}}+ \tfrac{1}{x(1-x)}\epsilon^{\frac{2}{3}} + \epsilon ) ds$ due to \lemref{YZ_bound1}, \eqref{L_difference_estimate}, \eqref{Lxx_twosides} and \eqref{Lx_bound}. These bounds and $\lambda=c \, \epsilon^{-\frac{2}{3}}$ produce \eqref{DDDv_ep_DDDv_0_difference}.
\end{proof}

\begin{lemma} \label{Various_orders2}
Let \assref{ass} hold. 
Let $C>0$ be a generic constant independent of $(t,s, x,\epsilon)\in [0,T)\times [t,T)\times(0,1) \times (0,1)$ (also independent of $\lambda$ due to relation $\lambda=c \, \epsilon^{-\frac{2}{3}}$) that may differ line by line. 
Recall that $\underline{t}^{\epsilon}$ and $\overline{t}^{\epsilon}$ appear in \eqref{ts}. 

\noindent(i) $v_{xt}^{\epsilon}(t, x)=v_{tx}^{\epsilon}(t, x)$ exist and are continuous in $(t,x)\in [0,T) \times (0,1)$. Furthermore,  
\begin{align}
&\left\vert v_{t}^{\epsilon}(t,x) \right\vert\leq C, \quad \left\vert v_{xt}^{\epsilon}(t,x) \right\vert  \leq C. \label{v_t_boundedness}
\end{align}

\noindent(ii) There exists $\epsilon_0>0$ such that for $\epsilon\in (0,\epsilon_0]$ and $t\in [0,\min\{\underline{t}^{\epsilon},\overline{t}^\epsilon\})$,
\begin{equation}
\begin{split}\label{Derivative_of_first_order_condition_in_time}
&\uy^{\epsilon}_t(t):=\tfrac{\partial \uy^{\epsilon}(t)}{\partial t}=\frac{\tfrac{ v_{x t}^{\epsilon}(t, \uy^{\epsilon}(t))}{1 - \gamma}- \frac{\epsilon  v_{t}^{\epsilon}(t, \uy^{\epsilon}(t))}{1 + \epsilon \uy^{\epsilon}(t)}}{- \frac{ v_{x x}^{\epsilon}(t, \uy^{\epsilon}(t))}{1 - \gamma} - \frac{\gamma \epsilon^{2}  v^{\epsilon}(t, \uy^{\epsilon}(t))}{(1 + \epsilon \uy^{\epsilon}(t))^{2}}}\quad 
\by^{\epsilon}_t(t):=\tfrac{\partial \by^{\epsilon}(t)}{\partial t}=\frac{\tfrac{ v_{x t}^{\epsilon}(t, \by^{\epsilon}(t))}{1 - \gamma}+ \frac{\epsilon  v_{t}^{\epsilon}(t, \by^{\epsilon}(t))}{1 - \epsilon \by^{\epsilon}(t)}}{- \frac{ v_{x x}^{\epsilon}(t, \by^{\epsilon}(t))}{1 - \gamma} - \frac{\gamma \epsilon^{2}  v^{\epsilon}(t, \by^{\epsilon}(t))}{(1 - \epsilon \by^{\epsilon}(t))^{2}}}.  \\
\end{split}
\end{equation}
Obviously, $\uy^{\epsilon}_t(t)=0$ for $t\in (\underline{t}^{\epsilon},T)$ and $\by^{\epsilon}_t(t)=0$ for $t\in (\overline{t}^{\epsilon},T)$.

\noindent(iii) Recall that $L^{\epsilon}$ appears in \eqref{Supremum_term_L}. For $\epsilon\in(0,\epsilon_0]$ with $\epsilon_0$ in (ii) and $t\in [0,T) \setminus\{\underline{t}^{\epsilon}, \overline{t}^{\epsilon}  \}$,
\begin{align}
&L_t^{\epsilon}(t,x)=
\begin{cases}
v_t^{\epsilon}(t, \uy^{\epsilon}(t)) \left( \tfrac{1 + \epsilon x}{1 + \epsilon \uy^{\epsilon}(t)} \right)^{1 - \gamma} &\textrm{if}\quad x\in (0,\uy^{\epsilon}(t)) \\
 v_{t}^{\epsilon} \left( t, x \right) &\textrm{if}\quad  x \in ( \uy^{\epsilon}(t), \by^{\epsilon}(t) )   \\
  v_t^{\epsilon}(t, \by^{\epsilon}(t)) \left( \tfrac{1 - \epsilon x}{1 - \epsilon \by^{\epsilon}(t)} \right)^{1 - \gamma} &\textrm{if}\quad  x \in (\by^{\epsilon}(t), 1)
\end{cases}  \label{Lt_expression}\\
&L_{xt}^{\epsilon}(t,x)=
\begin{cases}
\tfrac{\epsilon (1 - \gamma) v_t^{\epsilon}(t, \uy^{\epsilon}(t))}{1 + \epsilon x} \left( \tfrac{1 + \epsilon x}{1 + \epsilon \uy^{\epsilon}(t)} \right)^{1 - \gamma} &\textrm{if}\quad x\in (0,\uy^{\epsilon}(t)) \\
 v_{xt}^{\epsilon} \left( t, x \right) &\textrm{if}\quad  x \in ( \uy^{\epsilon}(t), \by^{\epsilon}(t) )   \\
  - \tfrac{\epsilon (1 - \gamma) v_t^{\epsilon}(t, \by^{\epsilon}(t))}{1 - \epsilon x} \left( \tfrac{1 - \epsilon x}{1 - \epsilon \by^{\epsilon}(t)} \right)^{1 - \gamma} &\textrm{if}\quad  x \in (\by^{\epsilon}(t), 1)
\end{cases}  \label{Lxt_expression}\\
&\left| L_t^{\epsilon}(t,x)\right|\leq C, \quad  \left| L_{xt}^{\epsilon}(t,x)\right| \leq C \quad \textrm{for}\quad x\in (0,1). \label{Lt_bound}
\end{align}

\noindent(iv) For $(t,x,\epsilon)\in [0,T)\times (0,1) \times (0,\epsilon_0]$ with $\epsilon_0$ in (ii), 
\begin{align}
&\left| v_{xt}^{\epsilon}(t,x)
- \lambda \int_t^T e^{-\lambda(s-t)} \EE\Big[ Z_{s}^{(t, x)}  \tfrac{\partial Y_{s}^{(t, x)}}{\partial x} v_{xt}^{\epsilon}(s, Y_s^{(t,x)}) \cdot 1_{ \left\{  \uy^{\epsilon}(s)<Y_{s}^{(t, x)}< \by^{\epsilon}(s) \right\} }  \Big] ds \right| \nonumber\\
& \leq C\big( e^{-\lambda(T-t)} + \epsilon^{\frac{2}{3}}\big).
\end{align}
\end{lemma}

\begin{proof}
(i) \lemref{Various_orders}, \eqref{merton_pde}, \eqref{Order_of_v_0_lambda_and_v_infinity} and \eqref{Dv_0_lambda_order} imply that for $(t,x)\in [0,T)\times (0,1)$,
\begin{equation}\begin{split}\label{various_vs}
\begin{dcases}
\left|v^{\epsilon}(t,x)\right| \leq C, \quad \left|v_x^{\epsilon}(t,x)\right| \leq C \epsilon^{\frac{2}{3}}, \\
 \left|x(1-x) v_{xx}^{\epsilon}(t,x)\right| \leq C \epsilon^{\frac{2}{3}}, \quad \left|x^2(1-x)^2v_{xxx}^{\epsilon}(t,x)\right| \leq C \epsilon^{\frac{1}{3}}. 
\end{dcases}
\end{split}
\end{equation}
Since $L^{\epsilon}(t,x)=v^{\epsilon}(t,x)$ for $x\in [\uy^\epsilon(t),\by^\epsilon(t)]$, the mean value theorem and the bounds of $L_x^\epsilon$ and $v_x^\epsilon$ in \eqref{Lx_bound} and \eqref{various_vs} imply
\begin{align}
\left| L^{\epsilon}(t,x)-v^{\epsilon}(t,x)\right|  \leq C \epsilon^{\frac{2}{3}} \quad \textrm{for}\quad (t,x)\in [0,T)\times (0,1). \label{L-v_estimate} 
\end{align}
By \eqref{SO_eq}, \eqref{L_difference_estimate} and \eqref{Order_of_v_0_lambda_and_v_infinity}, we observe that for $(t,x)\in [0,T)\times (0,1)$,
\begin{align}
\left| L^{\epsilon}(t,x)-v^{0}(t)\right| =\left| L^{\epsilon}(t,x)-L^{SO, \lambda}(t,x) + v^{SO, \lambda}(t,\hat y^{SO,\lambda} (t))-v^{0}(t)\right| \leq C \epsilon^{\frac{2}{3}}. \label{L-v_estimate2} 
\end{align}

Let $f^{\epsilon}: [0,T)\times (0,1)\to \R$ be defined as 
\begin{align}
  f^{\epsilon}(t, x) & := x (1 - x) \left( \mu - \gamma \sigma^{2} x \right) v_{x}^{\epsilon}(t, x) + \tfrac{\sigma^{2}x^{2} (1 - x)^{2}}{2}  v_{x x}^{\epsilon}(t, x) + \lambda ( L^{\epsilon}(t,x)-v^{\epsilon}(t,x)). \label{def_f}
\end{align}
Then, \eqref{various_vs}, \eqref{L-v_estimate} and \eqref{Lx_bound} imply that for $(t,x)\in [0,T)\times (0,1)$,
\begin{align}
\left|   f^{\epsilon}(t, x) \right| \leq C, \quad \left|   f_x^{\epsilon}(t, x) \right| \leq C. \label{f_estimate}
\end{align}
From \eqref{Main_PDE} and \eqref{def_f}, we obtain the bound of $v_t^\epsilon$ in \eqref{v_t_boundedness}: for $(t,x)\in [0,T)\times (0,1)$,
\begin{align}
 \left| v_{t}^{\epsilon}(t, x) \right| = \left| -Q(x)v^{\epsilon}(t, x)   - f^{\epsilon}(t, x)\right| \leq C,
\end{align}
where the inequality is due to \eqref{various_vs} and \eqref{f_estimate}.

Using \eqref{def_f}, we rewrite \eqref{Main_PDE} as
\begin{align}
  0 & = \tfrac{\partial}{\partial t} \left( e^{Q(x) t} v^{\epsilon}(t, x) \right) + e^{Q(x) t} f^{\epsilon}(t, x)\quad \textrm{with}\quad  v^{\epsilon}(T, x)=1 \nonumber\\
 &\Longrightarrow\quad  v^{\epsilon}(t, x) = e^{Q(x) (T - t)} + \int_{t}^{T} e^{Q(x) (s - t)} f^{\epsilon}(s, x) d s. \label{v_integral_form}
\end{align}
We differentiate \eqref{v_integral_form} with respect to $t$ and $x$ ($x$ and $t$, respectively) to obtain
\begin{align}
  v_{x t}^{\epsilon}(t, x) =v_{t x}^{\epsilon}(t, x)=- Q(x) v_{x}^{\epsilon}(t, x) - Q'(x) v^{\epsilon}(t, x) - f_{x}^{\epsilon}(t, x), \label{vxt_expression}
\end{align}
where the differentiations are justified by the bounds in \eqref{various_vs} and \eqref{f_estimate}. The above expression shows the continuity of $ v_{x t}^{\epsilon}$ and the boundedness of $v_{xt}^{\epsilon}$ in \eqref{v_t_boundedness} due to \eqref{various_vs} and \eqref{f_estimate}.

\medskip

(ii) We prove the result for $\uy^{\epsilon}_t(t)$ ($\by^{\epsilon}_t(t)$ case can be treated by the same way). Observe that
\begin{align*}
&\li_{\epsilon \downarrow 0}  \inf_{t \in [0, \underline{t}^{\epsilon})} \left( - \tfrac{\lambda v_{x x}^{\epsilon}(t, \uy^{\epsilon}(t))}{1 - \gamma} - \tfrac{\gamma \epsilon^{2} \lambda v^{\epsilon}(t, \uy^{\epsilon}(t))}{(1 + \epsilon \uy^{\epsilon}(t))^{2}}\right)/\epsilon^{\frac{1}{3}} \geq \li_{\epsilon \downarrow 0}  \inf_{t \in [0, T-\underline{C}\epsilon)}  - \tfrac{\lambda v_{x x}^{\epsilon}(t, \uy^{\epsilon}(t))}{(1 - \gamma)\epsilon^{\frac{1}{3}} } \geq C>0,
\end{align*}
where the first inequality is due to \lemref{NT_boundaries_hitting_times_bound} and \eqref{Range_of_v} and the second inequality is due to \eqref{Upper_bound_of_DDv}. The above observation implies that there exists $\epsilon_0>0$ such that
\begin{align}\label{need_positive}
- \tfrac{ v_{x x}^{\epsilon}(t, \uy^{\epsilon}(t))}{1 - \gamma} - \tfrac{\gamma \epsilon^{2}  v^{\epsilon}(t, \uy^{\epsilon}(t))}{(1 + \epsilon \uy^{\epsilon}(t))^{2}}\geq C \epsilon>0 \quad \textrm{for} \quad (t,\epsilon)\in [0,\underline{t}^{\epsilon})\times (0,\epsilon_0].
\end{align}
By \lemref{Boundaries_of_NT_region}, we have $ \tfrac{v_{x}^{\epsilon}(t, \uy^{\epsilon}(t))}{1 - \gamma} = \tfrac{\epsilon v^{\epsilon}(t, \uy^{\epsilon}(t))}{1 + \epsilon \uy^{\epsilon}(t)}$ for $ t\in [0,\underline{t}^{\epsilon})$.
For $(t,\epsilon)\in [0,\underline{t}^{\epsilon})\times (0,\epsilon_0]$, we apply the implicit function theorem to this equality (justified by \eqref{need_positive}) and conclude that $\uy^{\epsilon}_t(t)$ exists and is as in \eqref{Derivative_of_first_order_condition_in_time}. 

\medskip

(iii) We differentiate \eqref{Supremum_term_L} and \eqref{Lx_expression} with respect to $t$ and apply \eqref{Derivative_of_first_order_condition_in_time}, then we obtain \eqref{Lt_expression} and \eqref{Lxt_expression}.\footnote{Alternatively, considering the maximization problems in \eqref{uy_oy_argmax}, one may apply a suitable version of the envelope theorem to obtain the result.} 
These expressions and \eqref{v_t_boundedness} imply \eqref{Lt_bound}.

\medskip

(iv) As before, in this part of the proof, one can justify the interchanges of differentiations and integrations using suitable bounds such as \lemref{Boundedness_of_various_expectation} (i) and \eqref{Lt_bound}.

 Since $\lim_{t\uparrow T} \uy^{\epsilon}(t)=0$ and $\lim_{t\uparrow T} \by^{\epsilon}(t)=1$,
 \eqref{Supremum_term_L}, \eqref{Lx_expression}, $v^\epsilon(T,x)=1$ and $v_x^\epsilon(T,x)=0$ imply
\begin{align}
\lim_{t\uparrow T} L^\epsilon(t,x)=1, \quad \lim_{t\uparrow T} L_x^\epsilon(t,x)=0. \label{L(T)}
\end{align}
Since $(Z_{s}^{(t,x)}, Y_s^{(t,x)})$ and $(Z_{s-t}^{(0,x)}, Y_{s-t}^{(0,x)})$ have the same probability distribution, we have
\begin{align*}
 \lambda \int_{t}^{T} e^{- \lambda (s-t)} \EE \left[ Z_{s}^{(t, x)}  \tfrac{\partial Y_{s}^{(t, x)}}{\partial x} L_{x}^{\epsilon} (s, Y_{s}^{(t, x)}) \right] d s =  \lambda \int_{0}^{T-t} e^{- \lambda u} \EE \left[ Z_{u}^{(0, x)}  \tfrac{\partial Y_{u}^{(0, x)}}{\partial x} L_{x}^{\epsilon} (t+u, Y_{u}^{(0, x)}) \right] d u. 
\end{align*}
We differentiate above with respect to $t$, then \eqref{L(T)} and the continuity of $t\mapsto L_x^\epsilon(t,x)$ imply
\begin{align}
& \tfrac{\partial}{\partial t} \left( \lambda \int_{t}^{T} e^{- \lambda (s-t)} \EE \left[ Z_{s}^{(t, x)}  \tfrac{\partial Y_{s}^{(t, x)}}{\partial x} L_{x}^{\epsilon} (s, Y_{s}^{(t, x)}) \right] d s \right)  \nonumber \\
& = \lambda \int_{0}^{T-t} e^{- \lambda u} \EE \left[ Z_{u}^{(0, x)}  \tfrac{\partial Y_{u}^{(0, x)}}{\partial x}  L_{xt}^{\epsilon} (t+u, Y_{u}^{(0, x)}) \right] d u \nonumber\\
&=\lambda \int_{t}^{T} e^{- \lambda (s-t)} \EE \left[ Z_{s}^{(t, x)}  \tfrac{\partial Y_{s}^{(t, x)}}{\partial x}  L_{xt}^{\epsilon} (s, Y_{s}^{(t, x)}) \right] d s. \label{vxt2}
\end{align}
We differentiate \eqref{Dv_with_optimizer} with respect to $t$ and obtain 
\begin{align}
   v_{x t}^{\epsilon}(t, x)& = \lambda e^{- \lambda (T - t)} \EE \Big[ \tfrac{\partial Z_{T}^{(t, x)}}{\partial x} \Big] + e^{- \lambda (T - t)} \tfrac{\partial}{\partial t} \EE \Big[ \tfrac{\partial Z_{T}^{(t, x)}}{\partial x} \Big] \nonumber\\
  &\quad + \tfrac{\partial}{\partial t} \left( \lambda \int_{t}^{T} e^{- \lambda (s-t)} \EE \left[ Z_{s}^{(t, x)}  \tfrac{\partial Y_{s}^{(t, x)}}{\partial x} L_{x}^{\epsilon} (s, Y_{s}^{(t, x)}) + \tfrac{\partial Z_{s}^{(t, x)}}{\partial x} L^{\epsilon} ( s, Y_{s}^{(t, x)} )  \right] d s \right)\nonumber\\
&=e^{- \lambda (T - t)} \tfrac{\partial}{\partial t} \EE \Big[ \tfrac{\partial Z_{T}^{(t, x)}}{\partial x} \Big] + \lambda \int_{t}^{T} e^{- \lambda (s-t)} \EE \left[ Z_{s}^{(t, x)}  \tfrac{\partial Y_{s}^{(t, x)}}{\partial x} L_{xt}^{\epsilon} (s, Y_{s}^{(t, x)}) \right] ds \nonumber  \\
  &\quad + \tfrac{\partial}{\partial t} \left( \lambda \int_{t}^{T} e^{- \lambda (s-t)} \EE \left[  \tfrac{\partial Z_{s}^{(t, x)}}{\partial x} L^{\epsilon} ( s, Y_{s}^{(t, x)} )  \right] d s \right)+ \lambda e^{- \lambda (T - t)} \EE \Big[ \tfrac{\partial Z_{T}^{(t, x)}}{\partial x} \Big], \label{vxt1} 
\end{align}
where the second equality is due to \eqref{vxt2}. Using $\tfrac{\partial Z_{t}^{(t, x)}}{\partial x}=0$, we obtain
\begin{align}
&\left| \tfrac{\partial}{\partial t} \left( \lambda \int_{t}^{T} e^{- \lambda (s-t)} \EE \left[  \tfrac{\partial Z_{s}^{(t, x)}}{\partial x} \left( L^{\epsilon} ( s, Y_{s}^{(t, x)} ) - v^{0}(s)\right)  \right] d s \right) \right| \nonumber\\
&= \bigg| \lambda^2 \int_{t}^{T} e^{- \lambda (s-t)} \EE \left[  \tfrac{\partial Z_{s}^{(t, x)}}{\partial x} \left( L^{\epsilon} ( s, Y_{s}^{(t, x)} ) - v^{0}(s)\right)  \right] d s \nonumber\\
&\qquad + \lambda \int_{t}^{T} e^{- \lambda (s-t)} \tfrac{\partial}{\partial t} \EE \left[  \tfrac{\partial Z_{s}^{(t, x)}}{\partial x} \left( L^{\epsilon} ( s, Y_{s}^{(t, x)} ) - v^{0}(s)\right)  \right] d s \bigg| \nonumber\\
& \leq C \lambda^2 \int_{t}^{T} e^{- \lambda (s-t)} \big( \epsilon^{\frac{2}{3}}(s-t) + (s-t)^2 \big) d s 
+ C \lambda \int_{t}^{T} e^{- \lambda (s-t)} \big( \epsilon^{\frac{2}{3}}+s-t \big) d s \leq C \epsilon^{\frac{2}{3}}, \label{vxt3}
\end{align}
where the first inequality is due to \lemref{Boundedness_of_various_expectation} (iii) (with $F(y)=L^{\epsilon} ( s, y) - v^0(s)$), \eqref{L-v_estimate2} and \eqref{Lx_bound} and the second inequality is due to \lemref{Order_of_exponential_times} and $\lambda=c \, \epsilon^{-\frac{2}{3}}$. Similarly,
\begin{align}
&\left| \tfrac{\partial}{\partial t} \left( \lambda \int_{t}^{T} e^{- \lambda (s-t)} \EE \left[  \tfrac{\partial Z_{s}^{(t, x)}}{\partial x}  v^{0}(s)  \right] d s \right) + \lambda e^{- \lambda (T - t)} \EE \Big[ \tfrac{\partial Z_{T}^{(t, x)}}{\partial x} \Big]  \right|  \nonumber\\
&=\left| \tfrac{\partial}{\partial t} \left( \lambda \int_{0}^{T-t} e^{- \lambda u} \EE \left[  \tfrac{\partial Z_{u}^{(0, x)}}{\partial x}  v^{0}(u+t)  \right] d u \right) + \lambda e^{- \lambda (T - t)} \EE \Big[ \tfrac{\partial Z_{T-t}^{(0, x)}}{\partial x} \Big]  \right|  \nonumber\\
&= \left| - \lambda \int_{0}^{T-t} e^{- \lambda u}Q(y_M) \, v^{0}(u+t) \, \EE \left[  \tfrac{\partial Z_{u}^{(0, x)}}{\partial x}   \right] d u \right| \nonumber\\
&\leq C\lambda \int_0^{T-t} e^{-\lambda u} u\,  du \leq C \epsilon^{\frac{2}{3}}, \label{vxt4}
\end{align}
where the second equality is due to \eqref{merton_pde}, the first inequality is due to \lemref{Boundedness_of_various_expectation} (i) and the last inequality is due to $\lambda=c \, \epsilon^{-\frac{2}{3}}$. Combining \eqref{vxt1}, \eqref{vxt3}, and \eqref{vxt4}, we obtain
\begin{align}
&\left| v_{x t}^{\epsilon}(t, x) - \lambda \int_{t}^{T} e^{- \lambda (s-t)} \EE \left[ Z_{s}^{(t, x)}  \tfrac{\partial Y_{s}^{(t, x)}}{\partial x} L_{xt}^{\epsilon} (s, Y_{s}^{(t, x)}) \right] ds \right| \nonumber\\
&=\left| 
e^{- \lambda (T - t)} \tfrac{\partial}{\partial t} \EE \Big[ \tfrac{\partial Z_{T}^{(t, x)}}{\partial x} \Big] +  \tfrac{\partial}{\partial t} \left( \lambda \int_{t}^{T} e^{- \lambda (s-t)} \EE \Big[  \tfrac{\partial Z_{s}^{(t, x)}}{\partial x} L^{\epsilon} ( s, Y_{s}^{(t, x)} )  \Big] d s \right)+ \lambda e^{- \lambda (T - t)} \EE \Big[ \tfrac{\partial Z_{T}^{(t, x)}}{\partial x} \Big]
\right| \nonumber\\
&\leq C\big( e^{-\lambda(T-t)} + \epsilon^{\frac{2}{3}}\big),\label{vxt5}
\end{align}
where the boundedness of $\tfrac{\partial}{\partial t} \EE [ \tfrac{\partial Z_{T}^{(t, x)}}{\partial x} ]$ is due to \lemref{Boundedness_of_various_expectation} (iii) (with $F(y)=1$).
The expression of $L_{xt}^\epsilon$ in \eqref{Lxt_expression} and the bounds in \lemref{Boundedness_of_various_expectation} (i)  and \eqref{v_t_boundedness} imply
\begin{align*}
\left| \lambda \int_{t}^{T} e^{- \lambda (s-t)} \EE \Big[ Z_{s}^{(t, x)}  \tfrac{\partial Y_{s}^{(t, x)}}{\partial x}  \Big( L_{xt}^{\epsilon} (s, Y_{s}^{(t, x)}) - v_{xt}^{\epsilon} (s, Y_{s}^{(t, x)}) \cdot 1_{ \left\{ \uy^{\epsilon}(s)<Y_s^{(t, x)}<\by^{\epsilon}(s) \right\} } \Big) \Big] d s \right| \leq C \epsilon.
\end{align*}
Finally, we conclude the desired result by the above inequality and \eqref{vxt5}.
\end{proof}

\medskip

\section{Proof of \lemref{Merton_fraction_inside_NT}}\label{lemma 5.6}
Throughout this proof, $C, C_1, C_2>0$ are generic constants independent of $(t,s, x,\epsilon)\in [0,T)\times [t,T)\times(0,1) \times (0,1)$ (also independent of $\lambda$ due to relation $\lambda=c \, \epsilon^{-\frac{2}{3}}$) that may differ line by line.

(i) By \lemref{Boundaries_of_NT_region} and \eqref{SO_eq}, we have $v_{x}^{SO, \lambda}(t, \hat{y}^{SO, \lambda}(t))=0$. By the mean value theorem,
\begin{align}
 \left\vert v_{x}^{SO, \lambda}(t, \uy^{\epsilon}(t))\right| &=  \left\vert v_{x}^{SO, \lambda}(t, \uy^{\epsilon}(t)) - v_{x}^{SO, \lambda}(t, \hat{y}^{SO, \lambda}(t)) \right\vert 
\geq \inf_{y\in (0,1)} \left| v_{x x}^{SO, \lambda}(t, y) \right| \cdot \left\vert \uy^{\epsilon}(t) - \hat{y}^{SO, \lambda}(t) \right\vert  \nonumber\\
 & \geq C  \tfrac{\min \{ 1, \lambda (T - t) \}}{\lambda} \cdot  \left\vert \uy^{\epsilon}(t) - \hat{y}^{SO, \lambda}(t) \right\vert , \label{for_y_estimate1}
\end{align}
where the last inequality is due to \eqref{DDv_0_lambda_order}. By \lemref{Boundaries_of_NT_region} (iii) and \eqref{Dv_ep_Dv_0_difference}, we obtain
\begin{align}
 \left\vert v_{x}^{SO, \lambda}(t, \uy^{\epsilon}(t))\right| \leq  \left\vert v_{x}^{\epsilon}(t, \uy^{\epsilon}(t)) - v_{x}^{SO, \lambda}(t, \uy^{\epsilon}(t)) \right\vert + \left\vert v_{x}^{\epsilon}(t, \uy^{\epsilon}(t)) \right| \leq C \epsilon. \label{for_y_estimate2}
\end{align}
We combine \eqref{Order_of_hat_y0_and_mean_variance_ratio}, \eqref{for_y_estimate1} and \eqref{for_y_estimate2} to obtain
\begin{align*}
\left\vert \uy^{\epsilon}(t) -y_M \right\vert \leq \tfrac{C \epsilon^{\frac{1}{3}} }{\min \{ 1, \lambda (T - t) \}}. 
\end{align*}
By the same way, we obtain the other inequalities in \eqref{pre_y1}.

\medskip

(ii) Let $t\in [0,T)$ be fixed. Since $\min\{1,\lambda(T-t)\}=1$ for small enough $\epsilon$, the inequalities in \eqref{pre_y1} and $0<y_M<1$ imply that 
\begin{align}
 &0<\li_{\epsilon \downarrow 0}\uy^{\epsilon}(t)\leq   \ls_{\epsilon \downarrow 0}\by^{\epsilon}(t)<1.  \label{y_not_01}
\end{align}
By the mean value theorem, there exists $x^{\epsilon}\in [\uy^{\epsilon}(t), \by^{\epsilon}(t)]$ such that for small enough $\epsilon$,
\begin{align}
v_{x x}^{\epsilon}(t, x^{\epsilon}) (\by^{\epsilon}(t) - \uy^{\epsilon}(t)) &= v_{x}^{\epsilon}(t, \by^{\epsilon}(t)) - v_{x}^{\epsilon}(t, \uy^{\epsilon}(t)) = - \tfrac{\epsilon (1 - \gamma) v^{\epsilon}(t, \by^{\epsilon}(t))}{1 - \epsilon \by^{\epsilon}(t)} - \tfrac{\epsilon (1 - \gamma) v^{\epsilon}(t, \uy^{\epsilon}(t))}{1 + \epsilon \uy^{\epsilon}(t)},  \label{c4}
\end{align}
where the second equality is due to \eqref{y_not_01} and \lemref{Boundaries_of_NT_region}. We observe that \eqref{vxx_twosides} and \eqref{y_not_01} imply $\li_{\epsilon \downarrow 0} \tfrac{-(1-\gamma)}{ v_{x x}^{\epsilon}(t, x^{\epsilon})} \epsilon^{\frac{2}{3}}>0$. Therefore, \eqref{c4} and  \eqref{Range_of_v} produce
\begin{align*}
\li_{\epsilon \downarrow 0} \tfrac{ \by^{\epsilon}(t) - \uy^{\epsilon}(t)}{\epsilon^{\frac{1}{3}}}= \li_{\epsilon \downarrow 0} \tfrac{-(1-\gamma)}{ v_{x x}^{\epsilon}(t, x^{\epsilon})} \epsilon^{\frac{2}{3}} \left( \tfrac{v^{\epsilon}(t, \by^{\epsilon}(t))}{1 - \epsilon \by^{\epsilon}(t)} + \tfrac{v^{\epsilon}(t, \uy^{\epsilon}(t))}{1 + \epsilon \uy^{\epsilon}(t)} \right)>0.
\end{align*}

(iii) Considering \lemref{Boundaries_of_NT_region} (i), it is enough to show that 
\begin{align*}
\ls_{\epsilon \downarrow 0}\sup_{t\in [0,T)}  \left( \tfrac{v_{x}^{\epsilon}(t, y_{M})}{\epsilon (1 - \gamma)} - \tfrac{v^{\epsilon}(t, y_{M})}{1 + \epsilon y_{M}}  \right) <0 \quad \textrm{and} \quad \li_{\epsilon \downarrow 0}\inf_{t\in [0,T)}  \left( \tfrac{v_{x}^{\epsilon}(t, y_{M})}{\epsilon (1 - \gamma)} + \tfrac{v^{\epsilon}(t, y_{M})}{1 - \epsilon y_{M}}  \right) >0.
\end{align*} 
We prove the first inequality above. The other inequality can be proved by the same way.

By \eqref{pre_y1}, we have 
\begin{align}
\left\vert \by^{\epsilon}(s) -y_M \right\vert \leq C \epsilon^{\frac{1}{3}} \quad \textrm{for} \quad 0\leq s \leq T-\tfrac{1}{\lambda}. \label{pre_y2}
\end{align}
The expression of $Y_{s}^{(t, x)}$ in \eqref{AYZ} implies that
\begin{align}
\PP\left( Y_{s}^{(t, y_{M})} \geq \by^{\epsilon}(s) \right)
&=\PP\left(  B_s - B_t \geq (\tfrac{\sigma}{2} - \tfrac{\mu}{\sigma})(s-t) + \tfrac{1}{\sigma} \ln \left(\tfrac{\by^{\epsilon}(s)(1-y_M)}{y_M(1-\by^{\epsilon}(s))} \right) \right)\nonumber\\
&\geq \PP\left( B_1 \geq C\big(1+ \tfrac{1}{\sqrt{s-t}}\epsilon^{\frac{1}{3}} \big) \right) \quad \textrm{for} \quad  t< s \leq T-\tfrac{1}{\lambda}, \label{PY>y1}
\end{align}
where the inequality is due to \eqref{pre_y2} and the mean value theorem. By the same way as we prove \eqref{Expectation_Time_order}, we check that for $t\leq s$,
\begin{align}
\EE\Big[ \Big|  Z_{s}^{(t, x)} \tfrac{\partial Y_{s}^{(t, x)}}{\partial x} - 1 \Big| \Big] \bigg|_{x=y_M} \leq  C\sqrt{s-t}, \label{PY>y2}\\
\EE\Big[ \Big|  Z_{s}^{(t, x)} \Big(\tfrac{\partial Y_{s}^{(t, x)}}{\partial x} - 1 \Big) \Big| \Big] \bigg|_{x=y_M}    \leq C\sqrt{s-t}. \label{PY>y4}
\end{align}
By \eqref{Range_of_v} and \lemref{Boundaries_of_NT_region} (i), for $x>\uy^{\epsilon}(t)$,
\begin{align}
 \tfrac{v_x^{\epsilon}(t,x)}{\epsilon(1-\gamma)} - \tfrac{v^{\epsilon}(t, x)}{1 + \epsilon y_{M}} \leq \tfrac{v^{\epsilon} \left( t, x \right)}{1 + \epsilon x} - \tfrac{v^{\epsilon}(t, x)}{1 + \epsilon y_{M}} \leq C \epsilon. \label{c9}
\end{align}
From the expression of $L$ in \eqref{Supremum_term_L} and $L_x$ in \eqref{Lx_expression}, we obtain
\begin{align}
 \tfrac{L_{x}^{\epsilon} ( t, x )}{\epsilon (1 - \gamma)} - \tfrac{L^{\epsilon} ( t, x )}{1 + \epsilon y_{M}}
&=  \tfrac{\epsilon (y_M-x) v^{\epsilon}(t, \uy^{\epsilon}(t))}{(1 + \epsilon x)(1+\epsilon y_M)}  \left( \tfrac{1 + \epsilon x}{1 + \epsilon \uy^{\epsilon}(t)} \right)^{1 - \gamma} 1_{\left\{ x \leq \uy^{\epsilon}(t) \right\}} + \left(\tfrac{v_x^{\epsilon}(t,x)}{\epsilon(1-\gamma)} - \tfrac{v^{\epsilon}(t, x)}{1 + \epsilon y_{M}} \right) 1_{\left\{ x \in ( \uy^{\epsilon}(t), \by^{\epsilon}(t) ) \right\}} \nonumber\\
  &\quad- \tfrac{(2+\epsilon(y_M -x))v^{\epsilon}(t, \by^{\epsilon}(t))}{(1 - \epsilon x)(1+\epsilon y_M)} \left( \tfrac{1 - \epsilon x}{1 - \epsilon \by^{\epsilon}(t)} \right)^{1 - \gamma} 1_{\left\{ x \geq\by^{\epsilon}(t) \right\}} \nonumber\\
 &\leq C_1 \epsilon -  C_2 1_{\left\{ x \geq\by^{\epsilon}(t) \right\}},\label{PY>y3}
  \end{align}
where the inequality is due to \eqref{Range_of_v} and \eqref{c9}. Since $Z_{s}^{(t, x)} \tfrac{\partial Y_{s}^{(t, x)}}{\partial x}>0$, \eqref{PY>y2} and \eqref{PY>y3} imply
\begin{align}
& \EE \Big[ Z_{s}^{(t, x)} \tfrac{\partial Y_{s}^{(t, x)}}{\partial x}  \Big( \tfrac{L_{x}^{\epsilon} ( s, Y_{s}^{(t, x)} )}{\epsilon (1 - \gamma)} - \tfrac{L^{\epsilon} ( s, Y_{s}^{(t, x)} )}{1 + \epsilon x} \Big) \Big] \bigg|_{x=y_M} \nonumber\\
 &\leq \EE \Big[ \Big(Z_{s}^{(t, x)} \tfrac{\partial Y_{s}^{(t, x)}}{\partial x}-1\Big)  \Big( C_1 \epsilon -  C_2 1_{\{  Y_{s}^{(t, x)} \geq\by^{\epsilon}(s) \}} \Big) \Big] \bigg|_{x=y_M} +\EE \Big[ C_1 \epsilon -  C_2 1_{\{  Y_{s}^{(t, y_{M})} \geq\by^{\epsilon}(s) \}}  \Big] \nonumber\\
& \leq C_1 \big( \sqrt{s-t} + \epsilon\big)  -  C_2 \, \PP\left( Y_{s}^{(t, y_{M})} \geq \by^{\epsilon}(s) \right) \quad \textrm{for} \quad t\leq s \leq T. \label{c11}
\end{align}
By \eqref{c11} and \lemref{Order_of_exponential_times}, we have
\begin{align}
&\lambda \int_{t}^{T} e^{- \lambda (s - t)} \EE \Big[ Z_{s}^{(t, x)} \tfrac{\partial Y_{s}^{(t, x)}}{\partial x}  \Big( \tfrac{L_{x}^{\epsilon} ( s, Y_{s}^{(t, x)} )}{\epsilon (1 - \gamma)} - \tfrac{L^{\epsilon} ( s, Y_{s}^{(t, x)} )}{1 + \epsilon x} \Big) \Big] d s \bigg|_{x=y_M}     \nonumber\\
&\leq C_1 \epsilon^{\frac{1}{3}} - C_2 \lambda \int_{t}^{T} e^{- \lambda (s - t)} \PP\left( Y_{s}^{(t, y_{M})} \geq \by^{\epsilon}(s) \right)  ds \nonumber\\
&\leq C_1 \epsilon^{\frac{1}{3}} - C_2 \int_0^{(\lambda(T-t)-1)^+} e^{-u} \int_{C(1+\frac{1}{\sqrt{u}})}^\infty  \tfrac{e^{- \frac{z^{2}}{2}}}{\sqrt{2 \pi}} d z \, du, \label{vx-v1}
\end{align}
where the last inequality is due to \eqref{PY>y1} and the substitution $u=\lambda(s-t)$. 

\lemref{Boundedness_of_various_expectation} (iii) (with $F(y)=1$ and $F(y)=L^\epsilon(s,y)$, respectively) and \eqref{Lx_bound} produce 
\begin{align}
&\Big\vert \EE \Big[ \tfrac{\partial Z_{s}^{(t, x)}}{\partial x} \Big]  \Big\vert \bigg|_{x=y_M} \leq C (s-t)^{2},\label{PY>y5}\\
&   \Big\vert \EE \Big[ \tfrac{\partial Z_{s}^{(t, x)}}{\partial x}  L^{\epsilon} ( s, Y_{s}^{(t, x)} ) \Big] \Big| \bigg|_{x=y_M}
 \leq    C \left( \epsilon (s-t) + (s-t)^2 \right). \label{PY>y6}
\end{align}
By the same way as we prove \eqref{Expectation_Time_order}, we check that $\EE[Z_T^{(t,x)}]\geq 1-C(T-t)$. Then \eqref{PY>y5} produces 
\begin{align}
e^{- \lambda (T - t)} \EE \Big[
  \tfrac{1}{\epsilon (1 - \gamma)} \tfrac{\partial Z_{T}^{(t, x)}}{\partial x} \ - \tfrac{Z_{T}^{(t, x)}}{1 + \epsilon x} \Big] \bigg|_{x=y_M}&\leq C_1  e^{- \lambda (T - t)} \cdot \tfrac{(T - t)^{2}}{\epsilon} -  e^{- \lambda (T - t)} \left( 1- C_2(T-t)\right) \nonumber\\
  & \leq C\epsilon^{\frac{1}{3}} - e^{- \lambda (T - t)}, \label{vx-v2}
\end{align}
where the last inequality is due to $\sup_{t\geq 0} t^n e^{-t}<\infty$ for $n=1,2$.
By \eqref{PY>y6}, \eqref{PY>y4} and \lemref{Order_of_exponential_times}, 
\begin{equation}
\begin{split} \label{vx-v3}
& \lambda \int_{t}^{T} e^{- \lambda (s - t)} \EE \Big[ \tfrac{\partial Z_{s}^{(t, x)}}{\partial x} \tfrac{L^{\epsilon} ( s, Y_{s}^{(t, x)} )}{\epsilon (1 - \gamma)} \Big] d s \bigg|_{x=y_M} \leq C \lambda \int_{t}^{T} e^{- \lambda (s - t)} \left( (s-t) + \tfrac{1}{\epsilon}(s-t)^2 \right)  \leq C\epsilon^{\frac{1}{3}}, \\
& \lambda \int_{t}^{T} e^{- \lambda (s - t)} \EE \Big[  Z_{s}^{(t, x)} \left(\tfrac{\partial Y_{s}^{(t, x)}}{\partial x} - 1 \right) \tfrac{L^{\epsilon} ( s, Y_{s}^{(t, x)} )}{1 + \epsilon x} \Big] d s \bigg|_{x=y_M}\leq C \lambda \int_{t}^{T} e^{- \lambda (s - t)} \sqrt{s-t} \leq C\epsilon^{\frac{1}{3}}. 
\end{split}
\end{equation}
The representation of $v$ in \eqref{v_with_optimizer} and $v_x$ in \eqref{Dv_with_optimizer} produce 
\begin{align*}
   \tfrac{v_{x}^{\epsilon}(t, x)}{\epsilon (1 - \gamma)} - \tfrac{v^{\epsilon}(t, x)}{1 + \epsilon x} 
   &= e^{- \lambda (T - t)} \EE \Big[
  \tfrac{1}{\epsilon (1 - \gamma)}\tfrac{\partial Z_{T}^{(t, x)}}{\partial x} \ - \tfrac{Z_{T}^{(t, x)}}{1 + \epsilon x} \Big] + \lambda \int_{t}^{T} e^{- \lambda (s - t)} \EE \Big[ \tfrac{\partial Z_{s}^{(t, x)}}{\partial x} \tfrac{L^{\epsilon} ( s, Y_{s}^{(t, x)} )}{\epsilon (1 - \gamma)} \Big] d s \nonumber \\
  & \quad + \lambda \int_{t}^{T} e^{- \lambda (s - t)} \EE \Big[ Z_{s}^{(t, x)} \tfrac{\partial Y_{s}^{(t, x)}}{\partial x}  \Big( \tfrac{L_{x}^{\epsilon} ( s, Y_{s}^{(t, x)} )}{\epsilon (1 - \gamma)} - \tfrac{L^{\epsilon} ( s, Y_{s}^{(t, x)} )}{1 + \epsilon x} \Big) \Big] d s \nonumber \\
  & \quad + \lambda \int_{t}^{T} e^{- \lambda (s - t)} \EE \left[  Z_{s}^{(t, x)} \left(\tfrac{\partial Y_{s}^{(t, x)}}{\partial x} - 1 \right) \tfrac{L^{\epsilon} ( s, Y_{s}^{(t, x)} )}{1 + \epsilon x} \right] d s. 
\end{align*}
We substitute $x=y_M$ above and apply \eqref{vx-v1}, \eqref{vx-v2} and \eqref{vx-v3} to obtain
\begin{align*}
\ls_{\epsilon \downarrow 0}\sup_{t\in [0,T)}  \left( \tfrac{v_{x}^{\epsilon}(t, y_{M})}{\epsilon (1 - \gamma)} - \tfrac{v^{\epsilon}(t, y_{M})}{1 + \epsilon y_{M}}  \right) &\leq C \ls_{\epsilon \downarrow 0}\sup_{t\in [0,T)}  \Big(- e^{- \lambda (T - t)} - \int_0^{(\lambda(T-t)-1)^+} e^{-u} \int_{C_2(1+\frac{1}{\sqrt{u}})}^\infty  \tfrac{e^{- \frac{z^{2}}{2}}}{\sqrt{2 \pi}} d z \, du \Big)     \nonumber\\
& \leq C \sup_{a\in [0,\infty)}  \Big(- e^{- a} - \int_0^{(a-1)^+} e^{-u} \int_{C_2(1+\frac{1}{\sqrt{u}})}^\infty  \tfrac{e^{- \frac{z^{2}}{2}}}{\sqrt{2 \pi}} d z \, du \Big) <0.
\end{align*}

\section{Proof of \lemref{really_used}}\label{lemma 5.7}
Throughout this proof, $C>0$ is a generic constant independent of $(t,s, x,\epsilon)\in [0,T)\times [t,T)\times(0,1) \times (0,1)$ (also independent of $\lambda$ due to relation $\lambda=c \, \epsilon^{-\frac{2}{3}}$) that may differ line by line.

\medskip

(i) We obtain \eqref{vx_epsilon} by \lemref{Boundaries_of_NT_region} (iii) and \eqref{Range_of_v}. We combine \eqref{Order_of_v_0_lambda_and_v_infinity}, \eqref{Dv_0_lambda_order} and \eqref{v_ep_v_0_difference_more_strong} to obtain \eqref{v_epsilon}. 
To check \eqref{vt_ep_vt_0_difference_on_NT}, we rewrite \eqref{Main_PDE} using \eqref{merton_pde} as
\begin{align}
  v_{t}^{\epsilon}(t, x) - v_{t}^{0}(t) = Q(y_{M}) ( v^{0}(t) - v^{\epsilon}(t, x) )+ \tfrac{\gamma (1 - \gamma) \sigma^{2}}{2}  (x - y_{M})^{2} v^{\epsilon}(t, x) - f^{\epsilon}(t, x), \label{vt_form}
\end{align}
where $f^{\epsilon}$ is defined in \eqref{def_f}.
Since $L^{\epsilon}(t,x)=v^{\epsilon}(t,x)$ for $x\in [ \uy^{\epsilon}(t), \by^{\epsilon}(t)]$, \eqref{vx_epsilon} and \eqref{vxx_twosides} imply that $\left|  f^{\epsilon}(t, x) \right| \leq C  \epsilon^{\frac{2}{3}}$ for $x\in [ \uy^{\epsilon}(t), \by^{\epsilon}(t)]$. We apply this inequality, \eqref{v_epsilon} and \eqref{Range_of_v} to \eqref{vt_form} and conclude \eqref{vt_ep_vt_0_difference_on_NT}.

\medskip

(ii) For $x_1,x_2 \in  [ \uy^{\epsilon}(t), \by^{\epsilon}(t)]$, the mean value theorem and  \eqref{vx_epsilon} imply $\left| v^{\epsilon}(t, x_1)-v^{\epsilon}(t, x_2)\right| \leq C \epsilon$.
Then we conclude \eqref{Uniform_limit_of_v_ep_v_0_difference}. 

Since $v_{tx}^\epsilon$ is continuous by \lemref{Various_orders2} (i), the mean value theorem and \eqref{pre_y1} imply 
\begin{align*}
&\left|v_t^{\epsilon}(t, x_1) - v_t^{\epsilon}(t, x_2)\right| 
\leq  \sup_{x\in [ \uy^{\epsilon}(t), \by^{\epsilon}(t)]} \left| v_{tx}^{\epsilon}(t, x)\right| \cdot \tfrac{C \epsilon^{\frac{1}{3}} }{\min \{ 1, \lambda (T - t) \}} \quad \textrm{for} \quad x_1,x_2 \in  [ \uy^{\epsilon}(t), \by^{\epsilon}(t)].
\end{align*}
The above inequality and the following lemma produce \eqref{Uniform_limit_of_vt_ep_vt_0_difference}.
\begin{lemma}
Let \assref{ass} hold. Let $\epsilon_0>0$ be as in \lemref{Various_orders2} (ii).
For $\alpha\in (0,1)$, there exists a positive constant $C$ independent of $(t,x,\epsilon)\in [0,T)\times (0,1) \times (0,\epsilon_0]$ such that
\begin{align}
\left|v_{x t}^{\epsilon}(t, x) \right| \leq C\big( \epsilon^{\frac{2}{3}} + e^{- \alpha \lambda (T - t)} \big). \label{vxt_bound}
\end{align}
\end{lemma}
\begin{proof}
By \eqref{def_f}, \eqref{various_vs} and \eqref{Lx_bound}, we have $\left|  f_x^{\epsilon}(t, x) \right| \leq C  \epsilon^{\frac{1}{3}}$. We apply this inequality, \eqref{pre_y1} and \eqref{various_vs} to the expression in \eqref{vxt_expression} to obtain
\begin{align}
\left|v_{xt}^{\epsilon}(t,x) \cdot 1_{ \{  \uy^{\epsilon}(t)<x< \by^{\epsilon}(t) \} } \right| \leq C  \epsilon^{\frac{1}{3}} \quad \textrm{for} \quad (t,x,\epsilon)\in [0,T-\tfrac{1}{\lambda}]\times (0,1)\times (0,1).  \label{vxt1/3}
\end{align}
By \eqref{PY>y2}, \eqref{vxt1/3} and \lemref{Order_of_exponential_times}, we have
\begin{align}
&\bigg| \lambda \int_t^{(T-\frac{1}{\lambda})\vee t} e^{-\lambda(s-t)} \EE\Big[ \left(Z_{s}^{(t, x)}  \tfrac{\partial Y_{s}^{(t, x)}}{\partial x}-1\right) v_{xt}^{\epsilon}(s, Y_s^{(t,x)}) \cdot 1_{ \{  \uy^{\epsilon}(s)<Y_{s}^{(t, x)}< \by^{\epsilon}(s) \} }  \Big] ds \bigg| \nonumber \\
&\leq C\lambda \int_t^{(T-\frac{1}{\lambda})\vee t} e^{-\lambda(s-t)} \sqrt{s-t}\,\, \epsilon^{\frac{1}{3}}  ds   
\leq C \epsilon^{\frac{2}{3}}. \label{vxt2/3}
\end{align}
Since $\lambda \int_{T-\frac{1}{\lambda}}^T e^{-\lambda(s-t)} ds \leq C e^{- \lambda (T - t)}$, \lemref{Boundedness_of_various_expectation} (i) and \eqref{v_t_boundedness} produce
\begin{align}
 \lambda \int_{T-\frac{1}{\lambda}}^T e^{-\lambda(s-t)} \EE\Big[ \Big| Z_{s}^{(t, x)}  \tfrac{\partial Y_{s}^{(t, x)}}{\partial x} v_{xt}^{\epsilon}(s, Y_s^{(t,x)}) \cdot 1_{ \{  \uy^{\epsilon}(s)<Y_{s}^{(t, x)}< \by^{\epsilon}(s) \} } \Big| \Big] ds \leq C e^{- \lambda (T - t)}. \label{vxt_tail}
\end{align}
We combine \lemref{Various_orders2} (iv), \eqref{vxt2/3} and \eqref{vxt_tail} to conclude that for $(t,x,\epsilon)\in [0,T)\times (0,1) \times (0,\epsilon_0]$,
\begin{align*}
\left\vert v_{x t}^{\epsilon}(t, x) \right\vert &\leq C \big( e^{- \lambda (T - t)} + \epsilon^{\frac{2}{3}} \big)
+  \lambda \int_t^{(T-\frac{1}{\lambda})\vee t} e^{-\lambda(s-t)} \EE\Big[ \left| v_{xt}^{\epsilon}(s, Y_s^{(t,x)})\right| \cdot 1_{ \left\{  \uy^{\epsilon}(s)<Y_{s}^{(t, x)}< \by^{\epsilon}(s) \right\} }  \Big] ds.
\end{align*}
For $\alpha \in (0,1)$, let $k^{\epsilon}(\alpha):=\sup_{(t,x)\in [0,T)\times(0,1)} \tfrac{ | v_{xt}^{\epsilon}(t,x)|}{e^{- \alpha \lambda (T - t)} + \epsilon^{\frac{2}{3}}}$. Then, the above inequality implies
\begin{align}
 \tfrac{ | v_{xt}^{\epsilon}(t,x)|}{e^{- \alpha \lambda (T - t)} + \epsilon^{\frac{2}{3}}} &\leq C +   \lambda \int_t^{(T-\frac{1}{\lambda})\vee t} e^{-\lambda(s-t)}\,  \tfrac{e^{- \alpha \lambda (T - s)} + \epsilon^{\frac{2}{3}}}{e^{- \alpha \lambda (T - t)} + \epsilon^{\frac{2}{3}}} \, k^{\epsilon}(\alpha) \, \EE\Big[ 1_{ \{  \uy^{\epsilon}(s)<Y_{s}^{(t, x)}< \by^{\epsilon}(s) \} }  \Big] ds \nonumber\\
&\leq C + k^{\epsilon}(\alpha)\,  \lambda \int_t^{(T-\frac{1}{\lambda})\vee t} e^{-(1-\alpha)\lambda(s-t)}\EE\Big[ 1_{ \{  \uy^{\epsilon}(s)<Y_{s}^{(t, x)}< \by^{\epsilon}(s) \} }  \Big] ds. \label{vxt_k}
\end{align}
Observe that for $t<s< (T-\frac{1}{\lambda})\vee t$, the definition of $Y_s^{(t,x)}$ in \eqref{AYZ} produces
\begin{align}
&\PP \left(  \uy^{\epsilon}(s)<Y_{s}^{(t, x)}< \by^{\epsilon}(s)  \right )\nonumber\\
&=\PP \left( \ln\left(\tfrac{(1-x)\uy^{\epsilon}(s)}{x(1-\uy^{\epsilon}(s))} \right) - (\mu-\tfrac{\sigma^2}{2})(s-t)<\sigma (B_s - B_t)< \ln\left(\tfrac{(1-x)\by^{\epsilon}(s)}{x(1-\by^{\epsilon}(s))} \right) - (\mu-\tfrac{\sigma^2}{2})(s-t) \right ) \nonumber\\
&\leq \PP \left( -C \epsilon^{\frac{1}{3}} < B_s - B_t <C \epsilon^{\frac{1}{3}} \right)=\PP \Big( -\tfrac{C\epsilon^{\frac{1}{3}}}{\sqrt{s-t}} < B_1 < \tfrac{C\epsilon^{\frac{1}{3}}}{\sqrt{s-t}} \Big) , \label{PB}
\end{align}
where the inequality is due to $\PP\big(a<B_s - B_t <b \big) \leq \PP \big(-\tfrac{b-a}{2}<B_s - B_t <\tfrac{b-a}{2}\big)$ for $a<b$, \eqref{pre_y1} and the mean value theorem. Then, \eqref{PB} with the substitution $u=\lambda(s-t)$ produces
\begin{align*}
&\lambda \int_t^{(T-\frac{1}{\lambda})\vee t} e^{-(1-\alpha)\lambda(s-t)}\, \EE\Big[ 1_{ \{  \uy^{\epsilon}(s)<Y_{s}^{(t, x)}< \by^{\epsilon}(s) \} }  \Big] ds
\leq C_\alpha:= \int_0^{\infty} e^{-(1-\alpha)u} \,\,\PP \left( -\tfrac{C}{\sqrt{u}} < B_1 < \tfrac{C}{\sqrt{u}} \right) du.
\end{align*}
Since the constant $C_\alpha$ above does not depend on $t,x,\epsilon$ and $C_\alpha<1$, \eqref{vxt_k} implies $k^{\epsilon}(\alpha)\leq \tfrac{C}{1-C_\alpha}$.
\end{proof}

\section{Proof of \lemref{v_xx_conv_lem}}\label{lemma 5.8}

Throughout this appendix, $C>0$ is a generic constant independent of $(t,s, x,\epsilon)\in [0,T)\times [t,T)\times(0,1) \times (0,1)$ (also independent of $\lambda$ due to relation $\lambda=c \, \epsilon^{-\frac{2}{3}}$) that may differ line by line.
First, we prove the following lemma.

\begin{lemma}\label{vxxt_lem}
Let \assref{ass} hold. Let $\epsilon_0>0$ be as in \lemref{Various_orders2} (ii) and $\alpha\in (0,1)$. Then, 
there exists $\epsilon_{00}\in (0,\epsilon_0]$ such that
\begin{align}
& x(1-x)  \left| v_{x x t}^{\epsilon}(t, x) \right|\leq C\big( \epsilon^{\frac{1}{3}} + \sqrt{\lambda} \,e^{- \alpha \lambda (T - t)} \big) \quad \textrm{for} \quad  (t,x,\epsilon)\in [0,T)\times (0,1) \times (0,\epsilon_0],  \label{DDv_ep_xt}\\
&| \uy^{\epsilon}_t(t) | , \, \left| \by^{\epsilon}_t(t) \right|  \leq C \quad \textrm{for} \quad  (t,\epsilon)\in [0,T- \epsilon^{\frac{1}{3}}]\times  (0,\epsilon_{00}].\label{Derivative_of_optimal_points_in_time} 
\end{align}
\end{lemma}
\begin{proof}
Using \eqref{F_expression} with $F=L_x^{\epsilon}$, we rewrite \eqref{DDv_with_optimizer} as
\begin{align}
  & v_{x x}^{\epsilon}(t, x) =  e^{- \lambda (T - t)} \EE \Big[ \tfrac{\partial^{2} Z_{T}^{(t, x)}}{\partial x^{2}} \Big] +\lambda \int_{t}^{T} e^{- \lambda (s-t)} \EE \left[ \tfrac{\partial^{2} Z_{s}^{(t, x)}}{\partial x^{2}} L^{\epsilon} ( s, Y_{s}^{(t, x)} ) \right] d s  \nonumber \\
  &\quad + \lambda \int_t^T  e^{-\lambda(s-t)}  \mathbb{E} \left[ \left( \tfrac{1}{x(1-x)}Z_{s}^{(t, x)}  \tfrac{\partial Y_{s}^{(t, x)}}{\partial x}   \left( \tfrac{B_{s} - B_{t}}{\sigma (s - t)} - 1 + x (1 + \gamma ) \right) + \tfrac{\partial Z_{s}^{(t, x)}}{\partial x}  \tfrac{\partial Y_{s}^{(t, x)}}{\partial x} \right)  L_x^{\epsilon}(s, Y_{s}^{(t, x)}) \right] ds \nonumber\\
&= e^{- \lambda (T - t)} \EE \Big[ \tfrac{\partial^{2} Z_{T-t}^{(0, x)}}{\partial x^{2}} \Big] +\lambda \int_{0}^{T-t} e^{- \lambda u} \, \EE \left[ \tfrac{\partial^{2} Z_{u}^{(0, x)}}{\partial x^{2}} L^{\epsilon} ( u+t, Y_{u}^{(0, x)} ) \right] d u  \nonumber \\
  &\quad + \lambda \int_0^{T-t}  e^{-\lambda u} \, \mathbb{E} \Big[ \Big( \tfrac{\frac{B_{u}}{\sigma u} - 1 + x (1 + \gamma )}{x(1-x)}Z_{u}^{(0, x)}  \tfrac{\partial Y_{u}^{(0, x)}}{\partial x} + \tfrac{\partial Z_{u}^{(0, x)}}{\partial x}  \tfrac{\partial Y_{u}^{(0, x)}}{\partial x} \Big)  L_x^{\epsilon}(u+t, Y_{u}^{(0, x)})  \Big] du, 
  \label{Another_form_of_DDv}
\end{align}
because $(Z_{s}^{(t,x)}, Y_s^{(t,x)}, B_s-B_t)$ and $(Z_{s-t}^{(0,x)}, Y_{s-t}^{(0,x)},B_{s-t})$ have the same distribution.
By the same way as in the proof of \lemref{Boundedness_of_various_expectation} (iii), we can check that
\begin{align}
\Big| \tfrac{\partial}{\partial t}  \EE \Big[ \tfrac{\partial^{2} Z_{T}^{(t, x)}}{\partial x^{2}} \Big]\Big| \leq C. \label{Zt_bound}
\end{align}
We observe that \eqref{L(T)} and the inequalities in \eqref{Lt_bound}, \eqref{Boundedness_of_Expectation_of_Multiplied_term} and \eqref{Zt_bound} produce
\begin{align}
&\left| \tfrac{\partial}{\partial t} \left( \lambda \int_{0}^{T-t} e^{- \lambda u} \, \EE \left[ \tfrac{\partial^{2} Z_{u}^{(0, x)}}{\partial x^{2}} L^{\epsilon} ( u+t, Y_{u}^{(0, x)} ) \right] d u+   e^{- \lambda (T - t)} \EE \Big[ \tfrac{\partial^{2} Z_{T-t}^{(0, x)}}{\partial x^{2}} \Big] \right)\right| \nonumber\\
&= \left| \lambda \int_{0}^{T-t} e^{- \lambda u} \, \EE \left[ \tfrac{\partial^{2} Z_{u}^{(0, x)}}{\partial x^{2}} L_t^{\epsilon} ( u+t, Y_{u}^{(0, x)} ) \right] d u + 
e^{- \lambda (T - t)}  \tfrac{\partial}{\partial t}  \EE \Big[ \tfrac{\partial^{2} Z_{T}^{(t, x)}}{\partial x^{2}} \Big] 
\right| \nonumber\\
&\leq C \Big( \lambda \int_{0}^{T-t} e^{- \lambda u} u \, du +e^{- \lambda (T - t)} \Big) \leq C \big( \epsilon^{\frac{2}{3}} + e^{- \lambda (T - t)} \big). \label{vxxt1}
\end{align}
The expression of $L_{xt}^{\epsilon}$ in \eqref{Lxt_expression} and the bound in \eqref{vxt_bound} imply 
\begin{align}
\left| L_{xt}^{\epsilon}(t,x)\right| \leq C\big( \epsilon^{\frac{2}{3}} + e^{- \alpha \lambda (T - t)} \big). \label{Lxt_bound2}
\end{align}
Using \eqref{L(T)} again, we obtain
\begin{align}
&\left| \tfrac{\partial}{\partial t} \left(  \lambda \int_0^{T-t}  e^{-\lambda u} \, \mathbb{E} \Big[ \Big( \tfrac{\frac{B_{u}}{\sigma u} - 1 + x (1 + \gamma )}{x(1-x)}Z_{u}^{(0, x)}  \tfrac{\partial Y_{u}^{(0, x)}}{\partial x}  + \tfrac{\partial Z_{u}^{(0, x)}}{\partial x}  \tfrac{\partial Y_{u}^{(0, x)}}{\partial x} \Big)  L_x^{\epsilon}(u+t, Y_{u}^{(0, x)})  \Big] du \right) \right| \nonumber\\
&=\left|   \lambda \int_0^{T-t}  e^{-\lambda u} \, \mathbb{E} \Big[ \Big( \tfrac{\frac{B_{u}}{\sigma u} - 1 + x (1 + \gamma )}{x(1-x)}Z_{u}^{(0, x)}  \tfrac{\partial Y_{u}^{(0, x)}}{\partial x}  + \tfrac{\partial Z_{u}^{(0, x)}}{\partial x}  \tfrac{\partial Y_{u}^{(0, x)}}{\partial x} \Big)  L_{xt}^{\epsilon}(u+t, Y_{u}^{(0, x)})  \Big] du \right| \nonumber\\
&\leq C \tfrac{1}{x(1-x)} \lambda \int_0^{T-t}  e^{-\lambda u} \big( 1+ \tfrac{1}{\sqrt{u}} \big)    \big( \epsilon^{\frac{2}{3}} + e^{- \alpha \lambda (T - t-u)} \big) du \leq C \tfrac{1}{x(1-x)}   \big( \epsilon^{\frac{1}{3}} + \sqrt{\lambda} \,e^{- \alpha \lambda (T - t)} \big), \label{vxxt2}
\end{align}
where the first inequality is due to \lemref{Boundedness_of_various_expectation} (i) and \eqref{Lxt_bound2} and the second inequality is due to \lemref{Order_of_exponential_times} and $\lambda=c \, \epsilon^{-\frac{2}{3}}$.
We combine \eqref{Another_form_of_DDv}, \eqref{vxxt1} and \eqref{vxxt2} to conclude \eqref{DDv_ep_xt}.

The bounds in \eqref{v_t_boundedness}, \eqref{vxt_bound},  \eqref{Upper_bound_of_DDv} and \eqref{Range_of_v} imply that there exists $\epsilon_{00}\in (0,\epsilon_0]$ such that for $(t,\epsilon)\in [0,T- \epsilon^{\frac{1}{3}}] \times (0,\epsilon_{00}]$,
\begin{align}
&\Big| \tfrac{ v_{x t}^{\epsilon}(t, \uy^{\epsilon}(t))}{1 - \gamma}- \tfrac{\epsilon  v_{t}^{\epsilon}(t, \uy^{\epsilon}(t))}{1 + \epsilon \uy^{\epsilon}(t)}\Big| \leq C  \epsilon^{\frac{2}{3}}, \quad  - \tfrac{ v_{x x}^{\epsilon}(t, \uy^{\epsilon}(t))}{1 - \gamma} - \tfrac{\gamma \epsilon^{2}  v^{\epsilon}(t, \uy^{\epsilon}(t))}{(1 + \epsilon \uy^{\epsilon}(t))^{2}}\geq C\epsilon^{\frac{2}{3}}.
\end{align}
The above inequalities and \eqref{Derivative_of_first_order_condition_in_time} produce the inequality for $\uy^{\epsilon}_t(t)$ in \eqref{Derivative_of_optimal_points_in_time}. The inequality for $\by^{\epsilon}_t(t)$ can be checked by the same way.
\end{proof}

Now we prove \lemref{v_xx_conv_lem}. 
By the mean value theorem and \eqref{Dv_0_lambda_order}, we have
$$\left| \tfrac{\lambda (v_{x x}^{SO, \lambda}(t, x_\epsilon)-v_{x x}^{SO, \lambda}(t, y_M))}{1 - \gamma}  \right| \leq C \left| x_\epsilon - y_M \right| \stackrel{\epsilon \downarrow 0}{\longrightarrow} 0,$$
where the convergence is due to \eqref{pre_y1}. The above inequality and \eqref{Limit_of_lambda_DDv} imply
\begin{align}
\lim_{\epsilon \downarrow 0} \tfrac{\lambda v_{x x}^{SO, \lambda}(t, x_\epsilon)}{1 - \gamma} =  -\gamma \sigma^2 v^{0}(t). \label{vxx_0_lim}
\end{align}
Direct computations using \eqref{AYZ} and the definition of $G^\epsilon$ produce
\begin{align}
&\lambda^2 \int_t^T  e^{-\lambda(s-t)} \mathbb{E} \Big[ Z_{s}^{(t, x)}\big(\tfrac{\partial Y_{s}^{(t, x)}}{\partial x}\big)^2 \tfrac{v^{\epsilon}_{xx} ( s, Y_{s}^{(t, x)} )  }{1-\gamma}  \cdot 1_{ \{ \uy^{\epsilon}(s)<Y_{s}^{(t, x)}<\by^{\epsilon}(s) \} } \Big] ds \nonumber\\
&=\tfrac{ \lambda}{x^{2} (1 - x)^{1 + \gamma}}  \int_t^T  e^{-\lambda(s-t)} \mathbb{E} \Big[ (Y_{s}^{(t, x)})^2 (1-Y_{s}^{(t, x)})^{1+\gamma} \tfrac{\lambda v^{\epsilon}_{xx} ( s, Y_{s}^{(t, x)} )}{1-\gamma}  \cdot 1_{ \{ \uy^{\epsilon}(s)<Y_{s}^{(t, x)}<\by^{\epsilon}(s) \} } \Big] ds \nonumber\\
&=\tfrac{ \lambda}{x^{2} (1 - x)^{1 + \gamma}}  \int_t^T  e^{-\lambda(s-t)} \int_{\uy^{\epsilon} \left( s \right)}^{\by^{\epsilon} \left( s \right)} G^{\epsilon} \left( s, y \right) \varphi(y;s - t, x) d y \, d s. \label{G1}
\end{align}
In \eqref{vxx_difference}, we apply \eqref{Lxx_expression} and follow the same procedure after \eqref{vxx_difference} to obtain
\begin{align*}
 \Big| \tfrac{ \lambda v_{x x}^{\epsilon}(t, x) -\lambda v_{x x}^{SO, \lambda}(t, x)}{1 - \gamma}  -\lambda^2 \int_t^T  e^{-\lambda(s-t)} \mathbb{E} \Big[ Z_{s}^{(t, x)}\big(\tfrac{\partial Y_{s}^{(t, x)}}{\partial x}\big)^2 \tfrac{v^{\epsilon}_{xx} ( s, Y_{s}^{(t, x)} )  }{1-\gamma} \cdot 1_{ \{ \uy^{\epsilon}(s)<Y_{s}^{(t, x)}<\by^{\epsilon}(s) \} }  \Big] ds \Big| \leq C \epsilon^{\frac{1}{3}}.
\end{align*}
We combine \eqref{vxx_0_lim}, \eqref{G1} and the above inequality to conclude that
\begin{align*}
G^{\epsilon} (t, x_\epsilon)  + y_M^2(1-y_M)^{1+\gamma} \gamma \sigma^2 v^{0}(t) - \lambda \int_{t}^{T} e^{- \lambda (s - t)} \int_{\uy^{\epsilon} \left( s \right)}^{\by^{\epsilon} \left( s \right)} G^{\epsilon} \left( s, y \right) \varphi(y;s - t, x_\epsilon) d y \, d s \stackrel{\epsilon \downarrow 0}{\longrightarrow} 0. 
\end{align*}
Therefore, to complete the proof, it is enough to prove the following:
\begin{align}
&\lambda \int_{t}^{T} e^{- \lambda (s - t)} \int_{\uy^{\epsilon} \left( s \right)}^{\by^{\epsilon} \left( s \right)} G^{\epsilon} \left( s, y \right) \varphi(y;s - t, x_\epsilon) d y \, d s \nonumber\\
&\quad - \int_{\uz^{\epsilon}(t)}^{\bz^{\epsilon}(t)} G^{\epsilon} (t,h(z)) \tfrac{\sqrt{2 \lambda} }{2 \sigma}\,e^{- \frac{\sqrt{2 \lambda}}{\sigma} \left\vert z - z_\epsilon \right\vert}  dz \stackrel{\epsilon \downarrow 0}{\longrightarrow}0. \label{G_goal}
\end{align}
By \eqref{pre_y1}, there exists $\epsilon_{00}'\in (0,\epsilon_{00}]$ such that $\tfrac{1}{x_\epsilon(1-x_\epsilon)}\leq C$ for $(t,\epsilon) \in[0, T-\tfrac{1}{\lambda}]\times (0,\epsilon_{00}']$.
Then the form of $\varphi$ in \eqref{density_form} implies
\begin{align}
0\leq \varphi(y;s-t, x_\epsilon) \leq  \tfrac{C}{\sqrt{s-t}}\quad \textrm{for} \quad (s,y,\epsilon)\in (t, T-\tfrac{1}{\lambda}]\times (0,1)\times (0,\epsilon_{00}'].
\label{density_bound1}
\end{align}
The mean value theorem, \eqref{Derivative_of_optimal_points_in_time}, \eqref{various_vs} and \eqref{density_bound1} imply that for $(s,\epsilon)\in (t, T-\epsilon^{\frac{1}{3}}) \times (0,\epsilon_{00}']$,
\begin{align}
&\bigg| \int_{\uy^{\epsilon}(s)}^{\by^{\epsilon}(s)} G^{\epsilon}(s, y) \varphi(y;s-t, x_\epsilon) d y   -  \int_{\uy^{\epsilon}(t)}^{\by^{\epsilon}(t)} G^{\epsilon}(s, y) \varphi(y;s-t, x_\epsilon) d y   \bigg| \leq C \sqrt{s-t}. \label{G2}
\end{align}
The mean value theorem and \eqref{DDv_ep_xt} imply that for $(s,\epsilon)\in (t, T-\epsilon^{\frac{1}{3}}) \times (0,\epsilon_{0}]$,
\begin{align}
&\bigg| \int_{\uy^{\epsilon}(t)}^{\by^{\epsilon}(t)} G^{\epsilon}(s, y) \varphi(y;s-t, x_\epsilon) d y -  \int_{\uy^{\epsilon}(t)}^{\by^{\epsilon}(t)} G^{\epsilon}(t, y) \varphi(y;s-t, x_\epsilon) d y  \bigg| \nonumber\\
&=\bigg|  \int_{\uy^{\epsilon}(t)}^{\by^{\epsilon}(t)} \lambda y (1-y)^\gamma \cdot \tfrac{y(1-y) \left( v_{x x}^{\epsilon}(s, y)- v_{x x}^{\epsilon}(t, y)\right) }{1 - \gamma} \varphi(y;s-t, x_\epsilon) d y  \bigg| \nonumber \\
&\leq C \lambda \big( \epsilon^{\frac{1}{3}} + \sqrt{\lambda} \,e^{- \alpha \lambda (T - s)} \big) (s-t) \cdot \PP\big( \uy^{\epsilon}(t)< Y_s^{(t,x_\epsilon)}< \by^{\epsilon}(t)\big) \nonumber\\
& \leq C\big( \sqrt{\lambda} + \lambda^{\frac{3}{2}} \,e^{- \alpha \lambda (T - s)} \big) (s-t). \label{G3}
\end{align}
We combine \eqref{G2} and \eqref{G3} to obtain
\begin{align}
&\bigg| \lambda \int_{t}^{T}  e^{- \lambda (s - t)} \left( \int_{\uy^{\epsilon}(s)}^{\by^{\epsilon}(s)} G^{\epsilon} \left( s, y \right) \varphi(y;s-t, x_\epsilon) d y - \int_{\uy^{\epsilon} \left( t \right)}^{\by^{\epsilon} \left( t \right)} G^{\epsilon} \left( t, y \right) \varphi(y;s-t, x_\epsilon) d y \right) d s \bigg|\nonumber \\
&\leq C \lambda \int_{t}^{T} e^{- \lambda (s - t)} \left(\sqrt{s-t} + \left(\sqrt{\lambda} + \lambda^{\frac{3}{2}} \,e^{- \alpha \lambda (T - s)}\right) (s-t)  \right)ds \stackrel{\epsilon\downarrow 0}{\longrightarrow} 0, \label{G4}
\end{align}
where the convergence can be checked using \lemref{Order_of_exponential_times} and $\lambda=c \, \epsilon^{-\frac{2}{3}}$.
Observe that 
\begin{align}
&\bigg| \lambda \int_{t}^{\infty}  e^{- \lambda (s - t)} \int_{\uy^{\epsilon}(t)}^{\by^{\epsilon}(t)} G^{\epsilon}(t, y) \bigg(\varphi(y;s-t, x_\epsilon) - \tfrac{1}{\sigma y (1 - y) \sqrt{2 \pi (s - t)}}  \exp \Big( - \tfrac{\big( \ln \big( \tfrac{y (1 - x_\epsilon)}{(1 - y) x_\epsilon} \big) \big)^{2}}{2 \sigma^{2} (s - t)}  \Big)\bigg) d y \, ds \bigg| \nonumber\\
&\leq C  \int_{0}^{\infty}  e^{-u} \int_{-\infty}^{\infty}  \bigg| \tfrac{\exp \big( - \frac{\left(z+ (\frac{\sigma^2}{2}-\mu) \sqrt{\frac{u}{\lambda}}  \right)^{2}}{2 \sigma^{2}} \big) - \exp \big( - \frac{z^2}{2 \sigma^{2}} \big)}{\sigma \sqrt{2 \pi}}  \bigg| d z \, ds  \stackrel{\epsilon\downarrow 0}{\longrightarrow} 0, \label{G5}
\end{align}
where the inequality is due to the boundedness of $G^\epsilon$ and the substitution $z=\tfrac{h^{-1}(y)- h^{-1}(x_\epsilon)}{\sqrt{s-t}}$ and $u=\lambda(s-t)$ and the convergence is due to the dominated convergence theorem. The boundedness of $G^{\epsilon}$ also implies
\begin{align}
\bigg| \lambda \int_T^{\infty}  e^{- \lambda (s - t)} \int_{\uy^{\epsilon}(t)}^{\by^{\epsilon}(t)} G^{\epsilon}(t, y) \varphi(y;s-t, x_\epsilon) dy \, ds \bigg|\leq C e^{-\lambda(T-t)} \stackrel{\epsilon\downarrow 0}{\longrightarrow} 0. \label{G6}
\end{align} 
We substitute $z=h^{-1}(y)$ and $u=\sqrt{\lambda(s-t)}$, then use Fubini's theorem below:
\begin{align}
&\lambda \int_{t}^{\infty}  e^{- \lambda (s - t)} \int_{\uy^{\epsilon}(t)}^{\by^{\epsilon}(t)} G^{\epsilon}(t, y)  \tfrac{1}{\sigma y (1 - y) \sqrt{2 \pi (s - t)}}  \exp \Big( - \tfrac{\big( \ln \big( \tfrac{y (1 - x_\epsilon)}{(1 - y) x_\epsilon} \big) \big)^{2}}{2 \sigma^{2} (s - t)}  \Big) d y \, ds\nonumber\\
&= \int_{\uz^{\epsilon}(t)}^{\bz^{\epsilon}(t)} G^{\epsilon}(t, h(z)) \int_{0}^{\infty}   \tfrac{\sqrt{2\lambda}}{\sigma  \sqrt{\pi}} \exp \left(-u^2 - \tfrac{ \lambda( z- z_\epsilon)^2}{2 \sigma^{2}u^2}\right)  d u \, dz \nonumber\\
&=\int_{\uz^{\epsilon}(t)}^{\bz^{\epsilon}(t)} G^{\epsilon} (t,h(z)) \tfrac{\sqrt{2 \lambda} }{2 \sigma}\,e^{- \frac{\sqrt{2 \lambda}}{\sigma} \left\vert z - z_\epsilon \right\vert}  dz, \label{G7} 
\end{align}
where the last equality is due to the observation that for $k>0$,
\begin{align*}
e^{-u^2-\frac{k}{u^2}} =\tfrac{d}{du}\Big(
\tfrac{e^{-2\sqrt{k}}}{2} \int_{\frac{\sqrt{k}}{u}-u}^\infty e^{-\zeta^2} d\zeta - \tfrac{e^{2\sqrt{k}}}{2} \int_{\frac{\sqrt{k}}{u}+u}^\infty e^{-\zeta^2} d\zeta
\Big)\quad \Longrightarrow \quad \int_0^\infty e^{-u^2-\frac{k}{u^2}} du = \tfrac{\sqrt{\pi}}{2}e^{-2\sqrt{k}}.
\end{align*}
Finally, we combine \eqref{G4}, \eqref{G5}, \eqref{G6} and \eqref{G7} to conclude \eqref{G_goal}.

\medskip

\section{Additional lemmas}

\begin{lemma}[Gronwall's lemma] \label{Gronwall_lemma}
Let $\alpha, \beta, f : [0, T] \to \mathbb{R}$ be measurable and $\beta\geq 0$. Assume that 
\begin{align*}
\int_0^T |f(t)| \beta(t) dt<\infty \quad \textrm{and} \quad  f(t) \leq \alpha(t) + \int_{t}^{T} \beta(s) f(s) d s \quad \textrm{for}\quad t \in [0, T].
\end{align*}
Then, $f$ satisfies the following inequality: for $t\in [0,T]$, 
\begin{align}
  f(t) \leq \alpha(t) + \int_{t}^{T} \alpha(s) \beta(s) e^{\int_{t}^{s} \beta(r) dr} d s. 
\end{align}
\end{lemma}

\begin{lemma}\label{meas_lem}
Let $F:[0,T]\times [0,1]^2\to \R$ be continuous. We define $f:[0,T]\times[0,1]\to [0,1]$ as
\begin{align}
f(t,x):=\max \Big\{ z: z\in \argmax_{y\in [0,1]} F(t,x,y)  \Big\},
\end{align}  
then $f$ is upper semicontinuous (which is obviously Borel-measurable). 
\end{lemma}
\begin{proof}
This is Lemma D.1 in \cite{GC23}.
\end{proof}

\begin{lemma}\label{Order_of_exponential_times}
There is a constant $C_\alpha$ independent of $t\in [0,T), \,\lambda\in [1,\infty)$ such that
\begin{align}
 \lambda \int_{t}^{T} e^{- \lambda (s - t)} (s - t)^{\alpha} d s 
\leq 
\begin{dcases}
C_\alpha \lambda^{- \alpha} \min \{ 1, \lambda (T - t) \} & \textrm{if} \quad \alpha \geq 0\\
C_\alpha \lambda^{- \alpha}& \textrm{if} \quad \alpha\in (-1, 0)\\
\end{dcases}
\end{align}
\end{lemma}

\begin{proof}
 Simple change of variable implies 
\begin{align*}
  \lambda \int_{t}^{T} e^{- \lambda (s - t)} (s - t)^{\alpha} d s = \lambda^{- \alpha} \int_{0}^{\lambda (T - t)} e^{- u} u^{\alpha} d u \leq  \begin{cases}
                                                        C_\alpha \lambda^{- \alpha} \min \{ 1 , \lambda (T - t) \} & \text{if } \alpha \geq 0, \\
                                                        C_\alpha \lambda^{- \alpha} & \text{if } \alpha \in (- 1, 0),
                                                         \end{cases}
\end{align*}
where we use the fact that $\int_{0}^{\infty} e^{- u} u^{\alpha} d u<\infty$ for $\alpha>-1$ and $e^{-u}u^\alpha\leq 1$ for $\alpha\geq 0$ and $u\leq 1$.
\end{proof}

\end{document}